%% file: main.tex
\newif\ifTR
\TRtrue

\PassOptionsToPackage{table}{xcolor}
\ifTR
\documentclass[acmsmall,screen,nonacm]{acmart}
\else
\documentclass[acmsmall,screen]{acmart}
\fi

\usepackage[utf8]{inputenc}

\usepackage{wrapfig}
\usepackage{paralist}
\usepackage{nth}
\usepackage{xspace}
\usepackage{multirow}
\usepackage{amsmath}
\usepackage{hyperref}
\usepackage{tikz}
\usepackage{todonotes}
\usepackage{mathtools}
\usepackage[linesnumbered,noend,ruled]{algorithm2e}
\usepackage[capitalize]{cleveref}
\usepackage{multicol}

\usepackage{tabularray}
\UseTblrLibrary{booktabs}

\usetikzlibrary{arrows,automata,calc, matrix, fit, arrows.meta}
\usetikzlibrary{decorations.pathreplacing,calligraphy,backgrounds}

\ifTR \else
\fi
\ifTR{}
\authorsaddresses{}
\fi

\usepackage[color,leftbars]{changebar}
\cbcolor{red}

\nochangebars

\pgfdeclarelayer{bg} 
\pgfdeclarelayer{bg1}
\pgfdeclarelayer{bg2}
\pgfdeclarelayer{bg3}
\pgfdeclarelayer{bg4}
\pgfdeclarelayer{bg5}
\pgfdeclarelayer{bg6}
\pgfsetlayers{bg6,bg5,bg4,bg3,bg2,bg1,bg,main}

\ifTR
\setcopyright{none}
\acmYear{}
\acmDOI{}
\acmVolume{}
\acmNumber{}
\acmArticle{}
\else
\setcopyright{cc}
\setcctype{by}
\acmDOI{10.1145/3729273}
\acmYear{2025}
\acmJournal{PACMPL}
\acmVolume{9}
\acmNumber{PLDI}
\acmArticle{169}
\acmMonth{6}
\received{2024-11-15}
\received[accepted]{2025-03-06}
\fi

\usepackage{booktabs}   %% For formal tables:
\usepackage{subcaption} %% For complex figures with subfigures/subcaptions
\input{macros.tex}
\input{macromove.tex}

\SetKwInput{KwInput}{Input}
\SetKw{Continue}{continue}

\crefname{claimcounter}{Claim}{Claims}
\crefname{section}{Sec.}{Secs.}

\begin{document}

% Styles used to draw runs in tag autamata
\tikzstyle{state} = [draw,circle,inner sep=2pt, scale=0.7, fill=white]
\tikzstyle{runState} = [rounded corners=3, fill=teal!30]
\tikzstyle{runArrows} = [line width=0.22mm]
\tikzstyle{tags} = [rounded corners=3, fill=gray!30, scale=0.6, text width=1cm, align=center]  % Text width might be overwritten locally
\tikzstyle{autLabel} = [scale=1.2]
\tikzstyle{autBoundingBox} = [draw, rounded corners=3, dashed]
% Styles for tag automaton example(s)
\tikzstyle{edge} = [line width=0.25mm, ->]
\tikzstyle{nonSamplingEdge} = [edge, bend right=25]
\tikzstyle{epsEdge} = [edge, dashed]
\tikzstyle{samplingEdge} = [edge, red!70]
\tikzstyle{samplingEdgeLabel} = [black, fill=red!40, text width=14mm]

% \pagestyle{plain}

%%%%%%%%%%%%%%%%%%%%%%%%%%%%%%%%%%%%%%%%%%%%%%%%%%%%%%
% \subtitle{Automata-Based Noodlification for Satisfiability of String Constraints with Lengths}

% \title[Short Title]{Full Title}         %% [Short Title] is optional;
%                                         %% when present, will be used in
%                                         %% header instead of Full Title.
% \titlenote{with title note}             %% \titlenote is optional;
%                                         %% can be repeated if necessary;
%                                         %% contents suppressed with 'anonymous'
% \subtitle{Subtitle}                     %% \subtitle is optional
% \subtitlenote{with subtitle note}       %% \subtitlenote is optional;
%                                         %% can be repeated if necessary;
%                                         %% contents suppressed with 'anonymous'

%% Author information
%% Contents and number of authors suppressed with 'anonymous'.
%% Each author should be introduced by \author, followed by
%% \authornote (optional), \orcid (optional), \affiliation, and
%% \email.
%% An author may have multiple affiliations and/or emails; repeat the
%% appropriate command.
%% Many elements are not rendered, but should be provided for metadata
%% extraction tools.

\title{A Uniform Framework for Handling Position Constraints in String Solving\ifTR{} (Technical Report) \fi}

\author{Yu-Fang Chen}
\orcid{0000-0003-2872-0336}
\affiliation{%
  \institution{Academia Sinica}
  \city{Taipei}
  \country{Taiwan}
}
\email{yfc@iis.sinica.edu.tw}

\author{Vojtěch Havlena}
\orcid{0000-0003-4375-7954}
\affiliation{%
  \institution{Brno University of Technology}
  \city{Brno}
  \country{Czech Republic}
}
\email{ihavlena@fit.vutbr.cz}

\author{Michal Hečko}
\orcid{0009-0003-2428-8547}
\affiliation{%
  \institution{Brno University of Technology}
  \city{Brno}
  \country{Czech Republic}
}
\email{ihecko@fit.vutbr.cz}

\author{Lukáš Holík}
\orcid{0000-0001-6957-1651}
\affiliation{%
  \institution{Brno University of Technology}
  \city{Brno}
  \country{Czech Republic}
}
\affiliation{%
  \institution{Aalborg University}
  \city{Aalborg}
  \country{Denmark}
}
\email{holik@fit.vutbr.cz}

\author{Ondřej Lengál}
\orcid{0000-0002-3038-5875}
\affiliation{%
  \institution{Brno University of Technology}
  \city{Brno}
  \country{Czech Republic}
}
\email{lengal@fit.vutbr.cz}

\begin{abstract}
We introduce a~novel decision procedure for solving the class of \emph{position
string constraints}, which includes string disequalities, $\notprefix$,
$\notsuffix$, $\strat$, and $\notstrat$.
These constraints are generated frequently in almost any application of string
constraint solving.
Our procedure avoids expensive encoding of the constraints to word equations
and, instead, reduces the problem to checking
conflicts on positions satisfying an integer constraint obtained
from the Parikh image of a~polynomial-sized finite automaton with a~special structure.
By the reduction to counting, solving position constraints becomes $\clNP$-complete
and for some cases even falls into~$\clP$.
This is much cheaper than the previously used techniques, which either used reductions generating word
equations and length constraints (for which modern string solvers use
exponential-space algorithms) or incomplete techniques.
Our method is relevant especially for automata-based string solvers,
\cbstart
which have recently achieved the best results in terms of practical efficiency,
\cbend
generality, and completeness guarantees.
This work allows them to excel also on position constraints, which used to be
their weakness.
Besides the efficiency gains, we show that our framework may be extended to
solve a~large fragment of $\notcontains$ (in $\clNEXPTIME$), for which decidability has been
long open, and gives a hope to solve the general problem. 
\cbstart
Our implementation of the technique within the \ziiinoodler solver significantly improves its performance on position constraints.
\cbend
%
%
%
%
% We show that a class of \emph{position string constraints}, which includes
% string disequalities, %not prefix, not suffix, (not) str at,
% $\notprefix$, $\notsuffix$, $\notcontains$, $\strat$, and $\notstrat$,
% can be solved cheaply by our novel decision procedure that amounts to checking
% (dis)equality of positions satisfying a linear counting constraint, obtained
% via a reduction to a variation of a Parikh/Cost Enriched automaton. 
% \ol{reduction to a finite automaton with a~specific structure and reasoning over its Parikh image}
% Such constraints are generated frequently in almost any application of string constraint solving. 
% By the reduction to counting, solving position constraints becomes NP-complete
% and falls even into P for some cases. 
% This is much cheaper that the previously used reductions that generate equations and length constraints,
% for which modern string solvers use exponential space algorithms.
% %
% Our methods are relevant especially for methods that use the automata approach to string constraint solving, 
% which recently generated the best results in terms of efficiency, generality, and also completeness guarantees.
% %
% This work allows it to excel also on the position constraints, which were before their weakness.  
% %
% Besides the efficiency gains, we show that our framework for reducing position
% constraints to counting may be extended to solve a large fragment of
% $\notcontains$, for which decidability has been long open, and gives a hope
% to solve the general problem. 
\end{abstract}



%% Keywords
%% comma separated list
\ifTR \else
  %% \keywords are mandatory in final camera-ready submission
\fi

\maketitle

%%%%%%%%%%%%%%%%%%%%%%%%%%%%%%%%%%%%%%%%%%%%%%%%%%%%%%%%%%%%%%%%%%%%%%%%%%%%%%%%
\vspace{-0.0mm}
\section{Introduction}\label{sec:introduction}
\vspace{-0.0mm}
%%%%%%%%%%%%%%%%%%%%%%%%%%%%%%%%%%%%%%%%%%%%%%%%%%%%%%%%%%%%%%%%%%%%%%%%%%%%%%%%

Solving string constraints (string solving) has been motivated initially 
by the analysis of string manipulations in programs, 
especially in preventing security risks such as cross-site scripting or SQL injection in web-applications.
\cbstart
In the last two decades,
the lively research community has managed to develop a number of string constraints solvers (overviewed in \cref{sec:related}) that may be capable of realizing this goal.
String constraints solvers are being integrated within
SMT solvers such as \ziii, \cvc{}4/5, or
Princess~\cite{Ruemmer2008,cvc4,cvc5,z3}, and the string
category has been introduced in the SMT competition~\cite{smtcomp24}.
% Two decades of development produced a number of string constraints solvers (overviewed in \cref{sec:related}) that may be capable of realizing this vision, which is mainly
% %(\cite{papersonstringsandwebapps}),
% thanks to the efforts of the lively research community, integration of the
% tools within SMT solvers such as \ziii, \cvc{}4/5, or
% Princess~\cite{Ruemmer2008,cvc4,cvc5,z3}, and introduction of the string
% category in the SMT competition~\cite{smtcomp24}.
\cbend
Since string solving is ubiquitous and string constraint logics are general,
string solving keeps finding also new applications,
such as analysing Simulink models~\cite{Ishii2022}, verifying UML models \cite{Jha2023},  
or checking cloud access policies at Amazon Web Services~\cite{Rungta2022}. 

The competition of string-solving methods is lively. While it was long dominated by the strongest industrial grade SMT-solvers \ziii and \cvc{}4/5, the leading positions have been recently taken by approaches based on finite automata. They include \ziiinoodler \cite{ChenCHHLS24}, 
a~recent winner of string categories of SMT-COMP~\cite{smtcomp24strings},
\ostrich~\cite{ChenFHHHKLRW22,ChenHLRW19}, which supports the richest palette of string constraints with strong completeness guarantees, \ziiistriiire \cite{Z3str3RE,BerzishDGKMMN23}, and loosely also one of the engines of \ziii~\cite{VeanesDerivatives21}.
Automata-based solvers excel especially in handling complex regular constraints with word equations and related constraints, such as transducer constraints or ReplaceAll.

%All string-solving tools, however, have weaknesses that would represent bottlenecks in demanding applications such as the analysis of string-manipulating programs. 
%
\paragraph{Position constraints.}
In this paper, we aim at remedying a weakness of automata-based techniques,
which is handling a class of constraints that we call \emph{position constraints}. 
The so-called \emph{existential} position constraints can be reduced to checking disequality 
of letters at two specific positions in strings (called a~\emph{mismatch}),
where the positions are determined through measuring lengths of certain
sub-strings---i.e., counting their letters and comparing the counts. 
The prime example of such constraints are string disequalities (negated word
equations) like $xyx \neq yxz$, which can be reduced to the existence of
mismatching positions in the two strings that are at the same distance from the
\cbstart
strings' start.\footnote{On a~high-level, we might write such a formula as $\exists i \in \nat\colon i \leq \max(\lenof{xyx}, \lenof{yxz}) \land (xyx)[i] \neq (yxz)[i]$ where $\len$ denotes the length of the corresponding string obtained by substituting into the variables and $t[i]$ denotes the $i$-th position in the string given by the term~$t$.}
\cbend
Other existential position constraints include $\notprefix$, $\notsuffix$,
$\strat$, and $\notstrat$.  
A special place among position constraints is occupied by $\notcontains(t_x,
t_y)$ where $t_x$ and $t_y$ are concatenations of variables, which does not
reduce to a~simple existence of mismatching positions (i.e., it is not an existential
constraints), but is equivalent to a~string formula with \emph{quantifier
alternation}, much harder to solve than existential formulae.
Informally, the formula says that there exists a~string assignment to variables
such that for the two strings~$w_x$ and~$w_y$ obtained by concatenating
assignments of the variables in~$t_x$ and~$t_y$, it holds that for all
alignments of the start of $w_x$ inside $w_y$, the letter at some position of
$w_x$ does not match the letter at the aligned position in $w_y$.
\cbstart
The formula falls into the $\forall\exists$ fragment of the string theory, which is undecidable~\cite{Marchenkov82,Durnev95};
% The theory of strings with universal quantification
% is undecidable~\cite{Quine46}, and 
the decidability of $\notcontains$ is a~known open problem~\cite{Notsubstring}.
\cbend
Solvers use heuristics that handle only very simple cases or rough approximations. 

Position constraints are practically relevant, for instance, a~disequality may be generated in symbolic execution
at every else-branch of a program that tests the equality of strings. 
Some solvers
%, especially those based on congruence closures such as CVC5 or Z3, 
may be able to guess the right solution for satisfiable position constraints with ease (CVC5 excels in this), 
but especially unsatisfiable position constraints may be much harder;
no current solver can handle them satisfactorily.

The automata-based approach particularly reduces position constraints into combinations of equations and length constraints. The work~\cite{ChenCHHLS23} even shows a version of this reduction that allows one to freely add disequalities to one of the largest known decidable fragments of basic constraints (word equations, regular membership, and string length constraints), the \emph{chain-free fragment} \cite{ChainFree}, 
while preserving decidability. 
However, since equations are expensive, 
the price of the reduction may be very high.
Indeed, essentially all string solvers use exponential-space algorithms to deal with word equations,
including the automata-based ones, and the problem is opaque even in theory: 
equations with regular memberships are in $\clPSPACE$~\cite{Plandowski99,Jez16},
but decidability of their combination with length
constraints is a long-standing open problem.

%, and solving them is at least PSPACE complete. 

\paragraph{The gist of our approach to solving position constraints.}
In contrast, the procedure we propose in this paper starts only \emph{after}
\cbstart
the rest of the constraint is solved---transformed into the \emph{monadic
decomposition}~\cite{HagueLRW20,ChenHLRW19,VeanesBNB17}, i.e., a~formula
obtained by transforming word equations into regular
constraints---and then solves
\cbend
position constraints quickly and efficiently by other means. 
%runs in NP for existential position constraints, and even in P for a single one. 
%
\newcommand{\problem}{MP\xspace}
The specific technical problem we are solving is therefore \emph{satisfiability of Monadic-Position constraints}: 
\begin{itemize}
\item[\bf (\problem)]
{\bfseries
Satisfiability of a quantifier-free conjunction of a monadic constraint (a conjunction of
\cbstart
regular membership constraints and linear integer arithmetic (LIA)
constraints over string lengths, a.k.a. length constraints) and a conjunction of
position constraints with string terms as parameters (e.g., $xyx\neq yxy$,
$\notcontains(xyx,yxy)$). }
\cbend
\end{itemize}
Monadic decomposition is at the heart of how automata-based solvers, such as \ziiinoodler and \ostrich, solve string constraints. 
The fact that position constraints can be ignored in this process is the main distinguishing feature of our approach and the key to its efficiency. 
Translating the entire \problem to a monadic constraint first, as automata-based solvers currently do,
\cbstart
takes \emph{exponential} space,
\cbend
while our algorithm runs in \emph{$\clNP$ for existential position constraints}
and in \emph{$\clP$ for a single~one}. 

Moreover, our framework also allows to make a step towards showing decidability of \problem with $\notcontains$. 
Namely, it allows to translate $\notcontains(t,t')$ combined with a~monadic constraint
to LIA with a nested universal quantifier
when the languages constraining variables in terms~$t$ and~$t'$ are \emph{flat}
(expressible as a~concatenation of a fixed number of parts, where every part is
an iteration of a~single word).
The resulting quantified LIA formula can be solved by an off-the-shelf solver
(SMT-solvers seem to be capable of solving the obtained formulae efficiently).
To the best of our knowledge,
our algorithm is the first one for exact reasoning with $\notcontains$ that is
complete for a~large and interesting fragment of the problem.
%, and we believe that it gives a hope of devising an algorithm for the general case.   

\paragraph{The technical basis of our approach.}
Technically, given a set of position constraints and automata-represented regular membership constraints in the monadic decomposition, our procedure derives a~LIA constraint that relates positions of mismatches and lengths of the strings assigned to variables. 
Repetition of variables, e.g., $xyx\neq yxy$ or $\notcontains(xyx,yxy)$, the main limiting factor of decidable fragments, is handled too:
The length of every variable is extracted from a single run of its automaton, and its contribution to a mismatch position in the string is counted with the multiplicity equal to the number of the occurrences of the variable that precede the mismatch position.
For a single existential position constraint, the LIA formula is constructed by taking the formula for the Parikh image of an automaton with a~specific structure, enriched with constraints that enforce seeing a~mismatch of symbols at at the proper positions.

% an extension of the classical construction of the Parikh image of an automaton extended with annotations of the mismatch positions.  
%

The construction is more complex with multiple position constraints that share
variables. The straightforward approach would be to enumerate an exponential
number of cases corresponding to all orders of mismatches in the position
constraints.
E.g., for the constraint $D_a \land D_b \land D_c$, we would need to consider orders like
$(D_a^1, D_a^2, D_b^1, D_b^2, D_c^1, D_c^2)$,
$(D_a^1, D_b^1, D_c^1, D_a^2, D_c^2, D_b^2)$, etc., where $D_x^i$ denotes the $i$-th mismatch in~$D_x$.
The generated LIA formula would then be exponential (more precisely in
$2^{\Theta(n \log n)}$) to the number of position constraints.
Our $\clNP$ algorithm for \problem depends on the discovery of an equivalent polynomial encoding.
The encoding combines the Parikh image of a~polynomial-size automaton that generates any number of mismatch positions in any order with an additional arithmetic constraint that rules out the ``wrong'' orderings.  

% The construction is more complex with multiple position constraints that share variables. The naive approach would be to enumerate an exponential number of cases corresponding to all possible placements of the mismatched positions of all position constraints relative to one another and relative to the borders of the variables in the string terms. 
% For instance, for the disequality $xyz\neq uvw$, the mismatch position on the left can occur in each of the variables $x,y,z$, 
% and the mismatch position on the right in any of $u,v,w$, hence we need to consider $3 \times 3 = 9$ cases. 
% %
% If another disequality over the same variables was present, the numbers of cases would multiply, and the relative ordering of the mismatch positions occurring in the same variable would have to be distinguished too.  
% The generated LIA formula would then be exponential to the number of position constraints. 
% %
% Our NP algorithm for \problem depends on the discovery of an equivalent polynomial encoding. It combines the Parikh image of a polynomial-size automaton that generates any number of mismatch positions in any order with an additional arithmetic constraint that rules out the "wrong" orderings.  

Moreover, we show that \problem with a \emph{single} existential position constraint is in $\clP$ by a reduction to 0-reachability in a~one-counter automaton~\cite{AlurC11,melskireps1997}. 

\vspace{-2mm}
\paragraph{Practical impact on performance.}
Our experiments show that our framework substantially improves the speed and
effectiveness of the automata approach to string constraint solving whenever
position constraints are involved.
On position-heavy cases from our benchmark (obtained from symbolic executions
of large software projects), the string solver \ziiinoodler extended with our
decision procedure for position constraints is the fastest of all solvers and
\cbstart
has less timeouts than \cvcv, the only other solver that handles
\cbend
these benchmarks well.
The performance of \cvcv and the modified \ziiinoodler is orthogonal,
reflecting the large difference between the two approaches, and together they
solve all but 10 out of the $\sim$150,000 benchmarks.
\ziiinoodler also demonstrates its ability to solve hard instances of
$\notcontains$ on an artificial benchmark made to test solvers on constraints
involving $\notcontains$, where all other solvers fail.

\vspace{-2mm}
\paragraph{Summary of contributions.}
Our contributions can be summarised as follows:  
\begin{enumerate}

\item 
We propose a procedure for handling position constraints in the automata-based approach efficiently, without the need of including them into a~procedure for monadic decomposition. 
%that makes it possible to handle them efficiently after the monadic decomposition of the other constraints is achieved. 
%
\item
We show that \problem is in $\clNP$ for existential position constraints and
in $\clP$ for a single one.
This contrasts with the exponential cost of monadic decomposition. 

\item
We propose the first algorithm for exact reasoning with $\notcontains$ that is guaranteed to work on a large and interesting fragment of the problem,
namely, the \problem where the regular languages restricting variables within the terms in the arguments of $\notcontains$ are \emph{flat}, and show that the problem is in $\clNEXPTIME$.
\item
Experimental results show that our techniques significantly improve the performance of one of the fastest string solvers, \ziiinoodler, on position constraints. 

\end{enumerate}

\cbstart
%------------------------------------------------------------------------------
\paragraph{Context of our work.}
Our focus is on the basic constraints that must be handled by a universal
practical string solver: combinations of word equations, regular constraints,
and length constraints (as witnessed, e.g., by the SMT-COMP
benchmark~\cite{smtlib2024}).
Besides the \clPSPACE-completeness of classical decidability of equations with
regular constraints by Makanin, Plandowski, and Je\.{z}~\cite{makanin,Plandowski99,Jez16}, decidability of the
full basic logic is a long-standing open problem, and it is open even when
equations are quadratic, where only two occurrences of every variable are
allowed~\cite{LinQuadra21}.
The known undecidability results are concerned with extensions of this logic
with other than basic constraints (here called extended constraints), and are
thus only marginally relevant to our current work.
These fragments concern unrestricted ReplaceAll, transducer-defined relations,
string-integer conversions, and other extended
constraints~\cite{AnthonyReplaceAll2018,ChenFHHHKLRW22,ChenHHHLRW20,ChenHLRW19,AnthonyTowards2016,day-str-dec-undec-18,DayGGKM24}.
The approach that started at~\cite{AutomataSplitting} led to a discovery of the
largest decidable fragment of the basic constraints, the chain-free
fragment~\cite{ChainFree}.
It generalizes the earlier acyclic fragment of~\cite{AnthonyTowards2016} and the larger
straight-line fragment of~\cite{AutomataSplitting}.
The other solvers, especially \ziii and \cvcv, which use different approaches
usually revolving around the congruence closure~\cite{MouraB07}, can handle
similar class of constraints in practice, but do not come with strong
theoretical guarantees.
Limited extensions of straight-line or chain-free logics with extended
constraints were shown decidable
in~\cite{AnthonyReplaceAll2018,ChenFHHHKLRW22,ChenHHHLRW20,ChenHLRW19,AnthonyTowards2016,day-str-dec-undec-18,DayGGKM24}.
The automata-based approach transforms equations and regular constraints to
a~monadic decomposition (a~disjunction of conjunctions of regular constraints
of at most doubly exponential size), which is in turn transformed to a~LIA
formula and solved using a~LIA solver.
The decidable fragments are all based on prohibiting forms of ``\emph{cyclic
dependencies}'' of string variables in word equations, and range from
\clPSPACE-complete to \cltwoEXPSPACE, depending on the allowed extended
constraints.
The position constraints that we discuss here are in this framework normally
solved by a~reduction to the basic constraints.
Essentially, after obtaining doubly exponential monadic decomposition of the
other basic constraints, the original approach needs to run the doubly
exponential space procedure to solve the position constraints.
Our work allows to replace the second \cltwoEXPSPACE phase by an
\clNP-algorithm (or \clNEXPTIME for our fragment of $\notcontains$), or even by
a~\clP one in the simplest case of one disequality.
\cbend

\newcommand{\figSemantics}[0]{
\begin{figure}[t]
\begin{center}
\begin{minipage}{10cm}
\begin{align*}
  \assgn \models {}& \stringvar \in \lang & {}\liff {} & \assgn(\stringvar) \in \lang, \\
  \assgn \models {}& \intterm \leq \intterm' & {}\liff {} & \assgn(\intterm) \leq \assgn(\intterm'), \\
  \assgn \models {}& \stringterm = \stringterm' & {}\liff {} & \assgn(\stringterm) = \assgn(\stringterm'), \\
  \assgn \models {}& \stringvar = \stratof{\stringterm}{\intterm} &{} \liff {} & \begin{cases}
     \assgn(\stringvar) = w_s[\assgn(\intterm)] \land \assgn(\stringterm) = w_s
 &\text{if } \assgn(\intterm) < |w_s|, \\
    \assgn(\stringvar) = \epsilon  &\text{otherwise},
  \end{cases}\\
  \assgn \models {}& \prefixof{\stringterm}{\stringterm'} &{} \liff {} & \exists z_p \in \alphabet^*\colon \assgn(\stringterm) \concat z_p = \assgn(\stringterm'), \\
  \assgn \models {}& \suffixof{\stringterm}{\stringterm'} &{} \liff {} & \exists z_s \in \alphabet^*\colon z_s \concat \assgn(\stringterm) = \assgn(\stringterm'), \text{and}\\
  \assgn \models {}& \containsof{\stringterm}{\stringterm'} &{} \liff {} & \exists z_c,z_c' \in \alphabet^*\colon z_c \concat \assgn(\stringterm)\concat z_c' = \assgn(\stringterm').
\end{align*}
\end{minipage}
\end{center}
\caption{Semantics of atomic predicates for a~variable assignment~$\assgn$}
\label{fig:semantics}
\end{figure}
}

%%%%%%%%%%%%%%%%%%%%%%%%%%%%%%%%%%%%%%%%%%%%%%%%%%%%%%%%%%%%%%%%%%%%%%%%%%%%%%%%
\vspace{-2.0mm}
\section{Preliminaries}\label{sec:preliminaries}
\vspace{-0.0mm}
%%%%%%%%%%%%%%%%%%%%%%%%%%%%%%%%%%%%%%%%%%%%%%%%%%%%%%%%%%%%%%%%%%%%%%%%%%%%%%%%

We fix a~finite non-empty \emph{alphabet}~$\alphabet$.
A~\emph{word} or \emph{string} over~$\alphabet$ is a~(finite) sequence of symbols $w = a_0 \ldots a_{n-1}$
from~$\alphabet$, its \emph{length}~$|w|$ is~$n$, and the symbol at index~$0 \leq
i < n$ is denoted as~$w[i] = a_i$.
We use~$\epsilon$ to denote the \emph{empty word} and $\concat$ denotes string concatenation
(we sometimes omit the operator, i.e., $a \concat b = ab$).
$\nat$ denotes the set of natural numbers $\{0, 1, \ldots\}$.

A~\emph{nondeterministic finite automaton} (NFA) is a~tuple
$\aut = (Q, \Delta, I, F)$ where
$Q$~is a~finite set of \emph{states},
$\Delta \subseteq Q \times \alphabet \times Q$ is a~\emph{transition relation}
with transitions denoted as $q \ltr a p$ for $q,p \in Q$ and $a\in\alphabet$,
and~$I, F \subseteq Q$ are sets of \emph{initial} and \emph{final} states respectively.
A~\emph{run}~$\run$ of~$\aut$ over a~word~$w = a_1 \ldots a_n \in \alphabet^*$ is
a~sequence of transitions $q_0 \ltr{a_1} q_1 \ltr{a_2} \cdots
\ltr{a_n} q_n$ s.t.~$q_0 \in I$ and $\forall 1 \leq i \leq n\colon q_{i-1}
\ltr{a_i} q_i \in \Delta$.
The run~$\run$ is \emph{accepting} if $q_n \in F$, and the language of~$\aut$
is the set $\langof \aut = \{w \in \alphabet^* \mid \textrm{there exists an
accepting run of } \aut \textrm{ over } w\}$.
A~language $\lang \subseteq \alphabet^*$ is \emph{regular} iff there exists an
NFA~$\aut$ such that $\lang = \langof \aut$.
We sometimes define regular languages using regular expressions with the
standard textbook notation.

Given a~set~$U$, let~$\numof U = \{\numof u \mid u \in U\}$ denote the set of elements obtained from~$U$'s elements by prepending them with~$\num$.
The \emph{Parikh image} of a~run~$\run$, denoted as $\parikhofrun \run$, is
a~mapping $\parikhofrun \run\colon \numof\Delta \to \nat$
such that $\parikhofrun{\run}(\numof t)$
denotes the number of occurrences of transition~$t$ in~$\run$.
We say that~$\aut$ is \emph{flat}\footnote{We note that our notion of
\emph{flatness} differs from the one from~\cite{AbdullaACDHRR17} and is similar
to the one from~\cite{LerouxS05}.}
if for every two runs~$\run_1$ and~$\run_2$ it holds that if
$\parikhofrun{\run_1} = \parikhofrun{\run_2}$, then $\run_1 = \run_2$.
Structurally, this means that flat automata have the form of \emph{directed
acyclic graphs} (DAGs) connecting \emph{simple} (i.e., non-nested) \emph{loops}.
We say that a~regular language~$\lang$ is flat iff there exists a~flat NFA~$\aut$
s.t.~$\lang = \langof \aut$.
For instance, the language $(ab)^*c((ab)^* + (ba)^*)$ is flat, while the
language $(a+b)^*$ is not flat.
% \lh{use the \# variables already here, so that it becomes easier later when we talk bout parikh for tag automata?}

Let~$\vars$ and $\intvars$ be the sets of \emph{string} and \emph{integer variables}. String formulae are of the form:
%
%Let~$\vars$ be the set of \emph{string variables} and~$\intvars$ be the set of \emph{integer variables}.
%We consider formulae~$\varphi$ of the following form:
%
\begin{align*}
\varphi ::={} ~~& \varphi \land \varphi \mid
                  \varphi \lor \varphi \mid
                  \neg \varphi \mid
                  \formatom \\
\formatom ::={} ~~& \stringvar \in \lang \mid
                  \intterm \leq \intterm \mid
                  \stringterm = \stringterm \mid 
                  \stringvar = \stratof{\stringterm}{\intterm} \mid
                  \prefixof{\stringterm}{\stringterm} \mid
                  \suffixof{\stringterm}{\stringterm} \mid
                  \containsof{\stringterm}{\stringterm} \\
\stringterm ::={}~~&\stringvar \mid
                    \stringterm \concat \stringterm \\
\intterm ::={}~~& \intvar \mid
                  k \mid
                  \lenof{\stringvar} \mid
                  \intterm + \intterm
\end{align*}
where
$\formatom$ is an atomic formula,
$\lang$ is a~regular language (given by a~regular expression or an NFA),
$\stringterm$~is a~string term consisting of a concatenation of string
variables~$\stringvar \in \vars$\footnote{A~string literal $u \in \alphabet^*$ can
be encoded by a~new string variable~$x_u$ and a~regular constraint $u \in L_u$
for $L_u = \{u\}$.}, and
$\intterm$ is an integer term composed of sums of integer variables~$\intvar$,
integers $k \in \integers$, and lengths of string variables~$\lenof{\stringvar}$.

\figSemantics   %%%%%%%%%%%%%%%%

The semantics of formulae is defined as follows.
A~(variable) \emph{assignment} is a~mapping $\assgn\colon(\vars \to \alphabet^*)
\cup (\intvars \to \integers)$ and we lift~$\assgn$ to string
terms~$\stringterm$ and integer terms~$\intterm$ in the usual way ($\stringterm
\concat \stringterm$ as string concatenation and $\intterm + \intterm$ as integer
addition), with $\assgn(\lenof{\stringvar})$ being interpreted as the length of
the string~$\assgn(\stringvar)$.
Semantics of atomic predicates is given in~\cref{fig:semantics}.

When used within the~DPLL(T) framework, it is sufficient to consider only conjunctions of atomic formulae and their negations.
%We introduce a~normal form of such a~formula.
A \emph{normal form} of such formula is
$\eqconstr \land \regconstr \land \intconstr \land \diseqconstr$
%\begin{equation}
%\eqconstr \land \regconstr \land \intconstr \land \diseqconstr
%\end{equation}
%
where
\begin{itemize}
  \item  $\eqconstr$ is a~conjunction of word equations $\stringterm = \stringterm$,
  \item  $\regconstr$ is a~conjunction of regular memberships
    $\stringvar \in \lang$\footnote{A~regular \emph{non-membership} constraint can be
    translated into a~\emph{membership} constraint for a~complement language.}
    such that for every $\stringvar \in \vars$, there is exactly one regular
    membership constraint in~$\regconstr$ (and we use $\langof{\stringvar}$ to denote it),
  \item  $\intconstr$ is a~conjunction of integer constraints $\intterm \leq
    \intterm$ where $\intterm$'s do not contain $\lenof{\stringterm}$ terms, and
  \item $\diseqconstr$ is a~conjunction of position constraints of the following form:
    \begin{align*}
      \stringterm \neq \stringterm \mid
      \intvar = \lenof{\stringvar} \mid
      \stringvar = \stratof{\stringterm}{\intterm} \mid {}&
      \stringvar \neq \stratof{\stringterm}{\intterm} \mid {} \\
      \neg\prefixof{\stringterm}{\stringterm} \mid
      \neg\suffixof{\stringterm}{\stringterm} \mid {}&
      \neg\containsof{\stringterm}{\stringterm}
    \end{align*}
\end{itemize}
%
%The transformation of a~conjunction of literals in the original language into
%the normal form can be implemented in the following way:
%
We transform a conjunction of literals into
the normal form in the following way:
\begin{inparaenum}[(i)]
  \item  We substitute every non-negated occurrence of the predicates
    $\prefixof{u_p}{v_p}$, $\suffixof{u_s}{v_s}$, and $\containsof{u_c}{v_c}$
    with the word equations $v_p = u_p z_p$, $v_s = z_s u_s$, and $v_s = z_c
    u_c z'_c$ respectively for fresh string variables $z_p$, $z_s$, $z_c$, and
    $z'_c$ (note that we cannot do a~similar thing for negated occurrences,
    since we would need to introduce universal quantifiers for the
    $z$-variables).
  \item  For every string variable~$x$, we compute a~single NFA that represents
    all regular membership constraints for~$x$.
\end{inparaenum}
\cbstart
We will use $|\regconstr|$ to denote the sum of numbers of states of all NFAs
used for encoding the~$\regconstr$ constraint.
\cbend
%
% We say that a~formula~$\varphi$ is
% \emph{stable}~\cite{BlahoudekCCHHLS23,ChenCHHLS23} iff it is in the normal form
% $\eqconstr \land \regconstr \land \intconstr \land \diseqconstr$ and for every
% word equation $x_{i_1} \ldots x_{i_m} = x_{j_1} \ldots x_{j_n}$
% from~$\eqconstr$, it holds that $\langof{x_{i_1}} \cdots \langof{x_{i_m}} =
% \langof{x_{j_1}} \cdots \langof{x_{j_n}}$.
% We say that~$\varphi$ is \emph{cooked}
% % \yfc{The footnote reads a bit vague to me (using somewhat similar). Should be made more precise.}
% iff it is stable and none of the
% variables occurring in~$\eqconstr$ occurs in~$\intconstr \land
% \diseqconstr$.

%%%%%%%%%%%%%%%%%%%%%%%%%%%%%%%%%%%%%%%%%%%%%%%%%%%%%%%%%%%%%%%%%%%%%%%%%%%%%%%%
\vspace{-0.0mm}
\section{Overview}\label{sec:label}
\vspace{-0.0mm}
%%%%%%%%%%%%%%%%%%%%%%%%%%%%%%%%%%%%%%%%%%%%%%%%%%%%%%%%%%%%%%%%%%%%%%%%%%%%%%%%

Given a~formula in the normal form $\eqconstr \land \regconstr \land
\intconstr \land \diseqconstr$ in the considered fragment, the main idea of our
approach is the following (focusing on the~$\diseqconstr$ part).
Other automata-based approaches usually try to get rid of the
predicates in~$\diseqconstr$ by transforming them into word equations and
length constraints (e.g.~\cite{ChenCHHLS23,HavlenaHLS24,ChenHHHLRW20}), thus making their
word equations potentially much harder to process.
Handling word equations is the 
most demanding task in string solving, as the best known algorithms
for dealing with word equations work in
$\clPSPACE$~\cite{Plandowski99,Jez16} and are not practical\footnote{Existing string solvers usually deal with word equations by
incomplete algorithms that do not guarantee termination.} (and in the presence of length
constraints, the decidability of the problem is currently unknown).

%
%
% \ol{}
% \footnote{
% Transforming a~formula into the \emph{cooked} form is related to the
% so-called \emph{monadic
% decomposition}~\cite{HagueLRW20,ChenHLRW19,VeanesBNB17}, which, in the context
% of word equations and regular constraints, corresponds to eliminating word
% equations and obtaining a~finite union of regular constraints of the form
% $\stringvar \in \lang$.
% The cooked form is, however, less restrictive, as there are still some word
% equations allowed (with the condition of being stable w.r.t.\ the regular
% constraints).}.

In our approach,
we take care of the constraints
in~$\diseqconstr$ only after the word equations in~$\eqconstr$ have been
processed and the obtained constraint $\regconstr' \land \intconstr' \land
\diseqconstr'$ contains no more word equations.
This is achieved by the stabilization-based procedure introduced in~\cite{ChenCHHLS23}, which transforms $\eqconstr \land \regconstr \land \intconstr$ 
into a~disjunction of constraints $\regconstr' \land \intconstr'$ extended with 
a substitution mapping variables from the original constraints to a concatenation of 
fresh variables occurring in~$\regconstr'$ and $\intconstr'$.
The transformation comes with the additional property that the resulting
constraint is a~\emph{monadic
decomposition}~\cite{HagueLRW20,ChenHLRW19,VeanesBNB17}, which means that each
choice of the fresh variable assignment given by $\regconstr'$ forms a solution
of the original system of equations (the substitution defines how to obtain
values of variables occurring in~$\eqconstr$ from the fresh variables)\footnote{Strictly speaking, we only need monadic decomposition on the variables that occur in position constraints.}.
Using the substitution 
map, we can substitute variables in $\diseqconstr$ in order to obtain a~position 
constraint $\regconstr' \wedge \intconstr' \wedge \diseqconstr'$. 
Therefore, we will now focus on solving formulae of the form $\regconstr' \wedge \intconstr' \wedge \diseqconstr'$, which is the main contribution of this paper.

% (e.g., by the stabilization-based procedure in~\cite{ChenCHHLS23}) to obtain a~cooked formula
% $\eqconstr' \land \regconstr' \land \intconstr' \land \diseqconstr'$
% (we note that~$\intconstr$ and~$\diseqconstr$ stay the same).
% Since variables in~$\eqconstr'$ do not have occurrences in $\intconstr \land
% \diseqconstr$, we can resolve~$\eqconstr'$ (together with the regular
% constraints from~$\regconstr'$ assigning languages to the variables occurring
% in~$\eqconstr'$) independently, e.g.~by the procedure
% from~\cite{BlahoudekCCHHLS23}.
% Therefore, we will now focus on solving formulae of the form $\regconstr' \land
% \intconstr \land \diseqconstr$, which is the main contribution of this paper.

The main idea of solving a~formula of the form $\regconstr' \land \intconstr'
\land \diseqconstr'$ is by transforming $\regconstr' \land
\diseqconstr'$ into a~LIA formula that is then added to~$\intconstr'$ to obtain
a~(potentially quantified) LIA constraint~$\intconstr''$, which can be solved by
an off-the-shelf LIA solver.
The procedure for transforming $\regconstr' \land \diseqconstr'$ into a~LIA
constraint is based on constructing a~\emph{tag automaton}~$\tagAut$, which is
an NFA whose
transitions are extended with \emph{tags}.
The tags do not affect the run of~$\tagAut$, but are used for counting
``\emph{positions}'' where something interesting happens, e.g., for a~string
disequality $x \neq y$, we count the position~$\ell$ of a~single character
mismatch $x[\ell] \neq y[\ell]$.
$\tagAut$ is constructed from the regular constraints in~$\regconstr'$ based on
the atomic predicates in~$\diseqconstr'$.

%%%%%%%%%%%%%%%%%%%%%%%%%%%%%%%%%%%%%%%%%%%%%%%%%%%%%%%%%%%%%%%%%%%%%%%%%%%%%%%%
\vspace{-0.0mm}
\section{Tag Automaton}\label{sec:tag_automaton}
\vspace{-0.0mm}
%%%%%%%%%%%%%%%%%%%%%%%%%%%%%%%%%%%%%%%%%%%%%%%%%%%%%%%%%%%%%%%%%%%%%%%%%%%%%%%%

% \lh{If tag automaton is syntactically only an NFA? over a strange alphabet, then we should define it as such}
In this section, we define \emph{tag automata} (TAs) used in the later sections
for encoding position constraints.
We note that TAs are used just to simplify
notation; one could build the framework on top of NFAs or some counter model,
such as Parikh automata~\cite{KlaedtkeR03}, cost-enriched finite
automata~\cite{ChenHHHLRW20}, or simply vector addition systems with
states~\cite{HopcroftP79},
\cbstart
for the price of a~more cumbersome notation.
\cbend

Let~$\tags$ be a~set of \emph{tags}.
A \emph{tag automaton} (TA) over~$\tags$ is a quadruple $T = (Q, \Delta, I, F)$, where $Q,
I, F$ are as for an NFA and the set of transitions $\Delta$ is $\Delta \subseteq Q \times 2^\tags \times Q$.
We use $\move q S p$ to denote a~transition $(q,S,p)$;
we drop duplicate braces, so, e.g., for $S = \{a,b\}$, we write $\move q {a,b}
p$.
A~run~$\run$ of~$T$ is a~sequence of transitions
$\move{q_0}{S_1}{q_1}\move{}{S_2}{} \cdots \move{}{S_n}{q_n}$ such that
$q_0 \in I$ and 
$\move{q_{i-1}}{S_i}{q_i} \in \Delta$ for all $1 \leq i \leq n$.
$\run$~is accepting if~$q_n \in F$.
The Parikh image of~$\run$, $\parikhofrun \run$, is defined in the same way as for~NFAs.
We may write  $Q(T)$, $\Delta(T)$, $I(T)$, and~$F(T)$ to refer
to the components of a~TA~$T$.

For an NFA $\aut = (Q,\Delta,I,F)$ and a~variable~$x$, we define the tag automaton
$\lentag_x(\aut) = (Q, \Delta', I, F)$ over a~set of tags~$\{ \tagsymof a \mid
a \in \alphabet\} \cup \{\taglenof x \}$
where $\Delta'= \{ \move q {\tagsymof a, \taglenof{x}} r \mid q \ltr{a} r \in \Delta  \}$.
The used tags denote the \textbf{S}ymbol and \textbf{L}ength (i.e., we will use
the number of occurrences of the $\taglen$ tag to derive the length of a~word
from the TA).
Given two TAs~$\aut = (Q_\aut, \Delta_\aut, I_\aut, F_\aut)$ and~$\but =
(Q_\but, \Delta_\but, I_\but, F_\but)$ with disjoint sets of states, their
$\epsilon$-concatenation is the TA~$\aut \concateps \but = (Q_\aut \cup
Q_\but, \Delta_\aut \cup \Delta_\but \cup \Delta_\epsilon, I_\aut, F_\but)$
with $\Delta_\epsilon = \{\move q {} r \mid q\in F_\aut, r \in I_\but\}$.

The \emph{Parikh formula} of a~TA~$T$, denoted as
$\parikhformof{T}$ is a~\emph{linear integer arithmetic}
(LIA) formula with free variables
$\numof \delta$ corresponding to numbers of occurrences of transitions $\delta
\in \Delta$.
Models of $\parikhformof{T}$ are, therefore, assignments
$\sigma\colon \numof \Delta \to \nat$ such that
\begin{equation}
\sigma \models \parikhformof{T}
\quad\text{iff there is an accepting run~$\run$ of~$T$
  s.t.~$\parikhofrun{\run} = \sigma$}
\end{equation}
Constructing~$\parikhformof{T}$ can be done in the usual
way~\cite{JankuT19} (cf. \cref{app:parikh}).
We will also work with the \emph{Parikh tag formula} $\parikhformtagof{T}$,
which is a~formula whose models are numbers of each tag seen on an accepting run
in~$T$, constructed as
\begin{equation}
\parikhformtagof{T} \quad\defiff\quad
  % \exists \numof{\delta_1}, \ldots, \numof{\delta_n} \colon
  \parikhformof{T} \land
  \bigwedge_{t \in \tags}
    \numof t = \textstyle\sum\{\numof \delta \in \numof\Delta\mid \delta = \move
    q S r \in \Delta, t \in S\}.
\end{equation}
Note that in $\parikhformtagof{T}$, the free variables
are now also the counts of tags from~$\tags$ and a~model is an
assignment $\sigma'\colon (\numof{\Delta} \cup \numof{\tags}) \to \nat$.
% \lh{if we define PIR already with \# variables, then things get easier here, right?}

% \vh{for the notcontains case we need to get rid of the existential quantification of transitions. We need to compare runs not just observed symbols.}

% \vh{for the notcontains case we use more Parikh formulae in a single formula. We need to somehow distinguish between them. So far for $\numof{a}$ you don't know which 
% Parikh formula this number refers to.}

%%%%%%%%%%%%%%%%%%%%%%%%%%%%%%%%%%%%%%%%%%%%%%%%%%%%%%%%%%%%%%%%%%%%%%%%%%%%%%%%
\vspace{-0.0mm}
\section{Solving Disequalities}\label{sec:diseq}
\vspace{-0.0mm}
%%%%%%%%%%%%%%%%%%%%%%%%%%%%%%%%%%%%%%%%%%%%%%%%%%%%%%%%%%%%%%%%%%%%%%%%%%%%%%%%

In this section, we will show how to solve a~formula $\regconstr' \land
\intconstr \land \diseqconstr$ where~$\diseqconstr$ only contains disequalities.
We will start from the simplest case (a single disequality with two different
variables), proceed to an arbitrary single disequality, and finish with a~system
of disequalities.

%%%%%%%%%%%%%%%%%%%%%%%%%%%%%%%%%%%%%%%%%%%%%%%%%%%%%%%%%%%%%%%%%%%%%%%%%%%%%%%%
\vspace{-0.0mm}
\subsection{\textrm{I}: A~Single Disequality of Two Variables}\label{sec:simpleDiseq}
\vspace{-0.0mm}
%%%%%%%%%%%%%%%%%%%%%%%%%%%%%%%%%%%%%%%%%%%%%%%%%%%%%%%%%%%%%%%%%%%%%%%%%%%%%%%%

First, we consider the simplest case, i.e., when the position
constraint~$\diseqconstr$ contains a~single disequality
\begin{equation}
  x \neq y,
\end{equation}
where~$x$ and~$y$ are two different string variables whose values are restricted
to regular languages~$\lang_x$ and~$\lang_y$.
The constraint is satisfiable iff there exist strings~$w_x \in \lang_x$ and~$w_y
\in \lang_y$ such that they are either 
\begin{inparaenum}[(i)]
  \item  of a~different length or
  \item  they are of the same length~$\ell$ and there exists a~position~$0\leq
    p < \ell$ such that $w_x[p] \neq w_y[p]$.
\end{inparaenum}

%\yfc{This paragraph is a bit hard to read to me. When I see three connected copies, I was not certain where are the three copies from  Maybe you should make it explicit that we make three copies and connect them by certain transitions for status marking. It gets clear only after I read the formal definition.}
%\ol{is it better now?} \yfc{yes, thanks}
We will show how to construct a~tag automaton and, from it, a~LIA formula~$\varphi^\singlesimplediseq$
that is satisfiable iff the disequality is satisfiable.
We assume that we are given NFAs~$\aut_x =(Q_x, \Delta_x, I_x, F_x)$
and~$\aut_y = (Q_y, \Delta_y, I_y, F_y)$ such that~$\langof{\aut_x} = \lang_x$
and $\langof{\aut_x} = \lang_y$ with $Q_x \cap Q_y = \emptyset$.
For this, we construct a~TA~$\aut^{\singlesimplediseq}$.
Intuitively, $\aut^{\singlesimplediseq}$~is obtained by first concatenating~$\lentag_x(\aut_x)$ with~$\lentag_y(\aut_y)$ using an
$\epsilon$-transition into a~TA~$\autcon$.
One can see~$\autcon$ as an encoding of all possible models of~$x$ and~$y$ w.r.t.\
only regular constraints\footnote{In fact, the order in which we do the concatenation
does not really matter---the main objective is to obtain a~TA such that an
accepting run in it represents a~model of regular constraints}.
Then, we take three copies of~$\autcon$ and connect them together with transitions 
that represent detected mismatches:
% \lh{we should explain what the role of concatenation is: no role, order does not matter, we could work with an unordered mmultiset of automata too, concatenating just makes our notation easier}
the first copy is used for tracking the run of~$\aut_x$ before the position of
the mismatch in~$x$ is encountered,
the second copy is used for tracking the rest of the run in~$\aut_x$ and the
first part of the run in~$\aut_y$ (before the mismatch in~$\aut_y$), and
the last copy tracks the rest of the run in~$\aut_y$.
Moreover, the automaton is enhanced with tags that keep track of the position
of the mismatch in~$x$ and in~$y$ and the values of the mismatched symbols.

\newcommand{\figSingleSimple}[0]{
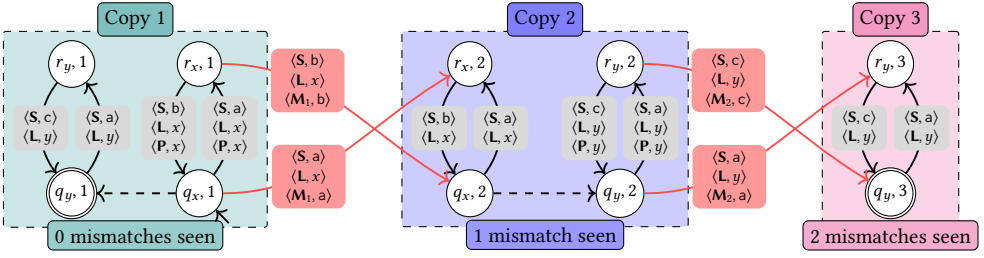
\begin{figure}
  \scalebox{1.0}{
    \input{figs/single-diseq-tag-aut}
  }
  \vspace*{-3mm}
    \caption{Example of a tag automaton for the disequality $x \neq y$ with $\langof{\aut_x} = (ab)^*$ and $\langof{\aut_y} = (ac)^*$.
    States~$q_x, r_x$ belong to~$\aut_x$, states~$q_y, r_y$ belong to~$\aut_y$.
    }
    \label{fig:tagautSimple}
    \vspace{-3mm}
\end{figure}
}

%%%%%%%%%%%%%%%%%%%%%%%%%%%%%%%%%%%%%%%%%%%%%%%%%%%%%%%%%%%%%%%%%%%%%%%%%%%%%%%%
\subsubsection{Tag Automaton Construction}
%%%%%%%%%%%%%%%%%%%%%%%%%%%%%%%%%%%%%%%%%%%%%%%%%%%%%%%%%%%%%%%%%%%%%%%%%%%%%%%%

Formally, let $\autcon = (Q_{\concat}, \Delta_{\concat}, I_{\concat}, F_{\concat})$ be a~TA over
$\tags_{\concat} = \{\tagsymof{a} \mid a\in \alphabet \} \cup \{ \taglenof{x},
\taglenof{y} \}$ obtained by the
$\epsilon$-concatenation of $\lentag_x(\aut_x)$ and $\lentag_y(\aut_y)$, i.e.,
$A_{\concat} = \lentag_x(\aut_x) \concateps \lentag_y(\aut_y)$.
 The tags $\tagsym$ are used for tracking the currently read symbol and 
$\taglen$-tags are used for counting of the length of a word from the language of the corresponding variable.
Then $\aut^{\singlesimplediseq} = (Q_1 \cup Q_2 \cup Q_3, \Delta, I, F)$ is a~TA
over $\tags^{\singlesimplediseq} = \tags_{\concat} \cup \{\tagmisoneof a,
\tagmistwoof a \mid a \in \alphabet \} \cup \{\tagposof{x}, \tagposof y\}$,  where
the~$\tagmisone$ and~$\tagmistwo$
tags denote the first and the second $\tagmis$ismatch respectively and $\tagposof x,
\tagposof y$ are used to count the $\tagpos$ositions of the mismatch in~$x$ and~$y$. 
$\aut^{\singlesimplediseq}$ is constructed as follows:

\begin{figure}
  \scalebox{1.0}{
    \input{figs/single-diseq-tag-aut}
  }
  \vspace*{-3mm}
    \caption{Example of a tag automaton for the disequality $x \neq y$ with $\langof{\aut_x} = (ab)^*$ and $\langof{\aut_y} = (ac)^*$.
    States~$q_x, r_x$ belong to~$\aut_x$, states~$q_y, r_y$ belong to~$\aut_y$.
    }
    \label{fig:tagautSimple}
    \vspace{-3mm}
\end{figure}
   %%%%%%%%%%%%%%%%%%
%
\begin{itemize}
  \item  $Q_1 = Q_{\concat} \times \{1\}$,
    $Q_2 = Q_{\concat} \times \{2\}$, and
    $Q_3 = Q_{\concat} \times \{3\}$;
    intuitively, $Q_1$ are states where no mismatch was seen,
    $Q_2$ are states where only the first mismatch symbol was seen, and
    $Q_3$ are states where both mismatch symbols were seen,

  \item  $I = I_{\concat} \times \{1\}$,
  \item  $F = F_{\concat} \times \{1,3\}$, and
  \item  $\Delta$ is the union of the following sets of transitions:
    \begin{itemize}
      \item  $\{\move {(q,1)} {\tagsymof a, \tagposof x, \taglenof{x}} {(r,
        1)} \mid \move q {\tagsymof a, \taglenof{x}} r \in \Delta_{\concat}\}$
        --- transitions in~$\aut_x$ before the first mismatch,
      \item  $\{\move {(q,1)} {\tagsymof a, \taglenof{y}} {(r,
        1)} \mid \move q {\tagsymof a, \taglenof{y}} r \in \Delta_{\concat}\}$
        --- transitions in~$\aut_y$ if no mismatch symbols are seen and the disequality is satisfied due to $x$ and $y$ having different lengths,
      \item  $\{\move{(q,1)}{\tagsymof a, \tagmisoneof a, \taglenof{x}}{(r, 2)}
        \mid \move q {\tagsymof a, \taglenof{x}} r \in \Delta_{\concat}\}$
        --- the first mismatch (in~$\aut_x$),
      \item  $\{\move {(q,2)} {\tagsymof a, \taglenof{x}} {(r, 2)} \mid \move
        q {\tagsymof a, \taglenof{x}} r \in \Delta_{\concat}\}$
        --- transitions in~$\aut_x$ after the first mismatch (we still need to
        finish reading~$x$ to make sure that it was accepted
        by~$\aut_x$),
      \item  $\{\move {(q,2)}{}{(r,2)} \mid \move q {} r \in \Delta_\circ\}$ ---
        jump from~$A_x$ to~$A_y$,
      \item  $\{\move {(q,2)} {\tagsymof a, \tagposof y, \taglenof{y}} {(r,
        2)} \mid \move q {\tagsymof a, \taglenof{y}} r \in \Delta_{\concat}\}$
        --- transitions in~$\aut_y$ before the second mismatch,
      \item  $\{\move{(q,2)}{\tagsymof a, \tagmistwoof a, \taglenof{y}}{(r, 3)}
        \mid \move q {\tagsymof a, \taglenof{y}} r \in \Delta_{\concat}\}$
        --- the second mismatch (in~$\aut_y$),
      \item  $\{\move {(q,3)} {\tagsymof a, \taglenof{y}} {(r, 3)} \mid \move q
        {\tagsymof a, \taglenof{y}} r \in \Delta_{\concat}\}$
        --- transitions in~$\aut_y$ after the second mismatch.
    \end{itemize}
\end{itemize}
% \lh{if the Sa tag is not used on most transitions, I would just remove it. Else one starts thinking bout it why it is there, finds no reason, and thinks he must have misunderstood something}
% \lh{also, my feeling is that having P1 and P2 instead of P and L would make somehow easiers to understand. We split the word into two parts plus the mismatch, and these parts are marked by the tags, easy to explain.}
%
Note that in~$\Delta$, the $\tagmisone$-tagged transitions denote the occurrence of the
first mismatch (which causes a~jump from~$Q_1$ to~$Q_2$) and the $\tagmistwo$-tagged
transitions denote the occurrence of the second mismatch (jumping
from~$Q_2$ to~$Q_3$).
For accepting runs of~$\aut^{\singlesimplediseq}$, it holds that they either
\begin{inparaenum}[(i)]
  \item  stay in~$Q_1$ and accept in some state from~$F_\circ \times \{1\}$ (so
    we only keep track of the lengths of the words $w_x \in L_x$ and $w_y \in
    L_y$) or
  \item  take a~$\tagmisone$-tagged transition to~$Q_2$ and then
    a~$\tagmistwo$-tagged transition to~$Q_3$ and accept in some state
    from~$F_\circ \times \{3\}$.
\end{inparaenum}
An accepting run of the tag automaton encodes an assignment of~$x$ and~$y$ to
words from~$L_x$ and~$L_y$. An example of a~constructed tag automaton 
is given in \cref{fig:tagautSimple}.

%%%%%%%%%%%%%%%%%%%%%%%%%%%%%%%%%%%%%%%%%%%%%%%%%%%%%%%%%%%%%%%%%%%%%%%%%%%%%%%%
\subsubsection{Formula Construction}
%%%%%%%%%%%%%%%%%%%%%%%%%%%%%%%%%%%%%%%%%%%%%%%%%%%%%%%%%%%%%%%%%%%%%%%%%%%%%%%%

After~$\aut^{\singlesimplediseq}$ is constructed, it remains to test whether
there is a~run of~$\aut^{\singlesimplediseq}$ starting in~$I$ and ending
in~$F$ such that the number of occurrences of $\taglenof x$ and $\taglenof y$ differs 
(corresponding to the case $|x| \neq |y|$), or the number of occurrences of the~$\tagposof x$ and~$\tagposof
y$ tags is the same and there is one occurrence of a~$\tagmisoneof a$ tag and
one occurrence of a~$\tagmistwoof b$ tag with $a \neq b$.
The means to this is via the Parikh tag formula
of~$\aut^{\singlesimplediseq}$.
First, we define formulae~$\formdiseqsym^{\singlesimplediseq}$
and~$\formdiseqmis^{\singlesimplediseq}$, which express
that the two sampled symbols are a~mismatch and that there was a~mismatch:
\begin{equation}
  \formdiseqsym^{\singlesimplediseq} \quad\defiff\quad
  \bigwedge_{a \in \alphabet} \left(\numof{\tagmisoneof a} + \numof{\tagmistwoof
  a} < 2\right)
  \qquad\text{and}\qquad
  \formdiseqmis^\singlesimplediseq \quad\defiff\quad
  \sum_{a \in \alphabet} \numof{\tagmisoneof a} > 0.
\end{equation}
In the formula, the first sum is used to check that the mismatched symbols are
indeed different (from the construction of $\aut^{\singlesimplediseq}$,
there will be at most one~$\tagmisone$ and one~$\tagmistwo$ tags in every
accepting run) and the second sum makes sure that there was at least one
mismatch (so that we can only accept in~$Q_3$).
Finally, we construct the formula $\varphi^{\singlesimplediseq}$ equisatisfiable to
the disequality $x \neq y$ as follows:
\begin{equation}
\varphi^{\singlesimplediseq} \quad\defiff\quad \parikhformtagof{\aut^{\singlesimplediseq}} \land
  \left(\numof{\taglenof x} \neq \numof{\taglenof y} \lor 
  \left(\numof{\tagposof x} = \numof{\tagposof y} \land
  \formdiseqsym^{\singlesimplediseq} \land
  \formdiseqmis^{\singlesimplediseq}
  \right) \right).
\end{equation}

\cbstart
\begin{theorem}\label{thm:}
The formula $\regconstr' \land \intconstr \land x \neq y$ is equisatisfiable to the formula
$\intconstr \land \varphi^{\singlesimplediseq}$.
Moreover, the size of~$\varphi^{\singlesimplediseq}$ is polynomial to~$|\regconstr'|$.
\end{theorem}
\cbend

%%%%%%%%%%%%%%%%%%%%%%%%%%%%%%%%%%%%%%%%%%%%%%%%%%%%%%%%%%%%%%%%%%%%%%%%%%%%%%%%
\vspace{-0.0mm}
\subsection{\textrm{II}: A Single Unrestricted Disequality }\label{sec:single-gen-diseq}
\vspace{-0.0mm}
%%%%%%%%%%%%%%%%%%%%%%%%%%%%%%%%%%%%%%%%%%%%%%%%%%%%%%%%%%%%%%%%%%%%%%%%%%%%%%%%

Let us now move to the case of an arbitrary disequality between concatenations
of (potentially repeating) variables:
\vspace{-3mm}
\begin{equation}
x_1 \ldots x_n \neq y_1 \ldots y_m.
\vspace{-1mm}
\end{equation}
This complex disequality is satisfiable if there are words $w_{x_i} \in
\lang_{x_i}$ and $w_{y_j} \in \lang_{y_j}$ for all~$i,j$
such that either
\begin{inparaenum}[(i)]
  \item both sides have different lengths (given by $\sum_{1\leq
    i\leq n}|w_{x_i}|$ and $\sum_{1 \leq j \leq m}|w_{y_j}|$ respectively) or
  \item they are of the same length and there is a \emph{mismatch} position
    $\ell$ s.t.\ $w_x[\ell] \neq w_y[\ell]$ where $w_x = w_{x_1}\cdots w_{x_n}$
    and $w_y = w_{y_1} \cdots w_{y_m}$.
\end{inparaenum}
We emphasize that there might be multiple occurrences of a~single variable~$z$,
potentially on both sides of the disequality, and they all need to be assigned
the same word from~$\lang_z$.

We will again construct a tag automaton~$\aut^{\singlediseq}$ checking whether
one of the conditions to satisfy the disequality holds.
In this case, an accepting run in the tag automaton encodes an assignment
that maps every variable~$z$ from the disequality to a~word from~$\lang_z$.
The mismatch may happen in any pair of occurrences of variables
$(x_i, y_j)$ and, moreover, the variables might have multiple occurrences in the
disequality.
In the tag automaton, a~run encoding a~mismatch needs to nondeterministically guess 
a~pair of variables' occurrences where the mismatch happens, the mismatch
positions within the variables, and the mismatch symbol itself.
In order to check that the guess is valid, we then construct a~formula
that will use the Parikh tag image of~$\aut^{\singlediseq}$ and use it to check that
\begin{inparaenum}[(i)]
  \item the mismatch symbols are different and
  \item the \emph{global} positions of both mismatches are equal, meaning that
    for a~guess of mismatch variables $(x_i, y_j)$,
    the mismatch position in $x_i$ plus lengths of assignments of $x_1 \cdots x_{i-1}$ is equal to the 
    mismatch position in $y_j$ plus lengths of assignments of $y_1 \cdots y_{j-1}$.
\end{inparaenum}

%%%%%%%%%%%%%%%%%%%%%%%%%%%%%%%%%%%%%%%%%%%%%%%%%%%%%%%
\newcommand{\figTagAut}[0]{
\begin{figure}
  \scalebox{1.0}{
    \input{figs/single-diseq-tag-aut-gen}
  }
    \caption {Example of a tag automaton for the disequality $xy \neq yx$ with $\langof{\aut_x} = (ab)^*$ and $\langof{\aut_y} = (ac)^*$.}
    \label{fig:tagaut}
\vspace{-3mm}
\end{figure}
}

%%%%%%%%%%%%%%%%%%%%%%%%%%%%%%%%%%%%%%%%%%%%%%%%%%%%%%%%%%%%%%%%%%%%%%%%%%%%%%%%
\subsubsection{Tag Automaton Construction}
%%%%%%%%%%%%%%%%%%%%%%%%%%%%%%%%%%%%%%%%%%%%%%%%%%%%%%%%%%%%%%%%%%%%%%%%%%%%%%%%

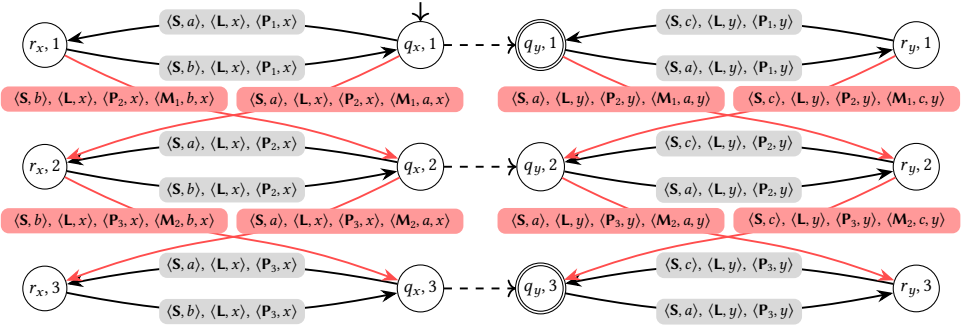
\begin{figure}
  \scalebox{1.0}{
    \input{figs/single-diseq-tag-aut-gen}
  }
    \caption {Example of a tag automaton for the disequality $xy \neq yx$ with $\langof{\aut_x} = (ab)^*$ and $\langof{\aut_y} = (ac)^*$.}
    \label{fig:tagaut}
\vspace{-3mm}
\end{figure}
   %%%%%%%%%%%%%%%%%%%%%%

Similarly to the previous section, we assume an NFA~$\aut_x$ for each
variable~$x$ describing the language~$\lang_x$.
We use~$\vars$ to denote the set of all variables in the disequality.
Without loss of generality, we assume that the sets of states of~$A_x$'s are
pairwise disjoint.
We also assume a~fixed linear order on variables $\preccurlyeq$, which is further used to create
a~unique concatenation of tag automata 
for each variable.
First, for each variable $x$ we construct the TA $T_x$ corresponding
to~$\aut_x$ enriched with lengths, i.e., $T_x = \lentag_x(\aut_x)$. Then, we construct 
$\autcon = (Q_{\concat}, \Delta_{\concat}, I_{\concat}, F_{\concat})$
over~$\tags_{\concat}$ as an~$\epsilon$-concatenation of all TAs $T_x$ for
$x\in\vars$ in the order given by~$\preccurlyeq$.

$\aut^{\singlediseq} = (Q_1 \cup Q_2 \cup Q_3, \Delta, I, F)$ is a~TA
over $\tags^{\singlediseq} = \tags_{\concat} \cup \{\tagmisoneof{a, x},
\tagmistwoof{a, x} \mid a \in \alphabet, x \in \vars\} \cup \{ \tagposoneof{x}, \tagpostwoof{x}, \tagposthreeof{x} \mid x\in\vars \}$,  where
the~$\tagmisone$ and~$\tagmistwo$
tags again denote the first and the second mismatch respectively (note that,
contrary to \cref{sec:simpleDiseq}, the mismatch tags here are extended with variables).
The tags $\tagposoneof z$ and $\tagpostwoof z$ are used to count the local 
positions of the first and second mismatch in~$z$ respectively\footnote{We need to consider two
possible mismatches in one variable because an occurrence of a~mismatch can be between
two positions in an~assignment for~$z$, one for the left-hand side and one for the
right-hand side. E.g., consider the disequality $xy \neq yx$ and the assignment
$\{x\mapsto ab,y\mapsto a\}$; the first mismatch is between the~$b$ in~$x$ on
the left-hand side and the~$a$ in~$x$ on the right-hand side.}.
The $\tagposthreeof{x}$ tag will become
important in \cref{sec:notprefix} when
reusing the automaton construction for the $\neg\suffix$ predicate.
$\aut^{\singlediseq}$ is constructed as follows:
\begin{itemize}
  \item  $Q_1 = Q_{\concat} \times \{1\}$,
    $Q_2 = Q_{\concat} \times \{2\}$, and
    $Q_3 = Q_{\concat} \times \{3\}$,
  \item  $I = I_{\concat} \times \{1\}$,
  \item  $F = F_{\concat} \times \{1,3\}$, and

  \item  $\Delta$ is the union of the following sets of transitions:
    \begin{itemize}
      \item  $\{\move {(q,1)} {\tagsymof a, \tagposoneof z, \taglenof{z}} {(r,
        1)} \mid \move q {\tagsymof a, \taglenof{z}} r \in \Delta_{\concat}\}$
        --- transitions in each $\aut_z$ before the first mismatch,
      \item  $\{\move{(q,1)}{\tagsymof a, \tagmisoneof{a,z}, \tagpostwoof z, \taglenof{z}}{(r, 2)}
        \mid \move q {\tagsymof a, \taglenof{z}} r \in \Delta_{\concat}\}$
        --- the first mismatch,
      % \item  $\{\move {(q,\leftsymb)} {\tagsymof a, \taglenof{x}} {(r, \leftsymb)} \mid \move
      %   q {\tagsymof a, \taglenof{x}} r \in \Delta_{\concat}, q,r \in Q_x\}$
      %   --- transitions in~$\aut_x$ after the first mismatch (we still need to
      %   finish reading the word for~$x$ to know whether it was accepted
      %   by~$\aut_x$),
      \item  $\{\move {(q,2)} {\tagsymof a, \tagpostwoof z, \taglenof{z}} {(r,
        2)} \mid \move q {\tagsymof a, \taglenof{z}} r \in \Delta_{\concat}\}$
        --- transitions in each $\aut_z$ before the second mismatch,
      \item  $\{\move{(q,2)}{\tagsymof a, \tagmistwoof{a,z}, \tagposthreeof{z}, \taglenof{z}}{(r, 3)}
        \mid \move q {\tagsymof a, \taglenof{z}} r \in \Delta_{\concat}\}$
        --- the second mismatch,
      \item  $\{\move {(q,3)} {\tagsymof a, \taglenof{z}, \tagposthreeof{z}} {(r,3)} \mid \move q
        {\tagsymof a, \taglenof{z}} r \in \Delta_{\concat}\}$
        --- transitions in each $\aut_z$ after the~second mismatch, and
      \item  $\{\move {(q,i)} {} {(r,i)} \mid \move q {} r \in \Delta_{\concat},
        1 \le i \le 3 \}$ --- transitions connecting
        variables on level~$i$.
    \end{itemize}
\end{itemize}
% An example of a tag automaton for a disequality $xy \neq xy$ is given in \cref{fig:tagaut}.

% \begin{figure}
%     \centering
%     \input{figs/run-single-diseq.tex}
%     \caption{
%         Example of a run satisfying $xx \neq y z$ with a model $\{x \mapsto \texttt{a\underline{b}},
%         y \mapsto \texttt{ab}, z \mapsto \texttt{a\underline{c}} \}$ where letters forming a~conflict
%         are underlined.
%         \ol{kick out? or modify for the previous example of tag aut?}
%     }
% \end{figure}

% \lh{I would again be for removing the Sa tag if not used and using P1, P2, P3 instad of P1, P2, and L. It will again give only one tag per transition and a simple explanation.}

%%%%%%%%%%%%%%%%%%%%%%%%%%%%%%%%%%%%%%%%%%%%%%%%%%%%%%%%%%%%%%%%%%%%%%%%%%%%%%%%
\subsubsection{Formula Construction}
%%%%%%%%%%%%%%%%%%%%%%%%%%%%%%%%%%%%%%%%%%%%%%%%%%%%%%%%%%%%%%%%%%%%%%%%%%%%%%%%

For the satisfiability checking of the general disequality, we generalize the LIA reduction 
from the previous section. As in the previous case, the LIA formula speaks about properties
of $\aut^{\singlediseq}$ using the Parikh tag formula $\parikhformtagof{\aut^{\singlediseq}}$.

First, the formula expressing that lengths of both sides are different can be
defined as follows:
\vspace{-1mm}
\begin{equation}
  \formdiseqlen^{\singlediseq} \quad\defiff\quad
    \sum_{1 \leq i \leq n} \numof{\taglenof{x_i}}
    \neq
    \sum_{1 \leq j \leq m} \numof{\taglenof{y_j}}.
\end{equation}
\vspace{-4mm}

Next, for the case the lengths are the same but there is a~mismatch, we begin by
defining a~formula that checks that the particular mismatch symbols are
different (and that there is at least one mismatch) by generalizing the
formula~$\formdiseqsym^{\singlesimplediseq}$ from the previous section:
\vspace{-1mm}
\begin{equation}
  \formdiseqsym^{\singlediseq} \quad\defiff\quad
  \bigwedge_{a \in \alphabet} \left(\sum_{x \in \vars} \left(\numof{\tagmisoneof
  {x,a}} + \numof{\tagmistwoof {x,a}}\right) < 2\right)
  % \land 
  % \sum_{\substack{a \in \alphabet\\x \in \vars}} \numof{\tagmisoneof {x,a}} > 0.
\end{equation}
\vspace{-2mm}

In order to check whether the global mismatch positions on both sides are equal,
we need to make a~case split ranging over all pairs~$(x_i, y_j)$ of occurrences of variables
from the left-hand side and the right-hand of the disequality.
For each such a~pair, we define an~auxiliary formula~$\formdiseqposij$ comparing
global mismatch positions when the mismatch is between the two occurrences.

\begin{enumerate}
  \item  If $x_i$ and $y_j$ are occurrences of a~different variable then:
    \begin{itemize}
      \item  if $x_i \prec y_j$,
        \vspace{-2mm}
        \begin{equation}
          \formdiseqposij \quad\defiff\quad
            \numof {\tagposoneof{x_i}} +
            \sum_{1 \leq u < i} \numof{\taglenof{x_u}}
            =
            \numof{\tagpostwoof{y_j}} +
            \sum_{1 \leq v < j} \numof{\taglenof{y_v}},
            %\quad \text{and}
        \end{equation}
        \vspace{-5mm}
      \item  if $x_i \succ y_j$,
        \vspace{-2mm}
        \begin{equation}
          \formdiseqposij \quad\defiff\quad
            \numof {\tagpostwoof{x_i}} +
            \sum_{1 \leq u < i} \numof{\taglenof{x_u}}
            =
            \numof{\tagposoneof{y_j}} +
            \sum_{1 \leq v < j} \numof{\taglenof{y_v}}.
        \end{equation}
        \vspace{-4mm}
    \end{itemize}
    The formulae express that the mismatch 
    position is given by the sum of lengths of preceding variable assignments and the local mismatch position in the particular variable $x_i$ or $y_j$.
  \item  If $x_i$ and $y_j$ are occurrences of the same variable~$z$, in order
    to get the position of the second local mismatch we have to add
    $\numof{\tagposoneof{z}}$ to $\numof{\tagpostwoof{z}}$ since the second
    local mismatch in $z$ has to be counted from the beginning of $z$ and not
    from the beginning of the previous mismatch.
    Formally, the formula is given as
    \vspace{-2mm}
    \begin{equation}
    \begin{aligned}
      \formdiseqposij \quad\defiff \quad& \left(
        \numof{\tagposoneof{z}} +
        \sum_{1 \leq u < i} \numof{\taglenof{x_u}}
        =
        \numof{\tagposoneof{z}} + \numof{\tagpostwoof{z}} +
        \sum_{1 \leq v < j} \numof{\taglenof{y_v}}
      \right) \lor {} \\
      &
      \left(
      \numof{\tagposoneof{z}} + \numof{\tagpostwoof{z}} +
      \sum_{1 \leq u < i} \numof{\taglenof{x_u}} =
      \numof{\tagposoneof{z}} +
      \sum_{1 \leq v < j} \numof{\taglenof{y_v}}
      \right).
    \end{aligned}
    \end{equation}
\end{enumerate}
Then, for each $\formdiseqposij$, we need to combine it with a~formula that
says that there are indeed mismatches in~$x_i$ and~$y_j$, to obtain formulae
$\varphi_{i,j}$ as follows:
\begin{itemize}
  \item  if $x_i \preccurlyeq y_j$ (including the case when they are occurrences of the same variable),

\vspace{-2mm}
\begin{equation}
  \varphi_{i,j} \quad\defiff\quad \formdiseqposij
  \land
  \sum_{a\in\alphabet} \numof{\tagmisoneof{x_i, a}} > 0
  \land
  \sum_{a\in\alphabet} \numof{\tagmistwoof{y_j, a}} > 0, %\quad\text{and}
\end{equation}
\vspace{-3mm}

  \item if $x_i \succ y_j$,
\vspace{-1mm}
\begin{equation}
  \varphi_{i,j} \quad\defiff\quad \formdiseqposij
  \land
  \sum_{a\in\alphabet} \numof{\tagmisoneof{y_j, a}} > 0
  \land
  \sum_{a\in\alphabet} \numof{\tagmistwoof{x_i, a}} > 0.
\end{equation}
\vspace{-2mm}
\end{itemize}
%
% \yfc{I thought that $x_i$ and $y_j$ are the mismatched position, then why you require both of the to have the same character (the summarization)?  Aha, the sum is to say at least one mismatched label is taken.}

We do the case split based on the order of variables since the construction
of~$\aut^{\singlesimplediseq}$ guarantees in which variable will be which
mismatch.
All~$\varphi_{i,j}$ formulae are then collected into the formula
\begin{equation}
  \formdiseqmis^{\singlediseq}\quad\defiff\quad \bigvee_{\substack{1 \leq i \leq n\\ 1 \leq j \leq m}} \varphi_{i,j} 
\end{equation}
and the final formula equisatisfiable to $x_1 \ldots x_n \neq y_1 \ldots y_m$ is
then defined as
\begin{equation}
  \varphi^{\singlediseq} \quad\defiff\quad \parikhformtagof{\aut^{\singlediseq}} \land
    \Big(\formdiseqlen^{\singlediseq} \lor 
    \big(\formdiseqsym^{\singlediseq} \land \formdiseqmis^{\singlediseq}
    \big)\Big).
\end{equation}

\cbstart
\begin{theorem}\label{thm:}
The formula $\regconstr' \land \intconstr \land x_1 \ldots x_n \neq y_1 \ldots y_m$
  is equisatisfiable to the formula
$\intconstr \land \varphi^{\singlediseq}$.
Moreover, the size of~$\varphi^{\singlediseq}$ is polynomial to~$nm\cdot|\regconstr'|$.
\end{theorem}
\cbend

%%%%%%%%%%%%%%%%%%%%%%%%%%%%%%%%%%%%%%%%%%%%%%%%%%%%%%%%%%%%%%%%%%%%%%%%%%%%%%%%
\vspace{-3.0mm}
\subsection{\textrm{III}: A System of Disequalities }\label{sec:diseqSystem}
\vspace{-0.0mm}
%%%%%%%%%%%%%%%%%%%%%%%%%%%%%%%%%%%%%%%%%%%%%%%%%%%%%%%%%%%%%%%%%%%%%%%%%%%%%%%%

We now move to the general case of a~system of disequalities, which all need to
be satisfied at the same time:
\begin{equation}
  \bigwedge_{1\leq i \leq n} L_i \neq R_i
\end{equation}
where each~$L_i$ and~$R_i$ are arbitrary concatenations of variables, with
potentially multiple occurrences in multiple disequalities.
We, again, construct a~tag automaton and a~corresponding LIA formula for it.
The tag automaton for this case will be more complex.

One could extend the construction of~$\aut^{\singlediseq}$ from the previous
section in a~straightforward manner by creating more copies of~$\autcon$.
With multiple disequalities, we need to keep track of satisfying position
mismatches in particular disequalities, so that we do not count two mismatches
in one disequality and zero mismatches in another disequality.
If done in a~straightforward way, we would need to consider all possible orders
of mismatches in different disequalities, basically having one copy of the tag
automaton~$\aut^{\singlediseq}$ for each such an order.
The number of these copies and the size of the resulting automaton would,
however, be intractable.
In particular, if we consider a~set of disequalities $D = \{D_1, \ldots, D_n\}$,
we would need $\frac{(2n)!}{2^n} \in 2^{\bigtheta{n \log n}}$ such copies
(obtained as the number of permutations of a~set of $n$~pairs of symbols
$\tagmisone^i, \tagmistwo^i$, respecting the order $\tagmisone^i \prec
\tagmistwo^i$).

Another issue that we need to take into consideration is the fact that one
mismatched symbol may be used for solving more than one disequality.
For instance, for the system of disequalities $x\neq y \land x\neq z$ and its
model $\{ x\mapsto a, y\mapsto b, z\mapsto c \}$, the value of~$x$ is
a~mismatch for both disequalities.
To deal with these issues, we take a~more involved approach.
Our approach is based on introducing two new types of tags:
\begin{enumerate}[(i)]
  \item  Instead of mismatch tags $\tagof{\tagmis_i, a, x}$ for $i \in \{1,2\}$
    from~$\tags^{\singlediseq}$, we use more complex tags for mismatches of the
    form $\tagof{\tagmis_i, x, D, s, a}$ denoting that the $i$-th mismatched
    symbol for the disequality $D$ on the side $s \in\{ \leftsymb, \rightsymb
    \}$ tracked for the variable $x$ was~$a$.
    The mismatches can appear in an arbitrary order in an accepting run of the
    tag automaton (potentially also multiple or zero times), so we will need to
    extend the final LIA formula with a~part that makes sure that we have
    a~mismatch for both sides of every disequality.

  \item  We introduce $\tagcop$opy tags $\tagof{\tagcop_i, x, D,
    s}$, which express that the $i$-th mismatch symbol for disequality~$D$
    and side~$s \in \{\leftsymb, \rightsymb\}$ is given by the latest symbol
    sampled by a~$\tagmis$-tag for variable~$x$.

\end{enumerate}
With these two new kinds of tags and corresponding constraints added to the
final LIA formula, we suffice with having a~tag automaton with only $2n+1$
copies of~$\autcon$.

%%%%%%%%%%%%%%%%%%%%%%%%%%%%%%%%%%%%%%%%%%%%%%%%%%%%%%%%%%
\newcommand{\figRunMulti}[0]{
\begin{figure}
  \centering
  \scalebox{1.0}{
  \input{./figs/run-multi-diseqs}
  }
  \vspace*{-3mm}
  \caption{
      An example of a run satisfying the system $D_1 \land D_2$ where
      $D_1 \defiff x \neq y$ and $D_2 \defiff x \neq z$.
  }
  \label{fig:runMultiDiseqs}
  \vspace{-3mm}
\end{figure}
}

%%%%%%%%%%%%%%%%%%%%%%%%%%%%%%%%%%%%%%%%%%%%%%%%%%%%%%%%%%%%%%%%%%%%%%%%%%%%%%%%
\subsubsection{Tag Automaton Construction}

Let $\autcon$ be the $\epsilon$-concatenation of NFAs for all variables
obtained in the same way as described in \cref{sec:single-gen-diseq}.
Then $\aut^{\multidiseq} = ( Q, \Delta, I, F)$ is a~TA
over $\tags^{\multidiseq} = \tags_{\concat} \cup \{\tagof{\tagmis_i, x, D, s, a}, \tagof{\tagcop_i, x, D, s} \mid a \in \alphabet, x \in \vars, 1 \leq i \leq 2n, 1\leq D \leq n, s \in\{ \leftsymb, \rightsymb \}\} \cup \{ \tagof{\tagpos_i, x} \mid x\in\vars, 1 \leq i \leq 2n + 1 \}$.
$\aut^{\multidiseq}$ is constructed as follows:
\begin{itemize}
  \item  $Q = \{ (q,i) \mid q \in Q_{\concat}, 1 \leq i \leq 2n + 1 \}$,
  \item  $I = I_{\concat} \times \{1\}$,
  \item  $F = F_{\concat} \times \{1, 3, \dots, 2n + 1\}$, and

  \item  $\Delta$ is the union of the following sets of transitions:
    \begin{itemize}
      \item  $\{\move {(q,i)} {\tagsymof a, \taglenof z, \tagof{\tagpos_i, z}} {(r,
        i)} \mid \move q {\tagsymof a, \taglenof{z}} r \in \Delta_{\concat}, 1 \leq i \leq 2n + 1\}$,
      \item  $\{\move {(q,i)} {} {(r,
        i)} \mid \move q {} r \in \Delta_{\concat}, 1 \leq i \leq 2n + 1\}$,
    \item $\{\move{(q,i)}{\tagsymof a, \tagof{\tagmis_{i}, z, D, s, a}, \taglenof{z}, \tagof{\tagpos_{i+1}, z}}{(r, i+1 )}
        \mid \move q {\tagsymof a, \taglenof{z}} r \in \Delta_{\concat}, 1 \leq D \leq n, 1 \le i \le 2n, s \in\{ \leftsymb, \rightsymb \} \}$ --- a~mismatch guess for the disequality $D$ and its side $s$, and
      \item $\{\move{(q,i)}{\tagof{\tagcop_i, x, D, s}}{(q, i+1 )}
      \mid 1 \leq D \leq n, 2 \le i \le 2n, s \in\{ \leftsymb, \rightsymb \}\}$
        --- a~guess that a~mismatch previously seen in $x$ is shared with the
        disequality $D$ and its side $s$.
    \end{itemize}
\end{itemize}
A run of $\aut^{\multidiseq}$ nondeterministically guesses possible mismatches for 
disequalities, as well as which mismatch is shared by multiple disequalities (the
correctness of the guess is enforced by the final LIA formula).
The run also guesses which disequalities are satisfied due to a mismatch and
which are satisfied by the lengths violation (that is why $F$ contains
accepting states within all odd-labelled internal copies: each length-satisfied
disequality removes the need for two mismatches).
An example of a~selected run of~$\aut^{\multidiseq}$ is shown in
\cref{fig:runMultiDiseqs}.

\begin{figure}
  \centering
  \scalebox{1.0}{
  \input{./figs/run-multi-diseqs}
  }
  \vspace*{-3mm}
  \caption{
      An example of a run satisfying the system $D_1 \land D_2$ where
      $D_1 \defiff x \neq y$ and $D_2 \defiff x \neq z$.
  }
  \label{fig:runMultiDiseqs}
  \vspace{-3mm}
\end{figure}
   %%%%%%%%%%%%%%%%%%%%

% If a run reaches
% a state in the $m$-th copy (for $0 \le m \le 2n$), then there were already $m$ 
% mismatch symbols sampled and stored within the registers.
% There are two main caveats that need to be dealt with to avoid creating an automaton
% of an~exponential size. First, there are
% $\bigo{n!}$ different orders in which mismatch transitions can be sampled. Second,
% a single mismatch-sampling transition might in fact form conflict for multiple
% disequalities. Our construction addresses the first problem by introducing
% additional LIA constraints that enforce proper sampling order, and, thus, no special
% treatment is done to the automaton structure. The second problem is dealt with 
% by introducing a~new kind of transitions, expressing
% the fact that two registers have the same value.

%%%%%%%%%%%%%%%%%%%%%%%%%%%%%%%%%%%%%%%%%%%%%%%%%%%%%%%%%%%%%%%%%%%%%%%%%%%%%%%%
\subsubsection{Formula Construction}
%%%%%%%%%%%%%%%%%%%%%%%%%%%%%%%%%%%%%%%%%%%%%%%%%%%%%%%%%%%%%%%%%%%%%%%%%%%%%%%%

We construct a~LIA formula equisatisfiable to the system of disequalities 
based on the tag automaton described above.
Contrary to the case of a~single
disequality, the resulting formula is enhanced by subformulae ensuring
consistency of each nondeterministic choice. For simplicity, we introduce auxiliary
(integer) variables describing particular choices that are then used in the LIA
subformulae: 
\begin{inparaenum}[(i)]
  \item $m_{D,s}$ variables containing the mismatch symbol for a~disequality $D$ and its side $s$ and
  \item $c_i$ variables containing the shared $i$-th mismatch symbol (the mismatch symbol 
  preceding the $\tagcop_i$-tag).
  % \item $p_{D,s}$ variable holding a position inside a variable containing the mismatch for  a disequality $D$ and its side $s$. 
\end{inparaenum}

We start with auxiliary subformulae expressing that the mismatches are consistent, meaning 
that each disequality and side has at most one sampled mismatch and that these mismatches are 
sampled consistently for both sides.
The first subformula $\fFair$ checks that there is at most one mismatch for each side of each disequality:
\vspace{-1mm}
\begin{equation}
  \fFair \quad\defiff \bigwedge_{\substack{D \in \{D_1, \dots, D_n \} \\ s \in \{ \leftsymb, \rightsymb \}}}\Big(
  \sum_{\substack{1 \le i \le 2n \\ x \in \vars, \symbof{a} \in \alphabet }}
    \numof{\tagof{\tagmis_{i}, x, D, s, a}} + 
    \numof{\tagof{\tagcop_i, x, D, s}} \leq 1
  \Big).
\end{equation}
The subformula $\fConsist$ then ensures that the quantified variables containing the mismatch symbols 
are properly set, including the case of the copy tag, where the mismatch is
inherited from the previous mismatch transition. 
\vspace{-3mm}
\begin{equation}
\begin{aligned}
  \fConsist \quad \defiff {} 
  \bigwedge_{\substack{D \in \{D_1, \dots, D_n \} \\
               s \in \{ \leftsymb, \rightsymb\}, \symbof{a} \in \alphabet, \\
               1 \le i \le 2n
       }}&
        \bigg(
        \Big(\sum_{\substack{x \in \vars}}
        \numof{\tagof{\tagmis_{i}, x, D, s, a}}  = 1 \Big) \rightarrow
              c_i = m_{D, s} = \symbof{a}
        \bigg) \land {} \\[-2mm]
        &
        \bigg(
        \Big(\sum_{\substack{x \in \vars}}
          \numof{\tagof{\tagcop_i, x, D, s}} = 1 \Big) \rightarrow
               c_i = m_{D, s} = c_{i-1}
        \bigg).
\end{aligned}
\vspace{-0.5mm}
\end{equation}
\vspace{-0mm}
Since not all disequalities have to be satisfied by the existence of
a mismatch (but possibly also by a~length violation), the values of $m_{D,s}$
and $c_i$ variables for disequalities with a missing mismatch (one side or
both) might hold arbitrary values.
Therefore, 
it is not sufficient to compare only values of $m_{D,\leftsymb}$ and
$m_{D,\rightsymb}$ but it is necessary to take into account existing
mismatches.
It remains to check the consistency of copy tags.
In particular, we need to ensure that copy tags for a~variable~$x$ occur on
a~run only if the previous mismatch or copy transition for~$x$ was taken.
We also need to check that a~$\tagcop$-transition was taken immediately after
the previous mismatch or copy transition ($\tagmis$~or~$\tagcop$).
\begin{equation}
\begin{aligned}
  \fCopies \defiff
        \hspace*{-2mm}\bigwedge_{\substack{1 \le i \le 2n \\ x \in \vars}}&
        \bigg(
            \Big(\hspace*{4mm} \sum_{\mathclap{\substack{D \in \{D_1, \dots, D_n \} \\ s \in \{ \leftsymb, \rightsymb \}, \symbof{a} \in \alphabet }}}
              \numof{\tagof{\tagmis_{i}, x, D, s, a}} +
              \numof{\tagof{\tagcop_i, x, D, s}} = 0  \Big)  \rightarrow
              \Big( \sum_{\mathclap{\substack{D \in \{D_1, \dots, D_n \} \\ s \in \{ \leftsymb, \rightsymb \}}}} \numof{\tagof{\tagcop_{i+1}, x, D, s}} = 0  \Big)
        \bigg) \land {}
        \\
        \bigwedge_{\substack{2 \le i \le 2n \\ x \in \vars}}
        &
        \bigg(
          \Big(\hspace*{4mm} \sum_{\mathclap{\substack{D \in \{D_1, \dots, D_n \} \\ s \in \{ \leftsymb, \rightsymb \}}}}
            \numof{\tagof{\tagcop_{i}, x, D, s}} = 1 \Big)    \rightarrow
            \numof{\tagof{\tagpos_{i}, x}} - \sum_{\mathclap{\substack{D \in \{D_1, \dots, D_n \} \\ s \in \{ \leftsymb, \rightsymb \}, \symbof{a} \in \alphabet }}}\numof{\tagof{\tagmis_{i-1}, x, D, s, a}} = 0
        \bigg).
\end{aligned} 
\end{equation}
We note that the last expression in~$\fCopies$, $\numof{\tagof{\tagpos_i, x}} - \sum \ldots$
, is there since we need to make sure that if there is
a~$\tagcop_{i}$-tag immediately after an $\tagmis_{i-1}$-tag, the number
of~$\tagof{\tagpos_{i}, x}$ is one (because there was already one
$\tagpos_{i}$ tag on the $\tagmis_{i-1}$-transition), 
but if a copy tag follows another copy tag, then the number of corresponding
position tags is zero.

The final formula will be
\begin{equation}
\varphi^{\multidiseq}\quad\defiff\quad
\parikhformtagof{\aut^{\multidiseq}} \land \fFair \land \fConsist \land \fCopies \land
\bigwedge_{\mathclap{D \in \{D_1, \ldots, D_n\}}}
  \big(\formdiseqlen^D \lor (\formdiseqmis^D \land \formdiseqsym^D)\big),
\end{equation}
where $\formdiseqlen^D$, $\formdiseqmis^D$, and~$\formdiseqsym^D$ are similar
to their counterparts in \cref{sec:single-gen-diseq} but using the~$m_{D,s}$
variables instead of directly using~$\tagmis$-tags.
Details are in \cref{app:liareduction}.

\cbstart
\begin{theorem}\label{thm:}
The formula $\regconstr' \land \intconstr \land \bigwedge_{1\leq i \leq n} L_i \neq R_i$
is equisatisfiable to the formula
$\intconstr \land \varphi^{\multidiseq}$.
Moreover, the size of~$\varphi^{\multidiseq}$ is polynomial
to~$mn\cdot|\regconstr'|$ where $m$ is the maximum size of any~$L_i$
or~$R_i$.
\end{theorem}
\cbend

%%%%%%%%%%%%%%%%%%%%%%%%%%%%%%%%%%%%%%%%%%%%%%%%%%%%%%%%%%%%%%%%%%%%%%%%%%%%%%%%
\vspace{-0.0mm}
\section{Other Position Constraints}\label{sec:other_pos}
\vspace{-0.0mm}
%%%%%%%%%%%%%%%%%%%%%%%%%%%%%%%%%%%%%%%%%%%%%%%%%%%%%%%%%%%%%%%%%%%%%%%%%%%%%%%%

In this section, we show how the framework introduced in \cref{sec:diseq} can
be extended for solving other considered position constraints.

%*******************************************************************************
\vspace{-0.0mm}
\subsection{Length Constraints}\label{sec:label}
\vspace{-0.0mm}
%*******************************************************************************

For conjunctions $\bigwedge_{1\leq i \leq n} x_i =
\lenof{y_{i1}\cdots y_{im_i}}$, where~$x_i$ are integer variables, we create the
$\epsilon$-concatenation~$\autcon$ for all string variables occurring in the constraint and
construct the formula
\begin{equation}
\varphi^{\mathit{LEN}} \quad\defiff\quad \parikhformtagof{\autcon}
\land \bigwedge_{1 \leq i \leq n}\Big( x_i = \sum_{\mathclap{1 \leq j \leq m_i}} \numof{\taglenof{y_{ij}}}\Big).
\end{equation}
\vspace{-5mm}

\cbstart
\begin{theorem}\label{thm:}
The formula $\regconstr' \land \intconstr \land \bigwedge_{1\leq i \leq n} x_i =
\lenof{y_{i1}\cdots y_{im_i}}$
  is equisatisfiable to the formula
$\intconstr \land \varphi^{\mathit{LEN}}$ and
the size of~$\varphi^{\mathit{LEN}}$ is polynomial to~$mn\cdot|\regconstr'|$.
\end{theorem}
\cbend

%%%%%%%%%%%%%%%%%%%%%%%%%%%%%%%%%%%%%%%%%%%%%%%%%%%%%%%%%%%%%%%%%%%%%%%%%%%%%%%%
\vspace{-1.0mm}
\subsection{Not Prefix and Not Suffix Predicates} \label{sec:notprefix}
\vspace{-0.0mm}
%%%%%%%%%%%%%%%%%%%%%%%%%%%%%%%%%%%%%%%%%%%%%%%%%%%%%%%%%%%%%%%%%%%%%%%%%%%%%%%%
The $\notprefix(x_1 \cdots x_n, y_1 \cdots y_m)$ and $\notsuffix(x_1 \cdots
x_n, y_1 \cdots y_m)$ predicates are similar to a disequality $x_1 \cdots x_n
\neq y_1 \cdots y_m$ in that they are satisfied if there is a mismatch at the
same global position between their first and second argument.
Both predicates, however, have slightly different
conditions  than disequalities on satisfiability due to their sides having incompatible lengths---the
first argument ($x_1 \cdots y_n$) must be strictly longer than the second argument ($y_1 \cdots y_m$).
Therefore, the tag-automaton construction is the same as in the
case of a single unrestricted disequality given in \cref{sec:single-gen-diseq}.
Forming an equisatisfiable LIA formula also remains the same, save for small
differences.
The different condition on satisfiability by $x_1 \cdots x_n$ and $y_1 \cdots y_m$
having incompatible lengths requires replacing the corresponding subformula
$\formdiseqlen^{\singlediseq}$ by $\varphi^{*\mathrm{FIX}}_{\mathit{len}}$ defined as
\begin{equation}
    \varphi^{*\mathrm{FIX}}_{\mathit{len}} \quad\defequiv\quad
        \sum_{1 \le i \le n} \numof {\taglenof{x_i}} > \sum_{1 \le j \le m} \numof{\taglenof{y_j}}.
\end{equation}
\vspace{-3mm}

Furthermore, the $\notsuffix$ predicate treats the mismatch position
differently than $\notprefix$. Instead of being satisfied by a mismatch on the
same global position starting from the beginning of its arguments, the $\notsuffix$
predicate counts the mismatch position from the end of its arguments. Therefore, we
also need to replace the $\formdiseqposij$ subformulae with $\formdiseqposij^{\mathrm{NS}}$
asserting that the mismatch positions in both arguments are the same, using the~$\tagposthreeof x$-tags, which we already added into~$\aut^{\singlediseq}$ in \cref{sec:single-gen-diseq}. We define
$\formdiseqposij^{\mathrm{NS}}$ as follows:
\begin{enumerate}
  \item  If $x_i$ and $y_j$ are occurrences of a~different variable, then
    \begin{equation}
        \hspace*{-2mm}\formdiseqposij^{\mathrm{NS}} \!\!\defiff\!
        \begin{cases} 
            \numof {\tagpostwoof{x_i}} +
            \numof {\tagposthreeof{x_i}} +
            \sum_{1 \leq u < i} \numof{\taglenof{x_u}}
            =
            \numof{\tagposthreeof{y_j}} +
            \sum_{1 \leq v < j} \numof{\taglenof{y_v}}
            ~ \textrm{  if } x_i \prec y_j, \\[2mm]
            \numof {\tagposthreeof{x_i}} +
            \sum_{1 \leq u < i} \numof{\taglenof{x_u}}
            =
            \numof{\tagpostwoof{y_j}} +
            \numof{\tagposthreeof{y_j}} +
            \sum_{1 \leq v < j} \numof{\taglenof{y_v}} ~ \textrm{  otherwise.}
        \end{cases} 
    \end{equation}
  \item  If $x_i$ and $y_j$ are occurrences of the same variable~$z$, then
    \begin{equation}
    \begin{aligned}
      \formdiseqposij^{\mathrm{NS}} \defiff {}& \left(
        \numof{\tagpostwoof{z}} +
        \numof{\tagposthreeof{z}} +
        \sum_{1 \leq u < i} \numof{\taglenof{x_u}}
        =
        \numof{\tagposthreeof{z}} +
        \sum_{1 \leq v < j} \numof{\taglenof{y_v}}
      \right) \lor {}  \\
      &
      \left(
      \numof{\tagposthreeof{z}} +
      \sum_{1 \leq u < i} \numof{\taglenof{x_u}} =
      \numof{\tagpostwoof{z}} +
      \numof{\tagposthreeof{z}} +
      \sum_{1 \leq v < j} \numof{\taglenof{y_v}}
      \right).
    \end{aligned}
    \end{equation}
\end{enumerate}
Intuitively, we start counting the mismatch position inside a variable \emph{after} the
mismatch has been sampled rather than counting \emph{until} it has been sampled.
We denote the corresponding constructed formulae as~$\varphi^{\mathit{pre}}$
(for $\notprefix$) and $\varphi^{\mathit{suf}}$ (for $\notsuffix$).

\cbstart
\begin{theorem}\label{thm:}
The formula $\regconstr' \land \intconstr \land \notprefix(x_1 \cdots x_n, y_1
\cdots y_m)$ is equisatisfiable to the formula $\intconstr \land
\varphi^{\mathit{pre}}$ and the formula $\regconstr' \land \intconstr \land
\notsuffix(x_1 \cdots x_n, y_1 \cdots y_m)$ is equisatisfiable to the formula
$\intconstr \land \varphi^{\mathit{suf}}$.
The sizes of $\varphi^{\mathit{pre}}$ and $\varphi^{\mathit{suf}}$ are
polynomial to~$mn\cdot|\regconstr'|$.
\end{theorem}
\cbend

%%%%%%%%%%%%%%%%%%%%%%%%%%%%%%%%%%%%%%%%%%%%%%%%%%%%%%%%%%%%%%%%%%%%%%%%%%%%%%%%
\vspace{-0.0mm}
\subsection{Symbol (not) at a~Position}\label{sec:label}
\vspace{-0.0mm}
%%%%%%%%%%%%%%%%%%%%%%%%%%%%%%%%%%%%%%%%%%%%%%%%%%%%%%%%%%%%%%%%%%%%%%%%%%%%%%%%

Starting with the negative case first, let the input predicate be $x_s \neq \strat(y_1 \cdots y_m, x_i)$
where $x_s, y_1, \dots y_n$ are string variables and $x_i$ is an integer
variable. We construct the tag automaton $\aut$ in the same way as described in
\cref{sec:single-gen-diseq}. To form an equisatisfiable LIA formula, we
modify the reduction from \cref{sec:single-gen-diseq} to capture that
the mismatch position of the left-hand side is given by $x_i$ rather than being
nondeterministically given by a run in the automaton:
\begin{equation}
    \varphi_{1, j} \defiff \begin{cases}
        x_i = \numof{ \tagof{ \tagpos_1, y_j} } + \sum_{1 \le k < j} \numof{\taglenof{y_k}} & \text{  if } y_i \prec x_s, \\
        x_i = \numof{ \tagof{ \tagpos_2, y_j} } + \sum_{1 \le k < j} \numof{\taglenof{y_k}} & \text{  otherwise.}
    \end{cases}
\end{equation}
We further introduce an auxiliary predicate~$\varphi_{\mathit{InBounds}}$ checking that~$x_i$ is a~valid position in $y_1 \cdots y_m$:
\begin{equation}
    \varphi_{\mathit{InBounds}}\quad \defiff\quad 0 \leq x_i < \sum_{1 \le j \le m} \numof{\taglenof{y_j}}.
\end{equation}
The final formula $\varphi^{\neg\strat}$ is then modified to capture possible invalid positions of~$x_i$:
\begin{equation}
\begin{aligned}
  \varphi^{\neg\strat} \quad\defiff\quad \parikhformtagof{\aut} \land
    \bigg(
        & \big(\numof{\taglenof{x_s}} > 0 \land \neg \varphi_{\mathit{InBounds}} \big) \lor  \numof{\taglenof{x_s}} > 1 \lor {}\\
        & \big(\numof{\taglenof{x_s}} = 1 \land 
        \varphi_{\mathit{InBounds}} \land \formdiseqsym \land \bigvee_{\substack{1 \leq j \leq m}} \varphi_{1,j}
    \big)
    \bigg).
\end{aligned}
\end{equation}

\cbstart
\begin{theorem}\label{thm:}
The formula $\regconstr' \land \intconstr \land x_s \neq \strat(y_1 \cdots y_m,
x_i)$ is equisatisfiable to the formula $\intconstr \land \varphi^{\neg\strat}$
and the size of $\varphi^{\neg\strat}$ is polynomial to~$m\cdot|\regconstr'|$.
\end{theorem}
\cbend

The $x_s = \strat(y_1 \cdots y_m, x_i)$ predicate can be reduced in a similar
fashion, requiring us to replace $\formdiseqsym$ with $\formdiseqsym'$, which enforces
the sampled letters to be the same rather than being different.
\begin{equation}
\begin{aligned}
  \varphi^\strat \quad\defiff\quad \parikhformtagof{\aut} \land
    \bigg(
        & \big(\numof{\taglenof{x_s}} = 0 \land \neg \varphi_{\mathit{InBounds}} \big) \lor {} \\[-2mm]
        & \big(\numof{\taglenof{x_s}} = 1 \land 
        \varphi_{\mathit{InBounds}} \land \formdiseqsym' \land \bigvee_{\substack{1 \leq j \leq m}} \varphi_{1,j}
    \big)
    \bigg).
\end{aligned}
\vspace{-3mm}
\end{equation}

\cbstart
\begin{theorem}\label{thm:}
The formula $\regconstr' \land \intconstr \land x_s = \strat(y_1 \cdots y_m,
x_i)$ is equisatisfiable to the formula $\intconstr \land \varphi^{\strat}$
and the size of $\varphi^{\strat}$ is polynomial to~$m\cdot |\regconstr'|$.
\end{theorem}
\cbend

%%%%%%%%%%%%%%%%%%%%%%%%%%%%%%%%%%%%%%%%%%%%%%%%%%%%%%%%%%%%%%%%%%%%%%%%%%
\tikzstyle{cell} = [draw, fill=white, minimum size=5mm, text height=6pt]
\tikzstyle{ghostCell} = [minimum size=5mm]
\tikzstyle{mismatch} = [fill=red!20]

\newcommand{\figNotContainsSemantics}[0]{
\begin{figure}
    \resizebox{1.0\textwidth}{!}{
      \input{figs/diseq-vs-notcontains.tex}
    }
  \vspace*{-3mm}
  \caption{
    Demonstration of how $\notcontains(u, v)$ for $u, v \in
    \vars^*$ is satisfied by an assignment $\sigma = \{u \mapsto aba$, $v\mapsto aabba\}$.
    Symbols with \textcolor{red}{red} background present mismatches
    in~the~corresponding alignments.
  }
  \label{fig:notcontainsSemantics}
\end{figure}
}

%%%%%%%%%%%%%%%%%%%%%%%%%%%%%%%%%%%%%%%%%%%%%%%%%%%%%%%%%%%%%%%%%%%%%%%%%%%%%%%%
\vspace{-3.0mm}
\subsection{Not Contains Predicate}\label{sec:notcontains}
\vspace{-0.0mm}
%%%%%%%%%%%%%%%%%%%%%%%%%%%%%%%%%%%%%%%%%%%%%%%%%%%%%%%%%%%%%%%%%%%%%%%%%%%%%%%%

% In this section, we show first steps towards solving constraints involving the
% $\notcontains$ predicate, whose general decidability is, to the best of our
% knowledge, currently unknown.

% \ol{bla bla flat restriction, talk about the general case}

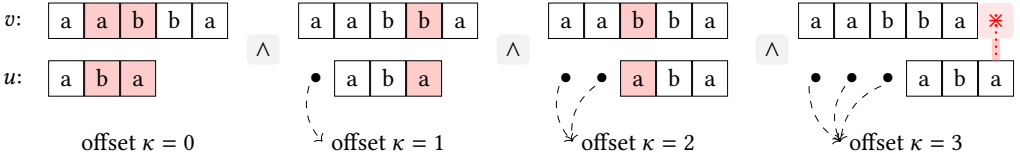
\begin{figure}
    \resizebox{1.0\textwidth}{!}{
      \input{figs/diseq-vs-notcontains.tex}
    }
  \vspace*{-3mm}
  \caption{
    Demonstration of how $\notcontains(u, v)$ for $u, v \in
    \vars^*$ is satisfied by an assignment $\sigma = \{u \mapsto aba$, $v\mapsto aabba\}$.
    Symbols with \textcolor{red}{red} background present mismatches
    in~the~corresponding alignments.
  }
  \label{fig:notcontainsSemantics}
\end{figure}
  %%%%%%%%%%%%%%%

In this section, we extend the reasoning about the disequality tag
automaton~$\aut^{\singlediseq}$ introduced in \cref{sec:single-gen-diseq} to
handling $\notcontains$ with flat languages.
The constraint $\notcontains(u,v)$ for $u,v\in\vars^*$ is satisfiable 
if there is a string assignment of variables from $u$ and $v$ yielding words $w_u$ and $w_v$ respectively such that 
for \emph{every} alignment of $w_u$ and $w_v$
\begin{inparaenum}[(i)]
  \item there is a mismatch symbol of $w_u$ and~$w_v$ or
  \item $w_u$~overflows $w_v$.
\end{inparaenum}
For example, considering the constraint $\notcontainsof{u}{v}$, the assignment
$\sigma = \{u \mapsto aba, v\mapsto aabba\}$ is a~model---for every 
alignment of $aba$ and $aabba$, there is either a~mismatching symbol or a~part
of~$aba$ is outside~$aabba$ (cf.\ \cref{fig:notcontainsSemantics}).
The alignment of $w_u$ and $w_v$ can be characterized by the offset $\offset
\in \nat$ of $w_u$ counted from the beginning of $w_v$.
Therefore, the semantics of $\notcontains$ implicitly involves a~universal
quantifier ranging over all possible offsets.

Since we need to consider mismatches for all offsets of~$w_u$ and~$w_v$, 
one might consider the formula~$\varphi^{\singlediseq}$ from
\cref{sec:single-gen-diseq}, but changed such that the position constraints
$\formdiseqposij$ take into account a~universally quantified offset
variable~$\offset$.
Such a~solution, however, does not work, since we need that
\begin{inparaenum}[(i)]
  \item for each value of~$\offset$, the string assignment remains the same (we
    want to shift the same assignment to different positions given by the
    offset and not obtain a different assignment for each offset), and 
  \item for each offset the particular mismatch position and symbol (if any)
    may be different.
\end{inparaenum}
The second property is problematic as the different offsets might involve
different runs in the tag automaton that are, however, over the same string
assignment.
For this reason, we need to impose a~\emph{flat language restriction}, since
for flat automata, the number of taken transitions (and so a~model
of~$\parikhformof \aut$) uniquely determines the accepted word.
On the other hand, for non-flat automata, this property does not hold.
E.g,. consider an NFA with a~single state~$q$ (which is both initial and
accepting) and two transitions: $q \ltr a q$ and $q \ltr b q$.
The two words $aabb$ and $bbaa$ accepted by the NFA are different, but they
have the same Parikh images.
Our restriction to flat languages gives us the guarantee that a~model of the
Parikh formula uniquely determines the string assignment.

%*******************************************************************************
\vspace{-0.0mm}
\subsubsection{Formula Construction}\label{sec:label}
\vspace{-0.0mm}
%*******************************************************************************

Let $\varphi = \notcontains(u, v)$ where $u = u_1\dots u_n$ and $v = v_1\dots v_m$ are sequences of variables from~$\vars$ such that
$\langof {\autass(x)}$ is flat for every $x \in \vars$, and let $\aut^{\singlediseq}$ be
a tag automaton constructed as described in \cref{sec:single-gen-diseq}.
Since we now need to speak about different runs of $\aut^{\singlediseq}$, we 
lift the definition of the Parikh tag image to 
explicitly speak about the Parikh variables, i.e., $\parikhformtag(T, \num)$
where $\num$ denotes the set of all $\num$-prefixed variables
in~$\parikhformtag$.
Since we need to speak about alignments of assignments, we also refine the
formulae $\formdiseqposij$ from \cref{sec:single-gen-diseq} to
$\formdiseqposij(\offset, \num)$, explicitly relating the used Parikh variables
and the particular offset~$\offset$.
The offset~$\offset$ is added to the left-hand side of every equation occurring
inside~$\formdiseqposij$, in order to express that the assignments of the
left-hand side are shifted to the right by~$\offset$.
The formula $\formdiseqmis^{\singlediseq}$ from \cref{sec:single-gen-diseq} is
then also changed into~$\formdiseqmis(\offset, \num)$, which
uses~$\formdiseqposij(\offset, \num)$ instead of~$\formdiseqposij$.

Let us start by defining some auxiliary predicates used later in the resulting formula:
\begin{align}
  \pi (\move {q} {\tagsymof a, \taglenof x} {p}) &\quad\defequal\quad \left\{ \move {(q,i)} {U} {(p,j)} \mid \move {(q,i)} {U} {(p,j)} \in \Delta(\aut^{\singlediseq}), \tagsymof a \in U \right\}~\text{and} \\
  % \pi (\move {q} {S} {p}) &\defequal \left\{ \move {(q,i)} {S} {(p,j)} \mid \move {(q,i)} {S} {(p,j)} \in \Delta(\aut^{\singlediseq}) \right\}  \\
  \equalword{\num_1} {\num_2} &\quad\defequiv\quad
    \bigwedge_{t \in \Delta (\autcon)} \left( \sum_{r \in \pi(t)} \num_1 r =
    \sum_{r \in \pi(t)} \num_2 r  \right).
\end{align}
Let $\num_1$ and $\num_2$ be two sets of Parikh variables encoding accepting runs of $
\aut^{\singlediseq}$, i.e., 
$\parikhformtagof{\aut^{\singlediseq}, \num_1}$ and $\parikhformtagof{\aut^{\singlediseq}, \num_2}$ hold. 
Intuitively, $\equalword{\num_1} {\num_2}$ is satisfied iff $\num_1$ and $\num_2$ correspond to the
same sets of runs in the $\epsilon$-concatenation~$\autcon$ serving as a basis for $\aut^{\singlediseq}$, and, therefore, 
since we assume flat automata, 
$\num_1$~and~$\num_2$ correspond to the same string assignments. 
% Since we assume that the variable languages are
% flat, $\pisAgree {\aut^{\singlediseq}} {\piModelA} {\piModelB}$ holds iff $\piModelA$ and $
% \piModelB$ correspond to the same string assignment. Finally, we can construct $
% \notcontainsLIA_\varphi$ as follows:
%
Further, we define $\lenDiff {\num}$ expressing the difference $|w_v| - |w_u|$ 
where $w_u$ and $w_v$ are concatenated assignments of $\notcontains$'s arguments 
given by the Parikh image:
\vspace{-2mm}
\begin{equation}
  \lenDiff {\num} \quad\defequal\quad
      \left( \sum_{1 \le j \le m} \numof{\taglenof{v_j}} \right)
      - \left(\sum_{1 \le i \le n} \numof{\taglenof{u_i}} \right).
%\vspace{-1mm}
\end{equation}
The resulting formula equisatisfiable to the $\notcontains$ is then given as
\begin{equation}
  \begin{aligned}
    \notcontainsLIA\quad \defequiv \quad
    \parikhformtagof{\aut^{\singlediseq}, \num_1} \land &
      \forall \offset \exists \num_2
      \Big(
        \big(
          \parikhformtagof{\aut^{\singlediseq}, \num_2} \land {}  \equalword{\num_1} {\num_2} \land \formdiseqmis(\offset, \num_2) 
        \big)
    \lor {} &\\[-2mm]
        &
        \hspace*{12mm}
        \offset < 0
        \lor
        \offset > \lenDiff{\num_1}
      \Big).
  \end{aligned}
\end{equation} 
Note that in the formula, we use $\exists \num$ to denote the existential quantification over all variables in $\num$.
Intuitively, $\notcontainsLIA$ is satisfied iff there is a string assignment corresponding to $\num_1$ such that
for every offset $\offset$, we can find some other model $\num_2$ satisfying
$\parikhformtagof{\aut^{\singlediseq}, \num_2}$ that encodes the same string
model (but potentially a~different run through~$\aut^{\singlediseq}$.
Moreover, $\notcontainsLIA$ says that the
run corresponding to~$\num_2$ contains a mismatch for the offset~$\kappa$.
Alternatively, the offset $\kappa$ might be larger than the difference between lengths of $\notcontains$ arguments or negative, 
in which case the corresponding alignment
is trivially satisfied.

\cbstart
\begin{theorem}\label{thm:}
The formula $\regconstr' \land \intconstr \land \notcontains(u_1\dots u_n,
v_1\dots v_m)$ where the language of each~$u_i$ and~$v_j$ is flat is
equisatisfiable to the formula $\intconstr \land \notcontainsLIA$.
Moreover, the size of $\notcontainsLIA$ is polynomial to~$mn\cdot|\regconstr'|$.
\end{theorem}
\cbend

%*******************************************************************************
\vspace{-0.0mm}
\subsection{Arbitrary Combination of Position Predicates}\label{sec:combination}
\vspace{-0.0mm}
%*******************************************************************************

The construction of the tag automaton for multiple disequalities and the
subsequent LIA reduction can be easily extended to a system $\psi \equiv \bigwedge_{1 \le k
\le K} P_{k}(x_{k, 1} \cdots x_{k, n_k}, y_{k, 1} \cdots y_{k, m_k})$ where
$P_{k} \in \{\neq$, $\notprefix$, $\notsuffix$, $\strat$, $\notstrat$, $\notcontains \}$.
From a high-level perspective, a single $\epsilon$-concatenation~$\autcon$ of
all automata of variables occurring in $\psi$ is created. The tag automaton
is formed in the same way as described in \cref{sec:diseqSystem},
containing $2K+1$ copies of $\autcon$ to track up to $2K$ possible
mismatch symbols (one for each side of every predicate).
The resulting LIA formula is then
\begin{equation}
    \varphi^\combination \defiff \fParikh \land \fConsist{} \land \fCopies{} \land
    \bigwedge_{1 \le i \le K} \fPredSat {i}
\end{equation}
where $\fPredSat {i}$ is a LIA formula specific to the type of $i$-th constraint
described in previous sections expressing that the predicate is satisfied. Note
that $\fPredSat {i}$ needs to be modified in the same way as in the case of a
system of multiple disequalities (cf. \cref{sec:diseqSystem}) to make use of
the $p_{D, s}$ and $m_{D, s}$ variables.
Furthermore, all variables present in any
$\notcontains$ predicate must have a flat language to maintain soundness, similar
as in the case of a~single $\notcontains$ predicate.

\cbstart
\begin{theorem}\label{thm:}
The formula $\regconstr' \land \intconstr \land \psi$
  is equisatisfiable to the formula
$\intconstr \land \varphi^{\combination}$, provided all variables occurring
  in~$\notcontains$ constraints within~$\psi$ are assigned flat languages
  by~$\regconstr'$.
  In addition, the size of~$\varphi^{\combination}$ is polynomial
  to~$mn\cdot|\regconstr'|$ where~$n$ is the number of constraints in~$\psi$
  and~$m$ is the maximum size of any side of the constraints in~$\psi$.
\end{theorem}
\cbend

%%%%%%%%%%%%%%%%%%%%%%%%%%%%%%%%%%%%%%%%%%%%%%%%%%%%%%%%%%%%%%%%%%%%%%%%%%%%%%%%
\vspace{-0.0mm}
\section{Decidability and Complexity}\label{sec:dec_complexity}
\vspace{-0.0mm}
%%%%%%%%%%%%%%%%%%%%%%%%%%%%%%%%%%%%%%%%%\begin{center}
This section covers theoretical results that follow from tag-automaton
constructions based on position predicates and subsequent reductions into LIA
described in previous sections. We will formulate our results as instances of
the following parametrized decision problem.

\begin{center}
\begin{tabular}{ll}
  \toprule
$\posregsat(\eqconstr, \regconstr, \intconstr, \diseqconstr)$
\\ \midrule
  \textbf{INPUT:}
  &
  $\bullet$ a~set of string variables~$\vars = \{x_1, x_2, \dots, x_k\}$,
  \\
  &
  $\bullet$ a conjunction of word equations $\eqconstr$,
  \\
  &
  $\bullet$ a conjunction of regular constraints $\regconstr = \big\{ x \in \langof{\aut_x} \mid x \in \vars \big\}$,
  \\
  &
  $\bullet$ a conjunction of length constraints $\intconstr$, and
  \\
  &
  $\bullet$ a conjunction of position constraints $\diseqconstr$.
  \\
  \textbf{QUESTION:}
  &
  \parbox[t]{9.8cm}{Is there an assignment $\sigma\colon \vars \to
\alphabet^*$ satisfying $\eqconstr \land \regconstr \land \intconstr \land \diseqconstr$?}\\
\bottomrule
\end{tabular}
\end{center}

In the following, we write $\regconstr$ to denote an~arbitrary conjunction of
regular constraints of the form $\regconstr = \bigwedge_{x \in \vars}  x \in
\langof{\aut_x}$ where $\aut_x$ is an NFA associated with the variable $x$
if not specified otherwise.

%%%%%%%%%%%%%%%%%%%%%%%%%%%%%%%%%%%%%%%

\cbstart
\begin{theorem}\label{thm:ptime}
    The time complexity of $\posregsat(\emptyset, \regconstr, \emptyset, \diseqconstr)$ with
    $\diseqconstr = P(x_1 \cdots x_n, $ $ y_1 \cdots y_m)$ for $P \in \{\neq, \notsuffix, \notprefix \}$ is
    in $\clP$, more concretely in $\bigOOf{nm \cdot |\alphabet|^3\cdot |\regconstr|^6}$.
\end{theorem}
\cbend

\newcommand{\ca}[0]{\mathcal{C}}
\newcommand{\caFirst}[0]{\ca^{1}}
\newcommand{\caSecond}[0]{\ca^{2}}
\newcommand{\caThird}[0]{\ca^{3}}
\newcommand{\caCounter}[0]{\mathbf{c}}
\newcommand{\caCounterL}[0]{\caCounter_\mathrm{L}}
\newcommand{\caCounterR}[0]{\caCounter_\mathrm{R}}

\begin{proof}[Proof outline]
  We outline the proof for $P$ being a disequality; the other cases are similar.
% Let $\psi \defiff x_1 \cdots x_n \neq y_1 \cdots y_m$ be the input
  %position predicate. 
  We construct a one-counter automaton~$C$ with
  a~counter~$\caCounter$ with updates limited to $\{-1, 0, +1\}$ such that 
  there is an accepting state reachable with
  $\caCounter = 0$ iff the input combination of regular constraints
  and~$P$ is satisfiable. We show that $C$ has a polynomial size to the input and 
  using the result of \cite[Lemma~11]{AlurC11} stating that $0$-reachability of a state in a~one-counter automaton can be decided
  in $\clP$, we get the theorem. The full proof is given in \cref{app:ptimeProof}.
\end{proof}

\begin{lemma}\label{lemma:nphard}
    $\posregsat(\emptyset, \regconstr, \emptyset, \diseqconstr)$ with
    $\diseqconstr = \bigwedge_{1 \le i \le K} (x_{i, 1} \cdots x_{i, n_i} \neq y_{i, 1} \cdots y_{i, m_i})$
    is $\clNP$-hard.
\end{lemma}
\begin{proof}
    By reduction from 3-SAT.
    Let $\varphi$ be an input 3-SAT formula.
    For each Boolean variable~$x_i$ in~$\varphi$, we create a~string
    variable~$y_i$ with $\autass(y_i)$ being a~DFA with 2 states accepting the
    language $\{0,1\}$.
    For each clause in~$\varphi$, we create a~new disequality such that, e.g.,
    for a~clause $(x_1 \lor \neg x_2 \lor x_3)$, we create the disequality $y_1
    y_2 y_3 \neq 010$.
    Then the system of disequalities is equisatisfiable to~$\varphi$.
\end{proof}

\begin{theorem}\label{thm:npcomplete}
    $\posregsat(\emptyset, \regconstr, \emptyset, \diseqconstr)$ with
    $\diseqconstr = \bigwedge_{1 \le i \le K} P_i(x_{i, 1} \cdots x_{i, n_i}, y_{i, 1} \cdots y_{i, m_i})$
    for $P_i \in \{\neq$, $\notsuffix, \notprefix, \strat, \notstrat \}$
    is $\clNP$-complete. 
\end{theorem}
\begin{proof}
    From \cref{lemma:nphard} we have that the problem is $\clNP$-hard.
    $\clNP$-membership follows from constructing an equisatisfiable quantifier-free
    LIA formula $\psi$ as described in~\cref{sec:combination} and
    observing that $\psi$ is of polynomial size.
    Satisfiability of quantifier-free LIA is in $\clNP$~\cite{qfLiaNP}.
\end{proof}

The following theorem states that position constraints with structurally
limited languages of variables occurring in $\notcontains$ predicates
can be decided in $\clNEXPTIME$.

\begin{theorem}\label{thm:notContainsComplexity}
    $\posregsat(\emptyset, \regconstr, \emptyset, \diseqconstr)$ with
    $\diseqconstr = \bigwedge_{1 \le i \le K} P_i(x_{i, 1} \cdots x_{i, n_i}, y_{i, 1} \cdots y_{i, m_i})$
    for $P_i \in \{\neq$, $\notsuffix, \notprefix, \notcontains, \strat, \notstrat \}$
    such that $\langof{\aut_x}$ of any variable $x$ that occurs in a~$\notcontains$
    predicate is flat can be decided in $\clNEXPTIME$.
\end{theorem}
\begin{proof}
    We can observe that, in the presence of $\notcontains$ predicates,
    the resulting formula $\psi$ constructed as described in \cref{sec:combination}
    falls into the $\exists \forall \exists$-fragment of LIA (after transforming into the prenex normal form).
    As the number of quantifier alternations is $2$, we obtain that
    $\psi$ is decidable in $\Sigma^{\clEXP}_1 = \clNEXPTIME$~\cite{Haase14},
    where $\Sigma^{\clEXP}_1$ is the first level of the weak exponential hierarchy.
\end{proof}

Contrary to deciding a~single disequality (which is in \clP), deciding
a~single $\notcontains$ is already \clNP{}-hard, as stated by the following theorem
(proven in \cref{app:npHardess}).

\begin{theorem}\label{thm:notContainsHardness}
    $\posregsat(\emptyset, \regconstr, \emptyset, \diseqconstr)$ with
    $\diseqconstr = \notcontains(x_1 \ldots x_n, y_1 \ldots y_m)$
    is $\clNP$-hard.
\end{theorem}

Finally, we obtain the decidability of the whole fragment considered in the
paper for chain-free word equations.

\begin{theorem}\label{thm:decidability}
    $\posregsat(\eqconstr, \regconstr, \intconstr, \diseqconstr)$ with
    $\eqconstr$ being chain-free~\cite{ChainFree},
    $\diseqconstr = \bigwedge_{1 \le i \le K} P_i(x_{i, 1} \cdots x_{i, n_i}, $ $ y_{i, 1} \cdots y_{i, m_i})$
    for $P_i \in \{\neq, \notsuffix, \notprefix, \notcontains, \strat, \notstrat \}$
    and
    $\regconstr = \bigwedge_{x \in \vars} x \in \langof{\aut_x}$
    such that $\langof{\aut_x}$ of any variable~$x$ that occurs in a~$\notcontains$
    predicate is flat
    is decidable. 
\end{theorem}

\begin{proof}
    As $\eqconstr$ is chain-free, we start by solving only $\eqconstr \land
    \regconstr \land \intconstr$ using the approach described
    in~\cite{ChenCHHLS23}, obtaining a new set of variables $\vars'$ along with
    a length constraint $\intconstr'$ (extension of $\intconstr$ with equalities 
    relating lengths of the original variables and the variables from $\vars'$), 
    a~monadic decomposition $\regconstr' = \bigwedge_{x' \in \vars'} x' \in
    \langof{\aut_{x'}}$ and a~length constraint (we note that the \emph{noodlification} procedure
    in~\cite{ChenCHHLS23} preserves flatness of languages), and a substitution
    map $\sigma\colon \vars \rightarrow (\vars')^*$ mapping original variables
    to (potentially concatenations of) new variables.
    Applying $\sigma$ to $\diseqconstr$, we obtain a~conjunction $\diseqconstr'$
    of new position predicates. We are left to solve a new system
    $\regconstr' \land \intconstr' \land \diseqconstr'$, for which we can construct an
    equisatisfiable LIA formula using the techiniques presented in this work.
    Therefore, we obtain an equisatisfiable formula in a decidable theory, concluding the proof.
\end{proof}

%%%%%%%%%%%%%%%%%%%%%%%%%%%%%%%%%%%%%%%%%%%%%%%%%%%%%%%%%%%%%%%%%%%%%%%%%%%%%%%%
\vspace{-0.0mm}
\section{Experimental Evaluation}\label{sec:label}
\vspace{-0.0mm}
%%%%%%%%%%%%%%%%%%%%%%%%%%%%%%%%%%%%%%%%%%%%%%%%%%%%%%%%%%%%%%%%%%%%%%%%%%%%%%%%

\newcommand{\figScatter}[0]{
\begin{figure}[t!]%
\centering
\begin{subfigure}{0.3\textwidth}
  \centering
  \includegraphics[width=\linewidth]{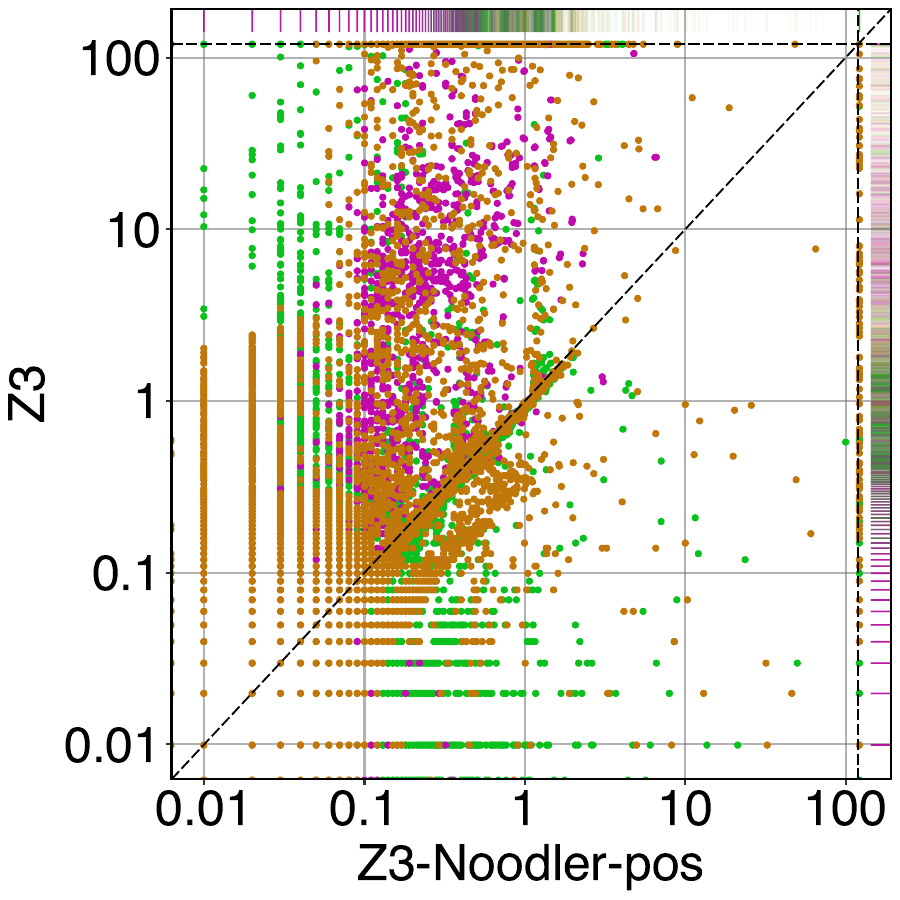}
  \vspace{-6mm}
  \caption{\footnotesize \ziiinoodlerpos vs. \ziii}
  \label{fig:z3}
\end{subfigure}
~
\begin{subfigure}{0.3\textwidth}
    \centering
    \includegraphics[width=\linewidth]{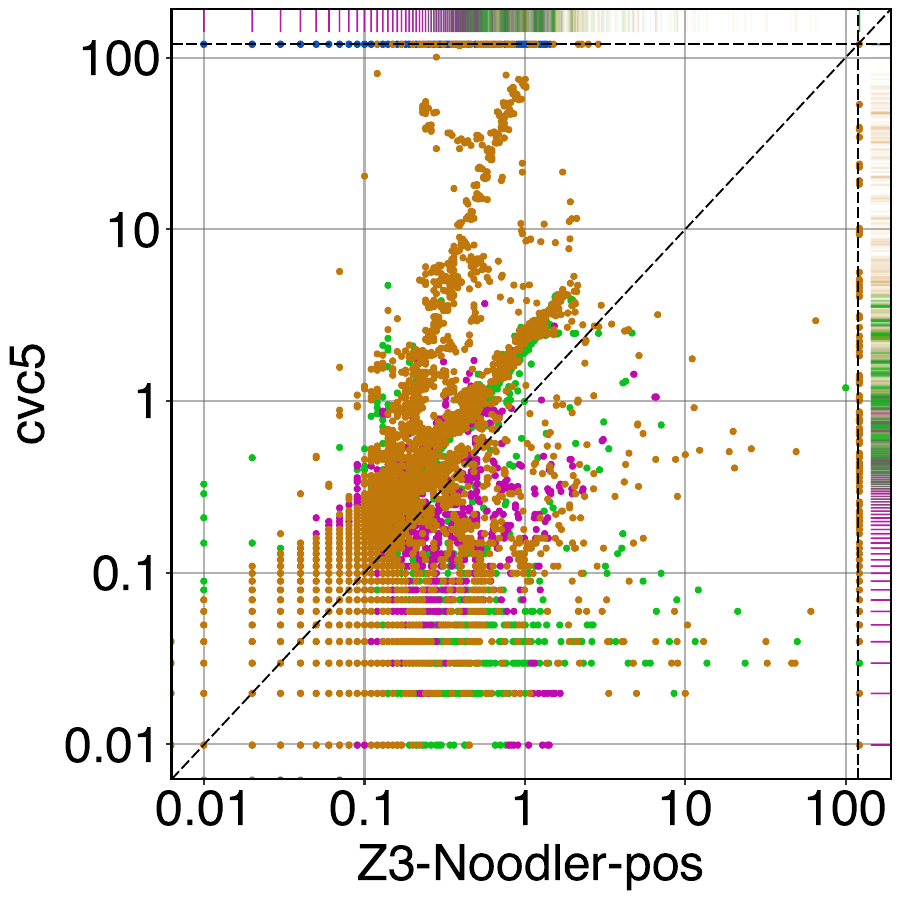}
    \vspace{-6mm}
    \caption{\footnotesize \ziiinoodlerpos vs. \cvcv}
    \label{fig:cvc5}
\end{subfigure}
~
\begin{subfigure}{0.3\textwidth}
    \centering
    \includegraphics[width=\linewidth]{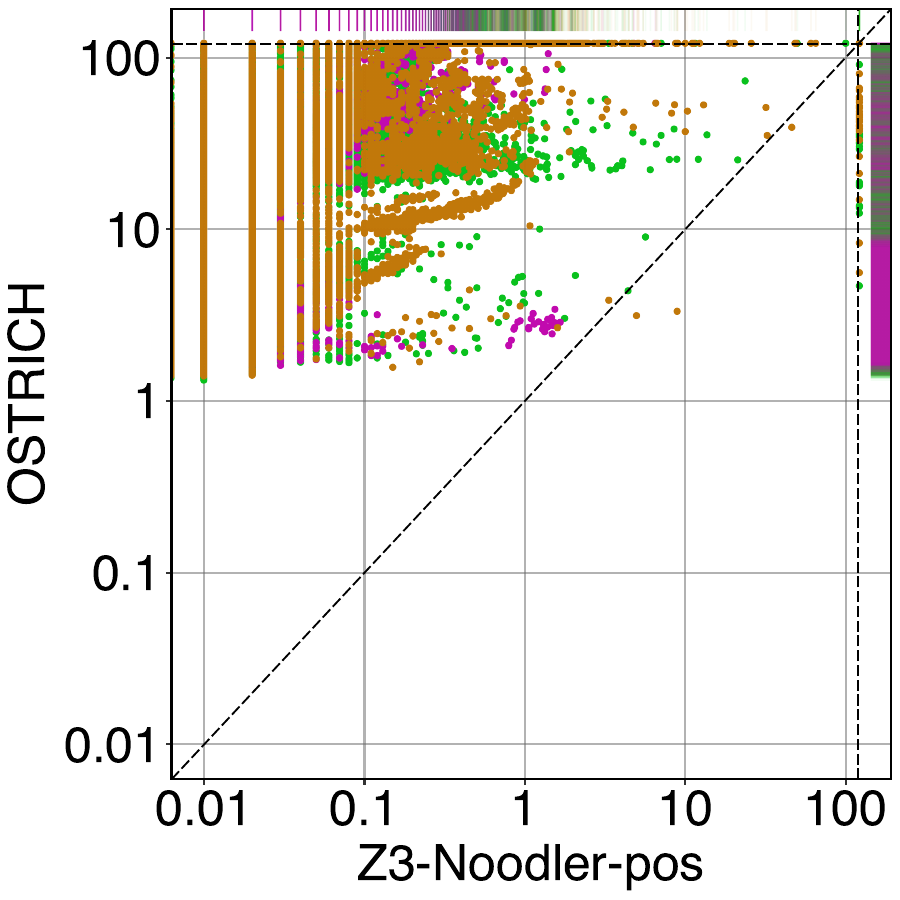}
    \vspace{-6mm}
    \caption{\footnotesize \ziiinoodlerpos\,vs.\,\ostrich}
    \label{fig:ostrich}
\end{subfigure}

\vspace*{-3mm}
\caption{
    Comparison of \ziiinoodlerpos with \ziiinoodler, \cvcv, \ziii, and \ostrich.
    Times are in seconds, axes are logarithmic.
    Dashed lines represent timeouts (120\,s).
    Colours distinguish benchmarks:
    \RGBcircle{192,106,29}\,\biopython,
    \RGBcircle{15,195,35}\,\django, 
    \RGBcircle{195,16,176}\,\thefuck, and 
    \RGBcircle{14,77,171}\,\poshard.
}
\label{fig:scatter}
%\vspace*{-6mm}
\end{figure}
}

\newcommand{\figCactus}[0]{
\begin{figure}[t!]%
\centering
\includegraphics[width=\linewidth]{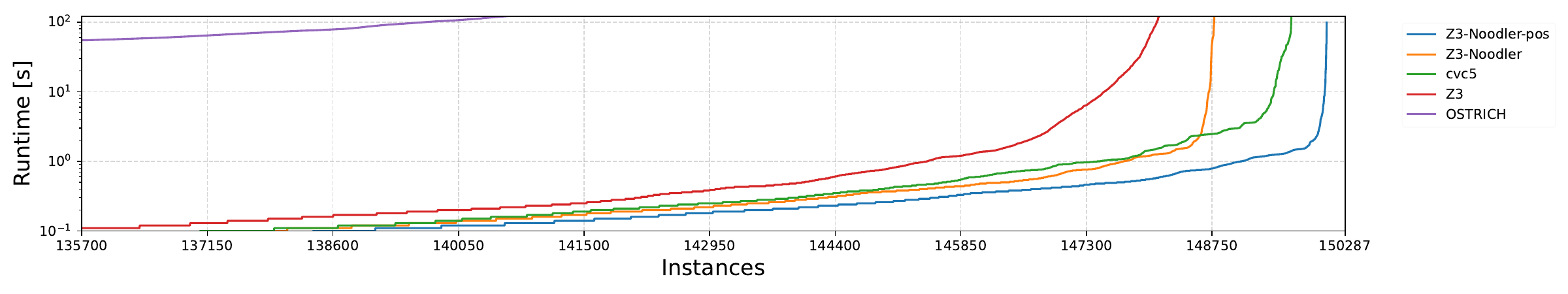}
\vspace*{-8mm}
\caption{
    Cactus plot comparing sorted runtimes of \ziiinoodlerpos with other tools.
    The $y$-axis denotes the time in seconds (the axis is logarithmic), the
    $x$-axis denotes the number of solved formulae ordered by their runtime (we
    show only the $\sim$14,500 hardest formulae for each solver).
}
\label{fig:cactus}
%\vspace*{-6mm}
\end{figure}
}

We implemented the proposed decision procedure in the \ziiinoodler solver version 1.3~\cite{ChenCHHLS24}.
We call the modified version \ziiinoodlerpos. 
Since we need a~monadic decomposition for dealing with position constraints, the
proposed decision procedure was integrated to the stabilization-based
procedure~\cite{ChenCHHLS23}, which computes the monadic decomposition.
For an input formula, we separate the position constraints, 
apply the stabilization-based procedure on the remaining constraints, and for
each of the possible obtained monadic decompositions (there might be more
depending on the case-splits within the stabilization), we add the LIA formula
describing satisfiability of those position constraints to the LIA formula
provided by the stabilization-based procedure (this LIA formula may contain
additional subformulae speaking, e.g., about lengths of the solution or
string-integer conversions~\cite{noodler-conv}).
For satisfiability checking of quantifier-free LIA formulae, \ziiinoodler uses
\ziii's internal LIA solver based on the Simplex method extended with
a~branch-and-cut strategy to obtain integer solutions~\cite{ziii}.
For universally quantified formulae (those obtained by the reduction of
$\notcontains$), we use an additional \ziii's internal solver based on the
\emph{model-based quantifier instantiation}~\cite{mbqi} approach.

For representing the tag automaton structure, we use the \mata
library~\cite{ChocholatyFHHHLS24}.
A tag automaton is represented 
as a nondeterministic finite automaton with additional mapping of \mata's integer symbols to sets of tags.
The LIA formula is generated in \ziiinoodler's internal format, which is then converted 
to \ziii's formula representation. Except of the proposed procedure, \ziiinoodlerpos implements 
heuristics for simple cases of the $\notcontains$ predicate. In particular, if $|u| < |v|$, then $\notcontains(u,v)$ is 
satisfied. Therefore, if we get a $\notcontains$ with not-flat languages, we apply this 
underapproximation. For the case when the language of $v$ is finite, we enumerate words from the 
language and check them separately instead of generating complex quantified formulae.

% We use \ziii's built-in LIA solver as the back-end, which is based on the Simplex
% method extended with a~branch-and-cut strategy to obtain integer solutions~\cite{ziii}
% to solve quantifier-free formulae. Support for quantifiers is implemented 
% using the Model-Based Quantifier Instantiation~\cite{mbqi} approach. 
% \mh{
%     We have implemented reductions for a general system of disequalities
%     (\cref{sec:diseqSystem}) and the $\notcontains$ predicate (\cref{sec:notcontains}).
%     The implementation
%     also uses the more specialized, but simpler, construction in the case that the input formula
%     contains only a single disequality
%     (\cref{sec:single-gen-diseq}).
%     Although constructions for the $\notsuffix$ and the $\notprefix$ predicates
%     are very similar to that for disequalities, we do not implement these as they would
%     require significant modifications to \ziiinoodler's preprocessing rules.
% }

%%%%%%%%%%%%%%%%%%%%%%%%%%%%%%%%%%%%%%%%%%%%%%%%%%%%%%%%%%%%%%%%
\newcommand{\tabResults}[0]{
\begin{table}[t]
  \caption{
  Results of experiments on all benchmarks. For each benchmark we give the number of 
  cases the tool runs out of resources (column ``\oob'', the number of timeouts and memory outs), the number of 
  unknowns (column ``\unknown''), and total time in seconds on finished instances (column ``\timet'')
  \begin{changebar}
  and on all instances (i.e., when taking the time 120\,s for \oob/\unknown instances; column ``\timeall'').
  The best \timeall results are \textbf{bold}.
  \end{changebar}
  }
  \vspace*{-4mm}
  \label{tab:solved}
  \resizebox{\linewidth}{!}{%
  \begin{changebar}
  \input{tables/solved2.tex}  \end{changebar}
  }
\end{table}
}

\begin{table}[t]
  \caption{
  Results of experiments on all benchmarks. For each benchmark we give the number of 
  cases the tool runs out of resources (column ``\oob'', the number of timeouts and memory outs), the number of 
  unknowns (column ``\unknown''), and total time in seconds on finished instances (column ``\timet'')
  \begin{changebar}
  and on all instances (i.e., when taking the time 120\,s for \oob/\unknown instances; column ``\timeall'').
  The best \timeall results are \textbf{bold}.
  \end{changebar}
  }
  \vspace*{-4mm}
  \label{tab:solved}
  \resizebox{\linewidth}{!}{%
  \begin{changebar}
  \input{tables/solved2.tex}  \end{changebar}
  }
\end{table}
  %%%%%%%%%%%%%%%%%

%*******************************************************************************
\vspace{-0.0mm}
\subsection{Experimental Settings}
\vspace{-0.0mm}
%*******************************************************************************
%
We evaluated \ziiinoodlerpos on benchmarks containing heavy position constraints. 
We collected 4 benchmark sets (the number of formulae is in parentheses): 
\begin{inparaenum}[(i)]
  \item \biopython (77,222) obtained by a~symbolic execution of Python tools for bioinformatics~\cite{Notsubstring},
  \item \django (52,643) from a~symbolic execution of the Django web application framework~\cite{Notsubstring}, 
    \cbstart
  \item \thefuck (19,872) from a~symbolic execution of a tool correcting
    command mistakes~\cite{Notsubstring}\footnote{The three benchmark
    sets \biopython, \django, and \thefuck are from~\cite{Notsubstring}, where
    they were obtained by running the symbolic executor PyCT~\cite{PyCT} on the
    respective projects and keeping formulae that contained at least one
    position string constraint or a~constraint that is naturally translated to
    a~position constraint, such as \texttt{indexof}.}, and 
  \item \poshard (550) containing difficult hand-crafted formulae with
    disequalities and $\notcontains$ predicates.\footnote{These are simple
    formulae inspired by the problem of testing primitiveness of a word. They
    contain one $\notcontains$ or $\neq$ predicate over concatenations of
    string variables (with possible repetitions, e.g., $x y z \neq x x y$)
    constrained by simple regular languages (e.g., $a^*$ or $(abc)^*$). Despite
    their apparent simplicity, a solution cannot be easily found by
    systematically trying different assignments, which seems to be the reason
    why these formulae are unsolvable by state-of-the-art solvers.}
    \cbend
\end{inparaenum}
Altogether, we collected 150,287 formulae for evaluation.
\cbstart
We note that these formulae (in particular the first three sets) are \emph{from
the wild}, i.e., they are not of the form $\regconstr' \land \intconstr \land
\diseqconstr$, which we consider in \cref{sec:diseq,sec:other_pos}.
Instead, they have a~more general Boolean structure with word equations
and other constraints (not even necessarily chain-free), which are (in our
case) handled and transformed into the monadic decomposition and our input form
by the stabilization-based procedure of \ziiinoodler.
Moreover, a~single input formula might cause multiple calls to the
string solver with formulae from different fragments.
Since the main goal of this evaluation is to show the improvement that our
contribution brings to automata-based techniques being able to handle position
constraints,
we excluded benchmarks from \smtlib~\cite{SMTLIB}, where
the number of position-heavy constraints is small.
\cbend

We compared \ziiinoodlerpos with 
\ziiinoodler~(version~1.3)~\cite{ChenCHHLS24}, \cvcv~(version~1.2.0)~\cite{cvc5}, \ziii~(version~4.13.3)~\cite{ziii}, and 
\ostrich~(version~1.4)~\cite{ChenHLRW19}. We excluded \ziiitrau~\cite{Trau} as it gives incorrect results on some benchmarks and \ziiialpha~\cite{ijcai2024p211} as it 
fails with an error on a~large fraction of the benchmark.
The experiments were executed on a server with an AMD EPYC 9124 64-Core
Processor (16~cores were used by our experiment) with 125\,GiB of RAM running
Ubuntu~22.04.5.
\cbstart
The timeout was set to 120\,s (from our experience, higher limit has only
a~negligible effect on the number of solved instances) and the memory limit was
\cbend
set to 8\,GiB (except for \ostrich, where we set the limit to 16\,GiB, since
\ostrich refuses to run with less than 8\,GiB of memory).

%*******************************************************************************
\vspace{-0.0mm}
\subsection{Results}\label{sec:label}
\vspace{-0.0mm}
%*******************************************************************************

\figScatter

\cbstart
The overall results comparing \ziiinoodlerpos with other tools are shown in \cref{tab:solved}. From the table 
you can see that \ziiinoodlerpos has the smallest number (210) of \oob{}s (i.e., time/memory-outs) followed by \ziiinoodler (443; however, it answers \unknown for 1,065 instances) and \cvcv~(619). 
\ziii and \ostrich have a bit more \oob{}s than these tools (2,146 and 8,907 respectively).
For the \biopython and \django benchmark sets, \cvcv is the best solver, having slightly less \oob{}s than 
\ziiinoodlerpos (171 vs.\ 69 on \biopython and 39 vs.\ 0 on \django, which is
$\sim$0.1\,\% of the two benchmark sets).
On the \thefuck set, the overall performance of \ziiinoodlerpos and \cvcv is roughly the same
(\ziiinoodlerpos is faster by 25\,s on the whole benchmark set, which is
negligible), though the performance on individual formulae in the set can
differ by a~lot (cf.\ \cref{fig:cvc5}).
On the \poshard benchmark, \ziiinoodlerpos can solve all formulae while no
other solver except \ziiinoodler can solve any of them (and \ziiinoodler can solve only 70).
% \ziiinoodlerpos has a bit more \oob{}s than \cvcv on \django (171 vs.\ 69) and \biopython (39 vs.\ 0), but on hard position constraints 
% it behaves much better than \cvcv, which was not able to solve a single instance. 
Note that concerning unsolved instances, \ziiinoodlerpos is quite orthogonal to \cvcv (only 10 formulae can be solved neither by \ziiinoodlerpos nor \cvcv).
Regarding the overall 
running time, \ziiinoodlerpos has the smallest time of all other tools.
\cbend

From a comparison of \ziiinoodlerpos and \ziiinoodler, 
it is evident that the proposed decision procedure significantly helps in solving instances 
that were unknown for the original \ziiinoodler without any performance regression. 
The \oob{}s of \ziiinoodlerpos on benchmarks obtained from symbolic execution are caused 
mainly by non-chain-freeness of the input constraint where the stabilization-based procedure was not able to get a stable solution before the time limit.
In \cref{fig:scatter} we show scatter plots comparing the performance of \ziiinoodlerpos
with \ziii, \cvcv, and \ostrich. It can be seen from the figures that \ziiinoodlerpos 
can significantly outperform other state-of-the-art solvers on many instances. 
\cbdelete
% In particular, \ziiinoodlerpos has both the lowest average time (0.05\,s), followed by 
% \ziiinoodler (0.06\,s), \cvcv (0.12\,s), and \ostrich (13.08\,s).
% Moreover, \ziiinoodlerpos also has same lowest median as \ziiinoodler and \ziii (0.01\,s).
In \cref{fig:cactus}
we give a cactus plot comparing sorted running times on all benchmarks, showing 
the superior performance of \ziiinoodlerpos.

\figCactus

\vspace{-0.0mm}
\section{Related Work}\label{sec:related}
\vspace{-0.0mm}
%%%%%%%%%%%%%%%%%%%%%%%%%%%%%%%%%%%%%%%%%%%%%%%%%%%%%%%%%%%%%%%%%%%%%%%%%%%%%%%%

Approaches and tools for string solving are numerous and diverse, with a variety of constraint representations, algorithms, and input types. Many approaches use automata, e.g., \stranger~\cite{Stranger,fmsd14,yu2011}, %is this right?
\ziiinoodler~\cite{ChenCHHLS24,chen2023solving,BlahoudekCCHHLS23},
\norn~\cite{AutomataSplitting,Norn},
\ostrich~\cite{AnthonyTowards2016,AnthonyReplaceAll2018,ChenHHHLRW20,ChenFHHHKLRW22,ChenHHHLRW20},
\trau~\cite{ChainFree,Trau,Flatten,Notsubstring},
\sloth~\cite{holik_string_2018},
\slog~\cite{fang-yu-circuits},
% Slent~\cite{slent},
\ziiistriiire~\cite{Z3str3RE,BerzishDGKMMN23},
\retro~\cite{ondravojtastrings20,ChenHLT23}.
% ABC~\cite{ABCpaper,ABCtool},
% Qzy~\cite{arlen},
% BEK~\cite{BEK},
% and~\cite{ZhuAM19}.
The most important tools focused on word equations include \cvc{}4/5~\cite{cvc4_string14,tinelli-fmsd16,tinelli-hotsos16,tinelli-frocos16,cvc417,cvc422,cvc420},
\ziii~\cite{BTV09,z3}.
% S3~\cite{S3},
% \kepler~\cite{LeH18}.
% StrSolve~\cite{HW12},
% and Woorpje~\cite{DayEKMNP19}.
Bit vectors are commonly used in tools like Z3Str/2/3/4~\cite{Z3-str,Z3Str3,MoraBKNG21,Z3-str15} and HAMPI~\cite{HAMPI},
while PASS~\cite{PASS} utilizes arrays, and G-strings~\cite{gstrings} and GECODE+S~\cite{gecode+s} use a SAT solver. \ziiialpha~\cite{ijcai2024p211}
synthesizes efficient strategies for \ziii in order to improve the performace.

The chain-free fragment~\cite{ChainFree}, which we extend in this paper, represents the largest fragment of string constraints for which any string solver offers formal completeness guarantees.
Quadratic equations, addressed by tools like \retro~\cite{ondravojtastrings20,ChenHLT23} and \kepler~\cite{LeH18}, are incomparable but have less practical relevance, though some tools, such as \ziiinoodler or \ostrich, implement Nielsen's algorithm~\cite{nielsen1917} to handle quadratic cases.
Most other solvers guarantee completeness on smaller fragments (e.g., \ostrich~\cite{AnthonyTowards2016}, \norn~\cite{AutomataSplitting,Norn}, and \ziiistriiire~\cite{Z3str3RE}), or use incomplete heuristics that work in practice by over-/under-approximating or by sacrificing termination guarantees.

% Typically, string solvers avoid processing regular expressions directly, instead deferring them or abstracting them into arithmetic/length constraints, as seen in tools like \trau, \ziiistriiire, \ziiistriv, \cvc{}4/5, and S3.
% Z3-Noolder~\cite{BlahoudekCCHHLS23,chen2023solving} is an exception, employing an approach that directly addresses regular expressions through a carefully designed automata-based algorithm. This approach builds upon the automata-based algorithm initially proposed in~\cite{AutomataSplitting}, which is used, at least in part, by several solvers, including
% \norn~\cite{AutomataSplitting,Norn,ChainFree},
% \trau~\cite{abdulla_efficient_2020,Trau,Flatten,Notsubstring},
% \ostrich~\cite{AnthonyTowards2016,AnthonyReplaceAll2018,AnthonyComplex2019},
% and \ziiistriiire~\cite{Z3str3RE,BerzishDGKMMN23}.

When it comes to handling position constraints, existing tools generally employ a similar approach---reducing these constraints to equations and length constraints, which are then solved using exponential-space algorithms or incomplete techniques in modern string solvers. However, this approach cannot be even used for $\notcontains$ as it cannot be directly reduced to quantifier-free 
combination of equations and length constraints. \cvc-4/5 transforms $\notcontains$ to quantified string formula, which is then solved by quantifier instantiation~\cite{ReynoldsHighlevel19}. 
Probably the closest approach to ours is~\cite{Notsubstring} converting the $\notcontains$ into a LIA formula. The main differences are threefold:
\begin{inparaenum}[(i)]
  \item the approach of~\cite{Notsubstring} builds on the flattening underapproximation, while our approach is precise. 
  \item our framework is more general that we can reduce to LIA all combinations of position constraints and not just $\notcontains$.
\item
\cbstart
The approach in~\cite{Notsubstring} avoids
\cbend
considering repetitions of variables, which is a central part of our work, by an aggressive overapproximation based on replacing repeating variables by fresh ones.  
\end{inparaenum} 
%
%Our work introduces a key innovation: a unified and direct algorithm for handling position constraints. 
%
The idea of using counting to determine positions in strings was also used in \cite{ChenHHHLRW20}, where cost enriched automata similar to our tag automata were used, though \cite{ChenHHHLRW20} aspires only to solve a substantially simpler problem of computing pre-images of basic constraints and does not consider position constraints.  
The inspiration for our use of tag automata was originally drawn  
from methods used in functional equivalence checking of streaming string transducers~\cite{AlurC11}. 
%, our technique establishes a new theoretical foundation that surpasses existing methods in terms of efficiency, generality, and completeness guarantees.

% This often results in unnecessary complexity and a loss of the high-level information inherent in position constraints. Our work introduces a key innovation: a unified and direct algorithm for handling position constraints. Drawing inspiration from methods used in functional equivalence checking of streaming string transducers~\cite{AlurC11}, this technique establishes a new theoretical foundation that surpasses existing methods in terms of efficiency, generality, and completeness guarantees.

% Altogether, a comprehensive overview of techniques is challenging due to the breadth and rapid evolution of the field. The mentioned solvers employ a plethora of heuristics and implementation techniques, accumulated over extensive publication histories, that interact and influence one another. Consequently, comparisons beyond overall tool performance are difficult to make.

\cbstart
%%%%%%%%%%%%%%%%%%%%%%%%%%%%%%%%%%%%%%%%%%%%%%%%%%%%%%%%%%%%%%%%%%%%%%%%%%%%%%%%
\vspace{-0.0mm}
\section*{Acknowledgements}\label{sec:label}
\vspace{-0.0mm}
%%%%%%%%%%%%%%%%%%%%%%%%%%%%%%%%%%%%%%%%%%%%%%%%%%%%%%%%%%%%%%%%%%%%%%%%%%%%%%%%

We thank the anonymous reviewers for careful reading of the paper and their
suggestions that greatly improved its quality.
This work was supported by 
the Czech Ministry of Education, Youth and Sports ERC.CZ project LL1908,
the Czech Science Foundation project 25-18318S, and
the FIT BUT internal project FIT-S-23-8151.
\ackPhdTalent

\cbend

%%%%%%%%%%%%%%%%%%%%%%%%%%%%%%%%%%%%%%%%%%%%%%%%
\bibliography{literature.bib}

\ifTR

\appendix
\clearpage

% some bug in cleveref
\crefalias{section}{appendix}

\input{appendix.tex}

\fi

\end{document}

%% file: macros.tex
\newcommand{\ol}[1]{\textcolor{blue}{\ifmmode \text{[OL: #1]}\else [OL: #1] \fi}}
\newcommand{\mh}[1]{\textcolor{violet}{\ifmmode \text{[MH: #1]}\else [MH: #1] \fi}}
\newcommand{\lh}[1]{\textcolor{pink}{\ifmmode \text{[LH: #1]}\else [LH: #1] \fi}}
\newcommand{\js}[1]{\textcolor{purple}{\ifmmode \text{[JS: #1]}\else [JS: #1] \fi}}
\newcommand{\vh}[1]{\textcolor{green!70!black}{\ifmmode \text{[VH: #1]}\else [VH: #1] \fi}}
\newcommand{\yfc}[1]{\textcolor{blue!70!black}{\ifmmode \text{[YFC: #1]}\else [YFC: #1] \fi}}

\newcommand{\defiff}[0]{\stackrel{\text{\tiny def.}}{\Leftrightarrow}}
\newcommand{\defequiv}[0]{\defiff}
\newcommand{\defequal}[0]{\stackrel{\text{\tiny def.}}{=}}

\newcommand{\alphabet}[0]{\Gamma}

\newcommand{\col}[0]{\mathcal{C}}

\newcommand{\aut}[0]{A}
\newcommand{\but}[0]{B}

\newcommand{\autcon}[0]{\aut_\circ}

\newcommand{\leftsymb}[0]{\mathcal{L}}
\newcommand{\rightsymb}[0]{\mathcal{R}}

\newcommand{\singlesimplediseq}[0]{\mathrm{I}}
\newcommand{\singlediseq}[0]{\mathrm{II}}
\newcommand{\multidiseq}[0]{\mathrm{III}}
\newcommand{\combination}[0]{\mathit{comb}}

\newcommand{\bigLeft}{\mathopen{}\mathclose\bgroup\left}
\newcommand{\bigRight}{\aftergroup\egroup\right}

\newcommand{\bigo}[1]{\mathcal{O}\bigLeft(#1\bigRight)}
\newcommand{\bigtheta}[1]{\Theta\bigLeft(#1\bigRight)}

\newcommand{\ltr}[1]{\xrightarrow{#1}}
\newcommand{\run}[0]{\mathcal R}

\newcommand{\num}[0]{\#}
\newcommand{\numof}[1]{{\num{}#1}}
\newcommand{\autass}{\mathsf{Aut}}

\makeatletter
\DeclareFontEncoding{LS1}{}{}
\makeatother
\DeclareFontSubstitution{LS1}{stix}{m}{n}
\DeclareSymbolFont{stixletters}{LS1}{stix}{m}{it}
\DeclareSymbolFont{stixoperators}{LS1}{stix}{m}{n}

\DeclareMathSymbol{\othersubseteq}{\mathrel}{stixletters}{"D3}
\DeclareMathSymbol{\otherin}{\mathrel}{stixoperators}{"CB}
\DeclareMathSymbol{\othereq}{\mathrel}{stixoperators}{`=}

\newcommand{\balance}{\mathtt{balance}}
\newcommand{\concat}{\mathbin{\circ}}

\newcommand{\concatxof}[1]{\concat_{#1}}
\newcommand{\concateps}[0]{\mathbin{\concat_\epsilon}}

\newcommand{\langof}[1]{\lang(#1)}

\renewcommand{\prod}{\mathit{Product}}

\newcommand{\var}{\mathsf{var}}

\newcommand{\vars}{\mathbb X}
\newcommand{\intvars}[0]{\mathbb{I}}

\newcommand{\ziiinoodler}{\textsc{Z3-Noodler}\xspace}
\newcommand{\cvc}{\textsc{cvc}\xspace}
\newcommand{\cvcv}{\cvc{}5\xspace}

\newcommand{\ziii}{\textsc{Z3}\xspace}
\newcommand{\ziiistriiire}{\textsc{Z3str3RE}\xspace}

\newcommand{\ziiitrau}{\textsc{Z3-Trau}\xspace}
\newcommand{\kepler}[0]{\texttt{Kepler}$_{\mathtt{22}}$\xspace}
\newcommand{\ostrich}[0]{\textsc{OSTRICH}\xspace}
\newcommand{\retro}[0]{\textsc{Retro}\xspace}
\newcommand{\sloth}[0]{\textsc{Sloth}\xspace}

\newcommand{\stranger}[0]{\textsc{Stranger}\xspace}
\newcommand{\trau}[0]{\textsc{Trau}\xspace}
\newcommand{\norn}[0]{\textsc{Norn}\xspace}
\newcommand{\slog}{\textsc{Slog}\xspace}

\newcommand{\ziiinoodlerpos}[0]{\textsc{Z3-Noodler-pos}\xspace}
\newcommand{\mata}[0]{\textsc{Mata}\xspace}
\newcommand{\ziiialpha}[0]{\textsc{Z3-Alpha}\xspace}

\newcommand{\biopython}{\textsf{biopython}\xspace}
\newcommand{\django}{\textsf{django}\xspace}
\newcommand{\thefuck}{\textsf{thefuck}\xspace}
\newcommand{\poshard}{\textsf{position-hard}\xspace}
\newcommand{\oob}{\textsf{OOR}\xspace}
\newcommand{\unknown}{\textsf{Unk}\xspace}
\newcommand{\timet}{\textsf{Time}\xspace}
\newcommand{\timeall}{\textsf{TimeAll}\xspace}

\newcommand{\intterm}[0]{t_i}
\newcommand{\stringterm}[0]{t_s}
\newcommand{\intvar}[0]{x_i}
\newcommand{\stringvar}[0]{x_s}

\newcommand{\bigOOf}[1]{\mathcal{O}(#1)}

\newcommand{\len}[0]{\mathit{len}}
\newcommand{\lenof}[1]{\len(#1)}
\newcommand{\limpl}[0]{\mathrel{\Rightarrow}}
\newcommand{\liff}[0]{\mathrel{\Leftrightarrow}}

\newcommand{\assgn}[0]{\nu}

\newcommand{\parikh}[0]{\mathit{PI}}

\newcommand{\parikhofrun}[1]{\parikh_{#1}}
\newcommand{\equalword}[2]{\mathit{EqualWords}(#1, #2)}  % Parikh images agree
\newcommand{\parikhform}[0]{\mathit{PF}}
\newcommand{\parikhformof}[1]{\parikhform(#1)}
\newcommand{\parikhformtag}[0]{\parikhform_{\mathit{tag}}}
\newcommand{\parikhformtagof}[1]{\parikhformtag(#1)}

\newcommand{\lenDiff}[1]{\mathit{LenDiff}(#1)}  % Difference in length of the left-hand side and right-hand-side

\newcommand{\dsof}[1]{\langle#1\rangle}
\newcommand{\tags}[0]{\mathbb{T}}
\newcommand{\tagmis}[0]{\textbf{\textsf{M}}}
\newcommand{\tagmisone}[0]{\textbf{\textsf{M}}_1}
\newcommand{\tagmistwo}[0]{\textbf{\textsf{M}}_2}
\newcommand{\tagpos}[0]{\textbf{\textsf{P}}}
\newcommand{\taglen}[0]{\textbf{\textsf{L}}}
\newcommand{\tagreg}[0]{\textbf{\textsf{R}}}
\newcommand{\tagcop}[0]{\textbf{\textsf{C}}}
\newcommand{\tagsym}[0]{\textbf{\textsf{S}}}
\newcommand{\tagposone}[0]{\textbf{\textsf{P}}_1}
\newcommand{\tagpostwo}[0]{\textbf{\textsf{P}}_2}
\newcommand{\tagposthree}[0]{\textbf{\textsf{P}}_3}

\newcommand{\tagof}[1]{\langle #1 \rangle}
\newcommand{\tagmisoneof}[1]{\tagof{\tagmisone, #1}}
\newcommand{\tagmistwoof}[1]{\tagof{\tagmistwo, #1}}
\newcommand{\tagsymof}[1]{\tagof{\tagsym, #1}}
\newcommand{\tagposof}[1]{\tagof{\tagpos, #1}}
\newcommand{\taglenof}[1]{\tagof{\taglen, #1}}
\newcommand{\tagposoneof}[1]{\tagof{\tagposone, #1}}
\newcommand{\tagpostwoof}[1]{\tagof{\tagpostwo, #1}}
\newcommand{\tagposthreeof}[1]{\tagof{\tagposthree, #1}}

\newcommand{\formdiseqlen}[0]{\varphi_{\mathit{len}}}

\newcommand{\formdiseqposij}[0]{\varphi_{\mathit{pos}(i,j)}}
\newcommand{\formdiseqsym}[0]{\varphi_{\mathit{sym}}}
\newcommand{\formdiseqmis}[0]{\varphi_{\mathit{mis}}}

\newcommand{\eqconstr}[0]{\mathcal{E}}
\newcommand{\intconstr}[0]{\mathcal{I}}
\newcommand{\regconstr}[0]{\mathcal{R}}
\newcommand{\diseqconstr}[0]{\mathcal{P}}

\newcommand{\formatom}[0]{\varphi_{\mathit{atom}}}

\newcommand{\transvarof}[1]{\numof{#1}}      % variable for transition
\newcommand{\initvarof}[1]{\gamma^I_{#1}}    % variable for initiality of state
\newcommand{\finalvarof}[1]{\gamma^F_{#1}}   % variable for finality of state
\newcommand{\formkirchhoff}[0]{\varphi_{\mathit{Kirch}}}
\newcommand{\forminitstates}[0]{\varphi_{\mathit{Init}}}
\newcommand{\formfinalstates}[0]{\varphi_{\mathit{Fin}}}
\newcommand{\formspan}[0]{\varphi_{\mathit{Span}}}

\newcommand{\spanningvarof}[1]{\sigma_{#1}}
\newcommand{\posregsat}[0]{\textsc{PosRegSAT}\xspace}

\newcommand{\clNP}[0]{\textsc{NP}\xspace}
\newcommand{\clNPhard}[0]{\textsc{NP-hard}\xspace}
\newcommand{\clPSPACE}[0]{\textsc{PSpace}\xspace}

\newcommand{\clEXP}[0]{\textsc{Exp}\xspace}
\newcommand{\clNEXPTIME}[0]{\textsc{NExpTime}\xspace}
\newcommand{\cltwoEXPSPACE}[0]{\textsc{2-ExpSpace}\xspace}
\newcommand{\clP}[0]{\textsc{PTime}\xspace}

\newcommand{\lang}{L}

\newcommand{\RGBcircle}[1]{\ensuremath{{\color[RGB]{#1}\bullet}}}

\newcommand{\smtlib}[0]{\texttt{SMT-LIB}\xspace}

\newcommand{\prefix}[0]{\mathit{prefixof}}
\newcommand{\prefixof}[2]{\prefix(#1, #2)}
\newcommand{\suffix}[0]{\mathit{suffixof}}
\newcommand{\suffixof}[2]{\suffix(#1, #2)}
\newcommand{\contains}[0]{\mathit{contains}}
\newcommand{\containsof}[2]{\contains(#1, #2)}

\newcommand{\strat}[0]{\mathit{str.at}}
\newcommand{\stratof}[2]{\strat(#1, #2)}
\newcommand{\notstrat}[0]{\neg \strat}

\newcommand{\notsuffix}[0]{\neg\suffix}
\newcommand{\notprefix}[0]{\neg\prefix}

\newcommand{\notcontains}{\neg \contains}
\newcommand{\notcontainsof}[2]{\notcontains(#1, #2)}

\newcommand{\nat}[0]{\mathbb{N}}

\newcommand{\integers}[0]{\mathbb{Z}}

\newcommand{\offset}[0]{\kappa}

\newcommand{\tagAut}[0]{\aut_{\mathit{tag}}}

\newcommand{\lentag}[0]{\mathsf{LenTag}}
\newcommand{\notcontainsLIA}[0]{\varphi^{\mathit{NC}}}

\newcommand{\fFair}[0]{\varphi_{\mathit{Fair}}}

\newcommand{\fConsist}[0]{\varphi_{\mathit{Consistent}}}

\newcommand{\fParikh}[0]{\varphi_{\mathit{Parikh}}}

\newcommand{\fCopies}[0]{\varphi_{\mathit{Copies}}}

\newcommand{\fMismatch}[1]{\varphi^{#1}_{\mathit{Mismatch}}}
\newcommand{\fLen}[1]{\varphi^{#1}_{\mathit{Len}}}
\newcommand{\fPredSat}[1]{\varphi^{#1}_{\mathit{Sat}}}

\newcommand{\symbof}[1]{#1}

\newcommand{\letter}[1]{\ensuremath{#1}}
\newcommand{\ackPhdTalent}[0]{
\noindent
The work of Michal Hečko, a~Brno Ph.D.\ Talent Scholarship
\raisebox{-6pt}{\protect\includegraphics[height=17pt]{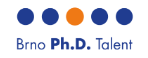}}
Holder, is funded by the Brno City Municipality.\xspace
}

%% file: macromove.tex
\usepackage{trimclip}

\makeatletter
\DeclareRobustCommand{\shortto}{%
  \mathrel{\mathpalette\short@to\relax}%
}

\DeclareRobustCommand{\shortminus}{%
  \mathrel{\mathpalette\short@minus\relax}%
}

\newcommand{\short@to}[2]{%
  \mkern2mu
  \clipbox{{.5\width} 0 0 0}{$\m@th#1\vphantom{+}{\rightarrow}$}%
}

\newcommand{\short@minus}[2]{%
  \mkern2mu
  \clipbox{{.5\width} 0 0 0}{$\m@th#1\vphantom{+}{-}$}%
}
\makeatother

\newcommand{\labeledtoCnt}[1]{{{\shortminus}\hspace{-1.0pt}\raisebox{0.16ex}{$\scriptstyle( #1\hspace{-0.28pt})$}\hspace{-1.6pt}{\shortto}}}
\newcommand{\labeledto}[1]{{{\shortminus}\hspace{-2pt}\raisebox{0.16ex}{$\scriptstyle\{ #1\hspace{-0.28pt}\}$}\hspace{-2.2pt}{\shortto}}}

\newcommand{\scriptlabeledto}[1]{{{\shortminus}\hspace{-1.0pt}\raisebox{0.12ex}{$\scriptscriptstyle\{ #1\hspace{-0.28pt}\}$}\hspace{-1.6pt}{\shortto}}}

\newcommand{\scriptlabeledtoCnt}[1]{{{\shortminus}\hspace{-1.0pt}\raisebox{0.12ex}{$\scriptscriptstyle #1\hspace{-0.28pt} $}\hspace{-1.6pt}{\shortto}}}

\newcommand\move[3]{%standard LaTex syntax where each parameter is put in {}
\mathchoice
{#1\,\labeledto{#2}\,#3}
{#1\labeledto{#2}#3}
{#1\scriptlabeledto{#2}#3}
{#1\scriptlabeledto{#2}#3}
}

\newcommand\cnmove[3]{%standard LaTex syntax where each parameter is put in {}
\mathchoice
{#1\,\labeledtoCnt{#2}\,#3}
{#1\labeledtoCnt{#2}#3}
{#1\scriptlabeledtoCnt{#2}#3}
{#1\scriptlabeledtoCnt{#2}#3}
}

%% file: figs/single-diseq-tag-aut.tex
\tikzstyle{myNonSamplingEdge} = [nonSamplingEdge, bend right=38]
\begin{tikzpicture}
    \def\stateDistance{1.7}
    \def\levelDistance{2.0}
    \def\samplingLevelDistance{3.6}
    \def\lastLevelDistance{2.3}
    \def\autDistance{1.7}

    \node[state] (rx0) at (0, 0)                        {$r_x, 1$};
    \node[state] (qx0) at ($(rx0)+(0, -\stateDistance)$) {$q_x, 1$};
    \node[state] (rx1) at ($(rx0)+(\samplingLevelDistance, 0)$) {$r_x, 2$};
    \node[state] (qx1) at ($(rx1)+(0,-\stateDistance)$) {$q_x, 2$};

    \node[state,accepting] (qy0) at ($(qx0)+(-\autDistance,0)$)      {$q_y, 1$};
    \node[state] (ry0) at ($(qy0)+(0, \stateDistance)$) {$r_y, 1$};
    \node[state] (qy1) at ($(qx1)+(+\levelDistance,0)$) {$q_y, 2$};
    \node[state] (ry1) at ($(qy1)+(0, \stateDistance)$) {$r_y, 2$};
    \node[state,accepting] (qy2) at ($(qy1)+(\samplingLevelDistance,0)$) {$q_y, 3$};
    \node[state] (ry2) at ($(qy2)+(0, \stateDistance)$) {$r_y, 3$};

    \draw[line width=0.25mm] ($(qx0.south east)+(0.17,-0.17)$) edge[->] (qx0.south east);

    % Epsilon transitions
    \draw (qx0) edge[epsEdge] (qy0);
    \draw (qx1) edge[epsEdge] (qy1);

    % Non-sampling transitions - aut x
    \draw (qx0)
        edge[myNonSamplingEdge]
        node[tags] (qx0_rx0) {$\tagsymof{\texttt{a}}$ \\ $\taglenof{x}$ \\ $\tagposof{x}$} (rx0);
    \draw (rx0)
        edge[myNonSamplingEdge]
        node[tags] (rx0_qx0) {$\tagsymof{\texttt{b}}$ \\ $\taglenof{x}$ \\ $\tagposof{x}$} (qx0);
    \draw (qx1)
        edge[myNonSamplingEdge]
        node[tags] (qx1_rx1) {$\tagsymof{\texttt{a}}$ \\ $\taglenof{x}$} (rx1);
    \draw (rx1)
        edge[myNonSamplingEdge]
        node[tags] (rx1_qx1) {$\tagsymof{\texttt{b}}$ \\ $\taglenof{x}$} (qx1);

    % Non-sampling transitions - aut y
    \draw (qy0) edge[myNonSamplingEdge] node[tags] (qy0_ry0) {$\tagsymof{\texttt{a}}$ \\ $\taglenof{y}$ } (ry0);
    \draw (ry0) edge[myNonSamplingEdge] node[tags] (ry0_qy0) {$\tagsymof{\texttt{c}}$ \\ $\taglenof{y}$ } (qy0);
    \draw (qy1) edge[myNonSamplingEdge] node[tags] (qy1_ry1) {$\tagsymof{\texttt{a}}$ \\ $\taglenof{y}$ \\ $\tagposof{y}$} (ry1);
    \draw (ry1) edge[myNonSamplingEdge] node[tags] (ry1_qy1) {$\tagsymof{\texttt{c}}$ \\ $\taglenof{y}$ \\ $\tagposof{y}$} (qy1);
    \draw (qy2) edge[myNonSamplingEdge] node[tags] (qy2_ry2) {$\tagsymof{\texttt{a}}$ \\ $\taglenof{y}$} (ry2);
    \draw (ry2) edge[myNonSamplingEdge] node[tags] (ry2_qy2) {$\tagsymof{\texttt{c}}$ \\ $\taglenof{y}$} (qy2);

    % Sampling transitions - aut x
    \def\gxshift{1mm}
    \def\gyshift{2mm}

    \def\angleShift{30}
    \pgfmathsetmacro{\rightInAngle}{180+\angleShift}
    \pgfmathsetmacro{\leftInAngle}{180-\angleShift}
    \def\glooseness{1.0}
    \draw
        (qx0)
        edge[out=0, in=\rightInAngle, looseness=\glooseness, samplingEdge] 
        node[tags, below left, xshift=\gxshift, yshift= \gyshift, text width=13mm, samplingEdgeLabel]
            {$\tagsymof{\texttt{a}}$ \\ $\taglenof{x}$ \\ $\tagmisoneof{\texttt{a}}$}
        (rx1);
    \draw
        (rx0)
        edge[edge, out=0, in=\leftInAngle, looseness=\glooseness, samplingEdge]
        node[tags, above left, xshift=\gxshift, yshift=-\gyshift, text width=13mm, samplingEdgeLabel]
            {$\tagsymof{\texttt{b}}$ \\ $\taglenof{x}$ \\ $\tagmisoneof{\texttt{b}}$}
        (qx1);

    % % Sampling transitions - aut y
    \draw
        (qy1)
        edge[edge, out=0, in=\rightInAngle, looseness=\glooseness, samplingEdge] 
        node[tags, below left, xshift=\gxshift, yshift= \gyshift, text width=13mm, samplingEdgeLabel]
            {$\tagsymof{\texttt{a}}$ \\ $\taglenof{y}$ \\ $\tagmistwoof{\texttt{a}}$}
        (ry2);
    \draw
        (ry1)
        edge[edge, out=0, in=\leftInAngle, looseness=\glooseness, samplingEdge]
        node[tags, above left, xshift=\gxshift, yshift=-\gyshift, text width=13mm, samplingEdgeLabel]
            {$\tagsymof{\texttt{c}}$ \\ $\taglenof{y}$  \\ $\tagmistwoof{\texttt{c}}$}
        (qy2);
    \begin{pgfonlayer}{bg}
        \node[fit=(qx0)(rx0)(qy0)(ry0)(qx0_rx0)(ry0_qy0), draw, dashed, fill=teal!20, rounded corners=1] (copy0_wrapper) {};
        \node[fill=teal!50, draw, rounded corners=2, scale=0.8] at ($(copy0_wrapper.north)+(0, +0.15)$) {Copy 1};
        \node[fill=teal!40, draw, rounded corners=2, scale=0.8] at ($(copy0_wrapper.south)+(0, -0.1)$) {0 mismatches seen};

        \node[fit=(qx1)(rx1)(qy1)(ry1)(rx1_qx1)(qy1_ry1), draw, dashed, fill=blue!20, rounded corners=1] (copy1_wrapper) {};
        \node[fill=blue!50, draw, rounded corners=2, scale=0.8] at ($(copy1_wrapper.north)+(0, +0.15)$) {Copy 2};
        \node[fill=blue!40, draw, rounded corners=2, scale=0.8] at ($(copy1_wrapper.south)+(0, -0.1)$) {1 mismatch seen};

        \node[fit=(qy2)(ry2)(ry2_qy2)(qy2_ry2), draw, dashed, fill=magenta!20, rounded corners=1] (copy2_wrapper) {};
        \node[fill=magenta!50, draw, rounded corners=2, scale=0.8] at ($(copy2_wrapper.north)+(0, +0.15)$) {Copy 3};
        \node[fill=magenta!40, draw, rounded corners=2, scale=0.8] at ($(copy2_wrapper.south)+(0, -0.1)$) {2 mismatches seen};
    \end{pgfonlayer}
\end{tikzpicture}

%% file: figs/single-diseq-tag-aut-gen.tex
\tikzstyle{mySamplingEdgeLabel} = [samplingEdgeLabel, text width=48mm]
\tikzstyle{myTags} = [tags, text width=30mm]
\tikzstyle{myNonSamplingEdge} = [nonSamplingEdge, bend right=11]

\begin{tikzpicture}
    \def\stateDistance{5.0}
    \def\levelDistance{1.6 }
    \def\lastLevelDistance{1.6}
    \def\autDistance{1.6}

    \node[state] (rx0) at (0, 0)                        {$r_x, 1$};
    \node[state]
        (qx0) at ($(rx0)+(\stateDistance, 0)$) {$q_x, 1$};
    \node[state] (rx1) at ($(rx0)+(0,-\levelDistance)$) {$r_x, 2$};
    \node[state] (qx1) at ($(rx1)+(\stateDistance, 0)$) {$q_x, 2$};
    \node[state] (rx2) at ($(rx1)+(0,-\lastLevelDistance)$) {$r_x, 3$};
    \node[state] (qx2) at ($(rx2)+(\stateDistance, 0)$) {$q_x, 3$};

    \node[state,accepting] (qy0) at ($(qx0)+(\autDistance,0)$)      {$q_y, 1$};
    \node[state] (ry0) at ($(qy0)+(\stateDistance, 0)$) {$r_y, 1$};
    \node[state] (qy1) at ($(qy0)+(0,-\levelDistance)$) {$q_y, 2$};
    \node[state] (ry1) at ($(qy1)+(\stateDistance, 0)$) {$r_y, 2$};
    \node[state,accepting] (qy2) at ($(qy1)+(0,-\lastLevelDistance)$) {$q_y, 3$};
    \node[state] (ry2) at ($(qy2)+(\stateDistance, 0)$) {$r_y, 3$};

    \draw[line width=0.25mm] ($(qx0.north)+(0, 0.25)$) edge[->] (qx0.north);

    % Epsilon transitions
    \draw (qx0) edge[epsEdge] (qy0);
    \draw (qx1) edge[epsEdge] (qy1);
    \draw (qx2) edge[epsEdge] (qy2);

    % Non-sampling transitions - aut x
    \draw[>=Stealth] (qx0) edge[myNonSamplingEdge] node[myTags] {$\tagsymof{\letter{a}}$, $\taglenof{x}$, $\tagposoneof{x}$} (rx0);
    \draw[>=Stealth] (rx0) edge[myNonSamplingEdge] node[myTags] {$\tagsymof{\letter{b}}$, $\taglenof{x}$, $\tagposoneof{x}$} (qx0);

    \draw[>=Stealth] (qx1) edge[myNonSamplingEdge] node[myTags] {$\tagsymof{\letter{a}}$, $\taglenof{x}$, $\tagpostwoof{x}$} (rx1);
    \draw[>=Stealth] (rx1) edge[myNonSamplingEdge] node[myTags] {$\tagsymof{\letter{b}}$, $\taglenof{x}$, $\tagpostwoof{x}$} (qx1);

    \draw[>=Stealth] (qx2) edge[myNonSamplingEdge] node[myTags] {$\tagsymof{\letter{a}}$, $\taglenof{x}$, $\tagposthreeof{x}$} (rx2);
    \draw[>=Stealth] (rx2) edge[myNonSamplingEdge] node[myTags] {$\tagsymof{\letter{b}}$, $\taglenof{x}$, $\tagposthreeof{x}$} (qx2);

    % Non-sampling transitions - aut y
    \draw[>=Stealth] (qy0) edge[myNonSamplingEdge] node[myTags] {$\tagsymof{\letter{a}}$, $\taglenof{y}$, $\tagposoneof{y}$} (ry0);
    \draw[>=Stealth] (ry0) edge[myNonSamplingEdge] node[myTags] {$\tagsymof{\letter{c}}$, $\taglenof{y}$, $\tagposoneof{y}$} (qy0);
    \draw[>=Stealth] (qy1) edge[myNonSamplingEdge] node[myTags] {$\tagsymof{\letter{a}}$, $\taglenof{y}$, $\tagpostwoof{y}$} (ry1);
    \draw[>=Stealth] (ry1) edge[myNonSamplingEdge] node[myTags] {$\tagsymof{\letter{c}}$, $\taglenof{y}$, $\tagpostwoof{y}$} (qy1);
    \draw[>=Stealth] (qy2) edge[myNonSamplingEdge] node[myTags] {$\tagsymof{\letter{a}}$, $\taglenof{y}$, $\tagposthreeof{y}$} (ry2);
    \draw[>=Stealth] (ry2) edge[myNonSamplingEdge] node[myTags] {$\tagsymof{\letter{c}}$, $\taglenof{y}$, $\tagposthreeof{y}$} (qy2);

    % Sampling transitions - aut x
    \def\gxshift{0mm}
    \def\gyshift{-0.5mm}
    \def\gxshiftinner{0mm}
    \def\gyshiftinner{-0.5mm}
    \def\angleShift{70}
    \pgfmathsetmacro{\rightInAngle}{90-\angleShift}
    \pgfmathsetmacro{\leftInAngle}{90+\angleShift}
    \def\outAngleShift{-63}
    \pgfmathsetmacro{\rightOutAngle}{-90+\outAngleShift}  % source states on the right (q)
    \pgfmathsetmacro{\leftOutAngle}{-90-\outAngleShift}   % source states on the left  (r)
    \def\glooseness{1.3}
    \draw[>=Stealth]
        (qx0)
        edge[out=\rightOutAngle, in=\rightInAngle, looseness=\glooseness, samplingEdge] 
        node[tags, above right, xshift=\gxshiftinner, yshift=-\gyshiftinner, mySamplingEdgeLabel]
            {$\tagsymof{\letter{a}}$, $\taglenof{x}$, $\tagpostwoof{x}$, $\tagmisoneof{\letter{a}, x}$}
        (rx1);
    \draw[>=Stealth]
        (rx0)
        edge[edge, out=\leftOutAngle, in=\leftInAngle, looseness=\glooseness, samplingEdge]
        node[tags, above left, xshift=-\gxshift, yshift=-\gyshift, mySamplingEdgeLabel]
            {$\tagsymof{\letter{b}}$, $\taglenof{x}$, $\tagpostwoof{x}$, $\tagmisoneof{\letter{b}, x}$}
        (qx1);
    \draw[>=Stealth]
        (qx1)
        edge[edge, out=\rightOutAngle, in=\rightInAngle, looseness=\glooseness, samplingEdge] 
        node[tags, above right, xshift=\gxshiftinner, yshift=-\gyshiftinner, mySamplingEdgeLabel]
            {$\tagsymof{\letter{a}}$, $\taglenof{x}$, $\tagposthreeof{x}$, $\tagmistwoof{\letter{a}, x}$}
        (rx2);
    \draw[>=Stealth]
        (rx1)
        edge[edge, out=\leftOutAngle, in=\leftInAngle, looseness=\glooseness, samplingEdge]
        node[tags, above left, xshift=-\gxshift, yshift=-\gyshift, mySamplingEdgeLabel]
            {$\tagsymof{\letter{b}}$, $\taglenof{x}$, $\tagposthreeof{x}$, $\tagmistwoof{\letter{b}, x}$}
        (qx2);

    % Sampling transitions - aut y
    \draw[>=Stealth]
        (qy0)
        edge[edge, out=\leftOutAngle, in=\leftInAngle, looseness=\glooseness, samplingEdge] 
        node[tags, above left, xshift=-\gxshift, yshift=-\gyshift, text width=13mm, mySamplingEdgeLabel]
            {$\tagsymof{\letter{a}}$, $\taglenof{y}$, $\tagpostwoof{y}$, $\tagmisoneof{\letter{a}, y}$}
        (ry1);
    \draw[>=Stealth]
        (ry0)
        edge[edge, out=\rightOutAngle, in=\rightInAngle, looseness=\glooseness, samplingEdge]
        node[tags, above right, xshift=\gxshift, yshift=-\gyshift, text width=13mm, mySamplingEdgeLabel]
            {$\tagsymof{\letter{c}}$, $\taglenof{y}$, $\tagpostwoof{y}$, $\tagmisoneof{\letter{c}, y}$}
        (qy1);
    \draw[>=Stealth]
        (qy1)
        edge[edge, out=\leftOutAngle, in=\leftInAngle, looseness=\glooseness, samplingEdge] 
        node[tags, above left, xshift=-\gxshift, yshift=-\gyshift, text width=13mm, mySamplingEdgeLabel]
            {$\tagsymof{\letter{a}}$, $\taglenof{y}$, $\tagposthreeof{y}$, $\tagmistwoof{\letter{a}, y}$}
        (ry2);
    \draw[>=Stealth]
        (ry1)
        edge[edge, out=\rightOutAngle, in=\rightInAngle, looseness=\glooseness, samplingEdge]
        node[tags, above right, xshift=\gxshift, yshift=-\gyshift, text width=13mm, mySamplingEdgeLabel]
            {$\tagsymof{\letter{c}}$, $\taglenof{y}$, $\tagposthreeof{y}$, $\tagmistwoof{\letter{c}, y}$}
        (qy2);
\end{tikzpicture}

%% file: figs/run-multi-diseqs.tex
\begin{tikzpicture}
    % Aut x
    \node[state, initial, initial text={}, initial distance=6pt] (qx) at (0, 0) {$q_x$};
    \node[state] (rx) at ($(qx)+(1.3,0)$) {$r_x$};
    \node[state, accepting] (sx) at ($(rx)+(1.3,0)$) {$s_x$};
    \draw (qx) edge[->] node[above] {\letter{a}} (rx);
    \draw (rx) edge[->] node[above] {\letter{b}} (sx);
    % Aut x - bounding box + label
    \node[autLabel] (ax_label) at ($(rx)+(0, 0.5)$) {$A_x$};
    \node[fit=(ax_label)(qx)(rx)(sx), autBoundingBox, inner xsep=2mm, inner ysep=0.45mm, xshift=-1mm] {};

    % Aut y
    \node[state, initial, initial text={}, initial distance=6pt] (qy) at (4, 0) {$q_y$};
    \node[state] (ry) at ($(qy)+(1.3,0)$) {$r_y$};
    \node[state, accepting] (sy) at ($(ry)+(1.3,0)$) {$s_y$};
    \draw (qy) edge[->] node[above] {\letter{a}} (ry);
    \draw (ry) edge[->] node[above] {\letter{c}} (sy);
    % Aut y - bounding box + label
    \node[autLabel] (ay_label) at ($(ry)+(0, 0.5)$) {$A_y$};
    \node[fit=(ay_label)(qy)(ry)(sy), autBoundingBox, inner xsep=2mm, xshift=-1mm, inner ysep=0.38mm, yshift=-0.1mm] {};

    % Aut z
    \node[state, initial, initial text={}, initial distance=6pt] (qz) at (8, 0) {$q_z$};
    \node[state] (rz) at ($(qz)+(1.3,0)$) {$r_z$};
    \node[state, accepting] (sz) at ($(rz)+(1.3,0)$) {$s_z$};
    \draw (qz) edge[->] node[above] {\letter{a}} (rz);
    \draw (rz) edge[->] node[above] {\letter{d}} (sz);
    % Aut x - bounding box + label
    \node[autLabel] (az_label) at ($(rz)+(0, 0.5)$) {$A_z$};
    \node[fit=(az_label)(qz)(rz)(sz), autBoundingBox, inner xsep=1mm, inner ysep=0.45mm, xshift=-1mm] {};

	\node[runState] (r1) at (-0.4, -0.9)       {$(q_x, 1)$};
	\node[runState] (r2) at ($(r1)+(2.5, 0)$) {$(r_x, 1)$};
	\node[runState] (r3) at ($(r2)+(3.2, 0)$) {$(s_x, 2)$};
	\node[runState] (r4) at ($(r3)+(3.1, 0)$) {$(s_x, 3)$};
	\node[runState] (r5) at ($(r4)+(1.7, 0)$) {$(q_y, 3)$};
	\node[runState] (r6) at ($(r5)+(0, -0.9)$)  {$(r_y, 3)$};
	\node[runState] (r7) at ($(r6)+(-3.2, 0)$)  {$(s_y, 4)$};
	\node[runState] (r8) at ($(r7)+(-1.7, 0)$)  {$(q_z, 4)$};
	\node[runState] (r9) at ($(r8)+(-2.5, 0)$)  {$(r_z, 4)$};
	\node[runState, fill=orange!50] (r10) at ($(r9)+(-3.2, 0)$)  {$(s_z, 5)$};

    \draw[runArrows] (r1) edge[->] node[tags] {$\tagsymof{\letter{a}}$ \\ $\taglenof{x}$ \\ $\tagof {\tagpos_1, x}$} (r2);
	\draw[runArrows] (r2) edge[->] node[tags, text width=2.3cm]
        {$\tagsymof{\letter{b}}$, $\taglenof{x}$ \\ $\tagof {\tagpos_2, x}$ \\ $\tagof{ \tagmis_1, x, D_1, {\leftsymb}, \letter{b}}$ } (r3);
    \draw[runArrows] (r3) edge[->] node[tags, text width=2.1cm] {$\tagof{ \tagcop_2, {x}, D_2, \leftsymb}$} (r4);
	\draw[runArrows] (r4) edge[->] node[tags, text width=2mm]
		{$ $} (r5);
	\draw[runArrows] (r5) edge[->, out=-10, in=10, looseness=1.6]
        node[tags, xshift=3mm] {$\tagsymof{\letter{a}}$ \\ $\taglenof{y}$ \\ $\tagof {\tagpos_3, y}$} (r6);

	\draw[runArrows] (r6) edge[->] node[tags, text width=2.3cm]
        {$\tagsymof{\letter{c}}$, $\taglenof{y}$ \\ $\tagof {\tagpos_4, y}$ \\ $\tagof{\tagmis_3, y, D_1, \rightsymb, \letter{c}}$ } (r7);

	\draw[runArrows] (r7) edge[->] node[tags, text width=2mm] {$ $} (r8);
    \draw[runArrows] (r8) edge[->] node[tags] {$\tagsymof{\letter{a}}$ \\ $\taglenof{z}$ \\ $\tagof {\tagpos_4, z}$} (r9);
	\draw[runArrows] (r9) edge[->] node[tags, text width=2.3cm]
        {$\tagsymof{\letter{d}}$, $\taglenof{z}$ \\ $\tagof {\tagpos_5, z}$ \\ $\tagof{\tagmis_4, x, D_2, \rightsymb, \letter{d}} $ } (r10);
\end{tikzpicture}

%% file: figs/diseq-vs-notcontains.tex
\begin{tikzpicture}
\begin{scope}
  \node[cell] (u1_1) {a};
  \node[cell, anchor=west, mismatch] at (u1_1.east) (u1_2) {a};
  \node[cell, anchor=west, mismatch] at (u1_2.east) (u1_3) {b};
  \node[cell, anchor=west] at (u1_3.east) (u1_4) {b};
  \node[cell, anchor=west] at (u1_4.east) (u1_5) {a};
\end{scope}

\begin{scope}
  \node[cell, anchor=north] at ($(u1_1.south)+(0, -0.3)$) (v1_1) {a};
  \node[cell, anchor=west, mismatch] at (v1_1.east) (v1_2) {b};
  \node[cell, anchor=west, mismatch] at (v1_2.east) (v1_3) {a};

  \node[anchor=north] (offset_label1) at ($(u1_3.south)+(0, -1.2)$) {offset $\offset = 0$};
  % \node[anchor=south, scale=0.8] (diseq_label)   at ($(u1_3.north)+(0,  0.1)$) {Disequation $u \neq v$};
\end{scope}

\begin{scope}
  \node[cell, anchor=west] at ($(u1_5.east)+(1,0)$) (u2_1) {a};
  \node[cell, anchor=west] at (u2_1.east) (u2_2) {a};
  \node[cell, anchor=west] at (u2_2.east) (u2_3) {b};
  \node[cell, anchor=west, mismatch] at (u2_3.east) (u2_4) {b};
  \node[cell, anchor=west] at (u2_4.east) (u2_5) {a};
\end{scope}

\begin{scope}
  \node[cell, anchor=north] at ($(u2_2.south)+(0, -0.3)$) (v2_1) {a};
  \node[cell, anchor=west]  at (v2_1.east) (v2_2) {b};
  \node[cell, anchor=west, mismatch]  at (v2_2.east) (v2_3) {a};

  \node[ghostCell, anchor=north] at ($(u2_1.south)+(0, -0.3)$) (b2_1) {$\bullet$};

  \node[anchor=north] (offset_label2) at ($(u2_3.south)+(0, -1.2)$) {offset $\offset = 1$};

  \coordinate (b2_1o) at ($(b2_1.south)+(-0.1, 0.15)$);
  \draw[dashed] (b2_1o) edge[->, bend right] (offset_label2.west);
\end{scope}

\begin{scope}
  \node[cell, anchor=west] at ($(u2_5.east)+(1,0)$) (u3_1) {a};
  \node[cell, anchor=west] at (u3_1.east) (u3_2) {a};
  \node[cell, anchor=west, mismatch] at (u3_2.east) (u3_3) {b};
  \node[cell, anchor=west] at (u3_3.east) (u3_4) {b};
  \node[cell, anchor=west] at (u3_4.east) (u3_5) {a};
\end{scope}

\begin{scope}
  \node[cell, anchor=north, mismatch] at ($(u3_3.south)+(0, -0.3)$) (v3_1) {a};
  \node[cell, anchor=west]  at (v3_1.east) (v3_2) {b};
  \node[cell, anchor=west]  at (v3_2.east) (v3_3) {a};

  \node[ghostCell, anchor=north] at ($(u3_1.south)+(0, -0.3)$) (b3_1) {$\bullet$};
  \node[ghostCell, anchor=west]  at (b3_1.east) (b3_2) {$\bullet$};

  \node[anchor=north] (offset_label3) at ($(u3_3.south)+(0, -1.2)$) {offset $\offset = 2$};

  \coordinate (b3_1o) at ($(b3_1.south)+(-0.1, 0.15)$);
  \coordinate (b3_2o) at ($(b3_2.south)+(-0.1, 0.15)$);

  \draw[dashed] (b3_1o) edge[->, bend right] (offset_label3.west);
  \draw[dashed] (b3_2o) edge[->, bend right] (offset_label3.west);
\end{scope}

%%%%%%%%%%%%%%%%%% Overflow %%%%%%%%%%%%%%%%%%
\begin{scope}
  \node[cell, anchor=west] at ($(u3_5.east)+(1,0)$) (u4_1) {a};
  \node[cell, anchor=west] at (u4_1.east)           (u4_2) {a};
  \node[cell, anchor=west] at (u4_2.east)           (u4_3) {b};
  \node[cell, anchor=west] at (u4_3.east)           (u4_4) {b};
  \node[cell, anchor=west] at (u4_4.east)           (u4_5) {a};
  \node[ghostCell, anchor=west, red, fill=red!10, rounded corners=2] at (u4_5.east)      (u4_6) {$\divideontimes$};
\end{scope}

\begin{scope}
  \node[cell, anchor=north] at ($(u4_4.south)+(0, -0.3)$) (v4_1) {a};
  \node[cell, anchor=west]  at (v4_1.east)                (v4_2) {b};
  \node[cell, anchor=west]  at (v4_2.east)                (v4_3) {a};

  \node[ghostCell, anchor=north] at ($(u4_1.south)+(0, -0.3)$) (b4_1) {$\bullet$};
  \node[ghostCell, anchor=west]  at (b4_1.east) (b4_2) {$\bullet$};
  \node[ghostCell, anchor=west]  at (b4_2.east) (b4_3) {$\bullet$};
  
  \coordinate (mid) at ($(u4_3.south)!0.5!(u4_4.south)$);
  \node[anchor=north] (offset_label4) at ($(mid)+(0, -1.2)$) {offset $\offset = 3$};

  \coordinate (b4_1o) at ($(b4_1.south)+(-0.1, 0.15)$);
  \coordinate (b4_2o) at ($(b4_2.south)+(-0.1, 0.15)$);
  \coordinate (b4_3o) at ($(b4_3.south)+(-0.1, 0.15)$);

  \draw[dashed] (b4_1o) edge[->, bend right] (offset_label4.west);
  \draw[dashed] (b4_2o) edge[->, bend right] (offset_label4.west);
  \draw[dashed] (b4_3o) edge[->, bend right] (offset_label4.west);

  \coordinate (midpoint) at ($(u4_3.north)!0.5!(u4_4.north)$);

  % \node[anchor=south, scale=0.9] (overflow_label)   at ($(midpoint)+(0,  0.2)$) {Out of bounds};

  \draw[red!20, line width=1.2mm] (u4_6) edge[-] (v4_3.north);
  \draw[red, line width=0.3mm] ($(u4_6.center)+(0, -0.05)$) edge[dotted, -] (v4_3.north);
\end{scope}

\begin{scope}
  % \begin{pgfonlayer}{bg}
  %   \node[fit=(diseq_label)(u1_1)(u1_5)(v1_1)(v1_3)(offset_label1), fill=teal!60, rounded corners=2, draw, dashed] (offset1_wrapper) {};
  % \end{pgfonlayer}
  % \begin{pgfonlayer}{bg}
    \node[fit=(u4_1)(u4_5)(v4_1)(v4_3)] (offset4_wrapper) {};
  % \end{pgfonlayer}

  \node[ ] at ($(u1_1.west)+(-0.5, 0)$) (var_v) {$v{:}$};
  \node[ ] at ($(v1_1.west)+(-0.5, 0)$) (var_u) {$u{:}$};

  % \begin{pgfonlayer}{bg2}
  % \node[fit=(offset1_wrapper)(offset4_wrapper)(var_v)(var_u), fill=orange!40, rounded corners=3, draw, dashed, inner sep=3mm, xshift=1mm] (all_wrapper) {};
  % \end{pgfonlayer}

  % \node[anchor=south, scale=1.4, fill=orange!10, rounded corners=2, draw]
  %   (notcontains_label)   at ($(all_wrapper.north)+(0, -0.3)$) {$\notcontains(u, v)$};
\end{scope}

\begin{pgfonlayer}{bg1}
  \node[fit=(v1_1)(u1_5)] (offset1_wrapper) {};
  \node[fit=(v2_1)(u2_1)(u2_5)] (offset2_wrapper) {};
  \node[fit=(v3_1)(u3_5)(u3_1)] (offset3_wrapper) {};
  \node[fill=gray!10, rounded corners=2] (land1) at ($(offset1_wrapper)!0.5!(offset2_wrapper)$) {$\land$};
  \node[fill=gray!10, rounded corners=2] (land2) at ($(offset2_wrapper)!0.5!(offset3_wrapper)$) {$\land$};
  \node[fill=gray!10, rounded corners=2] (land3) at ($(offset3_wrapper)!0.5!(offset4_wrapper)$) {$\land$};
\end{pgfonlayer}

\end{tikzpicture}

%% file: tables/solved2.tex
\begin{booktabs}{colspec={lQ[r,gray!30]Q[r,gray!30]Q[r,gray!30]Q[r,gray!30]rrrrQ[r,gray!30]Q[r,gray!30]Q[r,gray!30]Q[r,gray!30]rrrrQ[r,gray!30]Q[r,gray!30]Q[r,gray!30]Q[r,gray!30]},rowsep=0pt,row{4}={yellow},cell{1,2}{2,6,10,14,18}={c=4}{c}}
    \toprule
    &  \biopython & & & & \django & & & & \thefuck & & & & \poshard & && & All & & & \\
    &  (77,222)     && & & (52,643)   & & & & (19,872)    & && & (550)    & && & (150,287) & & & \\
    \cmidrule[lr]{2-5}\cmidrule[lr]{6-9}\cmidrule[lr]{10-13}\cmidrule[lr]{14-17}\cmidrule[lr]{18-21}
    & \oob & \unknown & \timet & \timeall & \oob & \unknown & \timet & \timeall& \oob & \unknown & \timet & \timeall & \oob & \unknown & \timet & \timeall & \oob & \unknown & \timet & \timeall\\
    \midrule
    \ziiinoodlerpos & 171    & 0     & 3,490     &    24,010 & 39     & 0      & 3,325      &     8,005 & 0     & 0       & 665        &  \bf   665 & 0     & 0       & 124    & \bf   124 & 210    & 0       & 7,604     & \bf   32,804 \\
    \ziiinoodler    & 171    & 336   & 3,545     &    64,385 & 37     & 108    & 3,473      &    20,873 & 1     & 375     & 637        &  45,757 & 234   & 246     & 1,912   & 59,512 & 443    & 1,065   & 9,567     &   190,527 \\
    \cvcv           & 69     & 0     & 12,834    &  \bf 21,114 & 0      & 0      & 4,515      &   \bf  4,515 & 0     & 0       & 690        &     690 & 550   & 0       & --      & 66,000 & 619    & 0       & 18,039    &    92,319 \\
    \ziii           & 1,047  & 0     & 15,661    &   141,301 & 502    & 0      & 7,501      &    67,741 & 47    & 0       & 9,457      &  15,097 & 550   & 0       & --      & 66,000 & 2,146  & 0       & 32,619    &   290,139 \\
    \ostrich        & 2,986  & 0     & 749,986   & 1,108,306 & 4,404  & 0      & 979,326    & 1,507,806 & 967   & 0       & 120,152    & 236,192 & 550   & 0       & --      & 66,000 & 8,907  & 0       & 1,849,464 & 2,918,304 \\
    \bottomrule
    \end{booktabs}

%% file: appendix.tex
\section{Constructing a Parikh Formula for an NFA}\label{app:parikh}

The formula $\parikhformof{\aut}$---whose models represent (non-uniquely) runs
of the automaton~$\aut = (Q, \Delta, I, F)$---is constructed in the following
way (cf.~\cite{JankuT19}):
\begin{enumerate}
  \item  For every state~$q \in Q$, we introduce two integer variables $\initvarof
    q$ and $\finalvarof q$ that denote
    whether the state is the first and/or the last state in the run, respectively.
    If a state is the first, its $\initvarof {q}$ will have the value~$1$, otherwise
    it will attain the~value $0$. The values of $\finalvarof {q}$ are assigned in the same fashion.
    Only initial states can be first and only final states can be last.
    The conditions are represented by the following two formulae.
    \begin{align}
      \forminitstates \defiff {} &
      \bigg(
        \bigwedge_{q \in I} 0 \leq \initvarof q \leq 1
      \bigg)
      \land
      \bigg(
        \bigwedge_{q \notin I} \initvarof q = 0
      \bigg)
      \land
      \bigg(
        \sum_{q \in I} \initvarof q = 1
      \bigg)
      \\
      \formfinalstates \defiff {} &
      \bigg(
        \bigwedge_{q \in F} 0 \leq \finalvarof q \leq 1
      \bigg)
      \land
      \bigg(
        \bigwedge_{q \notin F} \finalvarof q = 0
      \bigg)
    \end{align}
    Note the last part of~$\forminitstates$, which denotes that there will be
    exactly one first state in the run (a~corresponding condition
    for~$\formfinalstates$ is not required---the condition will be induced later
    by~$\formkirchhoff$).

  \item  For each transition~$t \in \Delta$, we introduce the
    variable~$\transvarof t$, whose value will represent how many times the
    transition is taken in the run.

  \item  Then, for each state~$q \in Q$, we introduce a~\emph{Kirchhoff's laws} kind of
    formula, which ensures that on a~run, the number of times we enter the state
    equals the number of times we leave the state (plus/minus one when the state
    is the first or the last):
    \begin{equation}
      \formkirchhoff(q) \defiff
      \initvarof q +
      \sum_{t = \move \cdot \cdot q \in \Delta} \transvarof t
      =
      \finalvarof q +
      \sum_{t = \move q \cdot \cdot \in \Delta} \transvarof t
    \end{equation}

  \item  The next subformula makes sure that the run represented by the model of
    the formulae is consistent with the structure of~$\aut$, i.e., that
    the states on the run and the transitions are connected and start in
    a~proper state.
    For this, for each state~$q \in Q$, we add a~variable~$\spanningvarof q$,
    which will be assigned the length of the shortest path from the first state
    to~$q$ of the spanning tree of~$\aut$ w.r.t.~the run.
    For states $q$ that do not occur in the run, their $\spanningvarof{q}$ will have assigned
    some negative value, i.e., ~$\spanningvarof q \le -1$.
    \begin{align}
      \formspan(q) \defiff {} &
      \left(
        \spanningvarof q = 0 \liff \initvarof q = 1
      \right)
      \land {} \\
      &
      \Bigg(
        \spanningvarof q \le -1 \limpl \bigg(\initvarof q = 0 \land \bigwedge_{t =
        \move \cdot \cdot q \in \Delta} \transvarof t = 0 \bigg)
      \Bigg)
      \land {}\\
      &
      \bigg(
        \spanningvarof q > 0 \limpl
        \bigvee_{t = \move {q'} \cdot q \in \Delta} (\transvarof t > 0 \land
        \spanningvarof{q'} \geq 0 \land \spanningvarof q = \spanningvarof{q'} + 1)
      \bigg)
    \end{align}
\end{enumerate}

Then the resulting $\parikhformof{\aut}$ is constructed as:
\begin{equation}
  \parikhformof{\aut} \defiff
  \forminitstates
  \land
  \formfinalstates
  \land
  \bigg(
    \bigwedge_{q \in Q}
      \formkirchhoff(q)
      \land
      \formspan(q)
  \bigg)
\end{equation}

\section{Proof of a Single Disequality Being in \clP} \label{app:ptimeProof}

\newcommand{\blank}{\square}
\newcommand{\hasMismatch}{\mathit{HasMismatch}}
\newcommand{\differAtPos}{\mathit{HasMismatch}}

In this section, we prove that $\posregsat(\emptyset, \regconstr, \emptyset,
\diseqconstr)$ where $\diseqconstr \defiff x_1 \dots x_n \neq y_1 \dots y_n$
can be decided in $\clP$. We start by introducing a technical lemma that allows
us to not consider cases when $\diseqconstr$ is satisfied by its sides having
different lengths. Throughout this section we fix $\regconstr$ to be
$\regconstr = \bigwedge_{x \in \vars} x \in \langof{\aut_x}$ where $\aut_x$ is
an arbitrary NFA.

First, we briefly sketch the outline of the proof.
    Let $\psi \defiff x_1 \cdots x_n \neq y_1 \cdots y_m$ be the input
    position predicate. We construct a one-counter automaton~$\caThird$
    (obtained as the last one in a~sequence of constructions $\caFirst$,
    $\caSecond$, $\caThird$) with
    a~counter~$\caCounter$ with updates limited to $\{-1, 0, +1\}$ such that
    there is an accepting state reachable with
    $\caCounter = 0$ iff the input combination of regular constraints
    and~$\psi$ is satisfiable.
    Our construction proceeds in several steps:
    \begin{enumerate}
        \item We create a two-counter increment-only automaton
            $\caFirst$ with counters
            $\caCounterL$ and $\caCounterR$ counting the left-hand
            and right-hand side global mismatch positions respectively.
            $\caFirst$ is created as a~union of
            two-counter automata $\caFirst_{i, j}$ for $1 \le i \le n$
            and $1 \le j \le m$, such that $\caFirst_{i, j}$ has an
            accepting state reachable with $\caCounterL = \caCounterR$ iff $\psi$
            can be satisfied by a~mismatch located in~$x_i$ and~$y_j$.
            The structure of $\caFirst_{i, j}$ consists of $\bigo {|\alphabet|}$ copies
            of $\aut_\circ$, sampling mismatch symbols for the disequality
            side in a nondeterministic fashion. Contrary to the tag automaton presented
            in \cref{sec:single-gen-diseq}, the sampled mismatch symbol is stored
            within the automaton states. The second mismatch symbol can be sampled
            only if it is different than the one seen previously. Finally,
            we note that all counter updates $(\caCounterL, \caCounterR)
            \mathrel{+}= (k_\mathrm{L}, k_\mathrm{R})$ are bounded by
            $k_\mathrm{L} \le n$ and $k_\mathrm{R} \le m$ and the overall size
            of $\caFirst$ is $\bigo {nm \cdot |\alphabet| \cdot |\aut_\circ|}$, i.e.,
            polynomial to the size of the input.
        \item Next, we construct a one-counter automaton $\caSecond$ with a
            counter $\caCounter$ based on $\caFirst$ such that $\caCounter$
            tracks the difference $\caCounterL - \caCounterR$.
            Therefore, for any transition $\move {q} {k_\mathrm{L}, k_\mathrm{R}} {r}$
            of $\caFirst$, we add a~transition $\move {q} {k_\mathrm{L} - k_\mathrm{R}} {r}$
            to $\caSecond$. We observe that $|\caFirst| = |\caSecond|$ and the counter
            updates can be bounded by $|k_\mathrm{L} - k_\mathrm{R}| \le m + n$.
        \item We transform $\caSecond$ into $\caThird$ by replacing any
            transition $\move {q} {k} {r}$ for $k \neq 0$ with $|k|$ transitions $t_1, \dots, t_{|k|}$
            with corresponding new states such that all $t_1, \dots, t_{|k|}$ transitions
            are labeled with either $+1$ (for $k > 0$) or $-1$ (for $k < 0$).
            Hence, the size of $\caThird$ is in $\bigo {(m+n)\big(nm \cdot
            |\alphabet| \cdot |\aut_\circ|\big)}$, i.e., polynomial in the size of the input.
            Detailed construction of $\caThird$ as well as the correspoding
            size analysis can be found in \cref{app:ptimeProof}.
    \end{enumerate}
Finally, $0$-reachability of a state in a~one-counter automaton can be decided
in $\clP$~\cite[Lemma~11]{AlurC11}, which concludes the proof.

\begin{lemma}\label{lemma:equisatLength}
  For any disequality $D$ with regular constraints $\regconstr$
  there is an equisatisfiable $D' \defiff L' \neq R'$
  with regular constraints $\regconstr'$ and $|D'| = \bigo{|D|}$
  such that if $D$ is satisfiable, then $D'$ is satisfiable due to a mismatch,
  i.e., there exists a model $\sigma'$ of $D'$ and a~position $0
  \le i < min(|\sigma'(L')|, |\sigma'(R')|)$ such that $\sigma'(L')[i] \neq
  \sigma'(R')[i]$.
\end{lemma}

\begin{proof}
    Let $D \defiff L \neq R$ be a disequality with
    regular constraints $\regconstr$. Then $D' \defiff L p \neq
    R p$ with $\regconstr' = \regconstr \land p \in \{ \blank
    \}^*$ where $p$ is a fresh variable and $\blank$ is a fresh alphabet symbol
    satisfies the lemma.

    Clearly, if $D$ is unsatisfiable then so is $D'$. Now, let $\sigma$
    be a model of $D$. If there is a position $i$ such that $\sigma(L)[i] \neq \sigma(R)[i]$
    then $\sigma' = \sigma \triangleleft \{ p \mapsto \epsilon \}$ is~a model of $D'$.
    Otherwise, by symmetry, $\sigma(L)$ is a proper prefix of $\sigma(R)$.
    Then $\sigma' = \sigma \triangleleft \{ p \mapsto \blank^{|\sigma(R)| - |\sigma(R)|} \}$
    is a model of $D'$. Moreover, there is a~position $i$ such that
    $\sigma'(L')[i] \neq \sigma'(R')[i]$.

    For the opposite direction it suffices to observe that any conflict in $D'$
    inside $D$'s variables is preserved. Otherwise, there must be a conflict
    between $p$ and some of $D$'s variables. But then one of $D$'s sides
    must be longer than the other one, meaning that $D$ is satisfiable.
\end{proof}

Intuitively, Lemma \ref{lemma:equisatLength} allows us to focus solely on
mismatches when proving that decidability of $\posregsat(\emptyset, \regconstr,
\emptyset, \diseqconstr)$ is in $\clP$. We start with a crucial lemma
describing the construction of a two-counter increment-only automaton with
equal-counter reachability of its accepting states being closely related to the
existence of a mismatch between two particular variables of $\diseqconstr$.
In the following, we fix $|\regconstr|$ to be $|\regconstr| = \sum_{x \in
\vars} |\aut_x|$ where $|\aut_x| = |Q_x| + |\Delta_x|$.

\begin{lemma} \label{lemma:partialTwoC}
    Let $D \defiff x_1 \cdots x_n \neq y_1 \cdots y_m$ be a disequation. Then for
    any choice $(i, j) \in \{1, \dots, n\} \times \{1, \dots, m\}$ there is
    a~two-counter automaton $\caFirst_{i, j}$ of size $|\caFirst_{i, j}| =
    \bigo {|\alphabet| \cdot |\regconstr|}$ with counters $\caCounterL$ and
    $\caCounterR$ such that there is an accepting state $q \in F$ reachable with
    $\caCounterL = \caCounterR$ iff there is a model $\sigma$ of $D$ and a~position
    $k$ satisfying $\sigma(x_1 \cdots x_n)[k] \neq \sigma(y_1 \cdots y_m)[k]$
    with $\sigma(x_1 \cdots x_n)[k]$ being a letter of $x_i$ and $\sigma(y_1 \cdots y_m)[k]$
    being a letter of $y_j$.
\end{lemma}

\newcommand{\varin}{\mathcal{V}_{\le}}     % #vars before, inclusive
\newcommand{\varex}{\mathcal{V}_{<}}       % #vars before, exclusive
\newcommand{\cntmoveeps}[2]{#1 \xrightarrow[c_R \gets c_R]{c_L \gets c_L} #2}  % two-counter move via epsilon

\newcommand{\cmove}[3]{#1 \xrightarrow{#2} #3}  % two-counter move

\newcommand{\cntmove}[4]{#1 \xrightarrow[c_R \gets c_R + #3]{c_L \gets c_L + #2} #4}  % two-counter move
\newcommand{\cntinc}[2]{#1 \gets #1 + #2}  % counter increment
\newcommand{\stateHom}[0]{\kappa}

\begin{proof}
    Let $\autcon = (Q_\epsilon, \Delta_\epsilon, I_\epsilon, Q_\epsilon )$ be
    an $\epsilon$-concatenation of all $\aut_x = (Q_x, \Delta_x, I_x F_x)$
    respecting the fixed order $\preceq$, and let $\var\colon Q_\epsilon
    \rightarrow \vars$ be a function such that $\var(q) = x$ iff $q$ is a state
    of $\aut_x$. Furthermore, Let $\mathcal{V}^{S}_{\sim} \colon (\vars \times
    \nat) \rightarrow \nat$ for $S \in \{ \leftsymb, \rightsymb \}$ and ${\sim}
    \in \{<, \le\}$ be functions defined as follows.
    \begin{align*}
        \varex^L(x, l) &= \left|\left\{k \in \nat: 0 \le k < l \land y_k = x \right\}\right|
        &&\varex^R(x, l) = \left|\left\{k \in \nat: 0 \le k < l \land z_k = x \right\}\right| \\
        \varin^L(x, l) &= \left|\left\{k \in \nat: 0 \le k \le l \land y_k = x \right\}\right|
        &&\varin^R(x, l) = \left|\left\{k \in \nat: 0 \le k \le l \land z_k = x \right\}\right| \\
    \end{align*}

    We construct a counter automaton $\caFirst_{i, j} = (Q, \Delta, I, F)$
    with counters $\caCounterL$ and $\caCounterR$
    where $Q = Q_\epsilon \times \alphabet \times \{\leftsymb, \rightsymb\} \cup Q_\epsilon
    \times \{ \bot, \top \} \big)$, $I = Q_\epsilon \times \{ \bot \}$, $F =
    Q_\epsilon \times \{ \top \}$. For transitions, we write $\cmove q {(k_L, k_R)} r$
    to denote a transition from $q$ to $r$ that updates $c_L$ by $c_L \gets c_L + k_L$
    and $c_R$ by $c_R \gets c_R + k_R$. The set of transitions $\Delta$ is then
    $$\Delta = \bigg( \bigcup_{x \in \vars \setminus \{x_i, y_j\}} \Delta_{\mathit{NoMis}}(x) \bigg) \cup
       \Delta_{\mathit{Mis}}(x_i) \cup \Delta_{\mathit{Mis}}(y_j) \cup \Delta_{1 \rightarrow 2} \cup \Delta_{2 \rightarrow 3} \cup \Delta_\epsilon $$
    where the constituting sets are the following:
    \begin{enumerate}
        \item $\Delta_\mathit{NoMis}(x)$ are transitions in a variable $x$ that does not contain a conflict:
            \begin{align*}
                \Delta_{\mathit{NoMis}}(x) =
                  &\big\{ \cnmove {(q, \bot)}
                                    {\varex^{L}(x,i), \varex^{R}(x,i)}
                                  {(r, \bot) }
                           \mid \move {q} {b} {r} \in \Delta_x )
                   \big\} \cup{}\\
                  &\big\{ \cnmove {(q, a, S)}
                                    {\varex^{L}(x,i), \varex^{R}(x,i)}
                                  {(r, a, S)}
                            \mid \move {q} {b} {r} \in \Delta_x \land a \in \alphabet \land S \in \{\mathcal{L}, \mathcal{R}\}
                    \big\} \cup{} \\
                  &\big\{ \cnmove {(q, \top)}
                                    {\varex^{L}(x,i), \varex^{R}(x,i)}
                                    {(r, \top)}
                           \mid \move {q} {b} {s} \in \Delta_x
                   \big\},
            \end{align*}
        \item $\Delta_\mathit{Mis}(x)$ are transitions corresponding to the variable $x$ containing a conflict:
            \begin{align*}
              \Delta_{\mathit{Mis}}(x) =
                  &\big\{ \cnmove {q}
                                    {\varin^{L}(x, i), \varin^{R}(x, j)}
                                  {r}
                            \mid \move {q} {b} {r} \in \Delta_x
                    \big\} \cup{} \\
                  &\big\{ \cnmove {(q, a, \mathcal{L})}
                                    {\varex^{L}(v, i), \varin^{R}(v, j)}
                                  {(r, a, \mathcal{L})}
                          \mid \move {q} {b} {r} \in \Delta_x \land
                               a \in \alphabet
                   \big\} \cup{} \\
                  &\big\{ \cnmove {(q, a, \mathcal{R})}
                                   {\varin^{L}(v, i), \varex^{R}(v, j)}
                                    {(r, a, \mathcal{R})}
                          \mid \move {q} {b} {r} \in \Delta_x \land
                               a \in \alphabet
                   \big\} \cup{} \\
                  &\big\{ \cnmove {(q, \top)}
                                    {\varex^{L}(v, i), \varex^{R}(v, j)}
                                    {(r, \top)}
                           \mid \move {q} {b} {r} \in \Delta_x \big\},
            \end{align*}
        \item then we have $\Delta_{1 \rightarrow 2}$ containing transitions that sample a mismatch symbol and store it in the state target state:
        \begin{equation*}
          \Delta_{1 \rightarrow 2} =
          \begin{cases}
               \big\{ \cnmove {s} {\varin^{L}(x_i, i), \varin^{R}(x_i, j)} {(q, a, \mathcal{L})}
                   \mid \move {s} {a} {q} \in \Delta_{x_i} \big\} &\mbox{if } x_i \prec y_j, \\[2mm]
               \big\{ \cnmove {s} {\varin^{L}(y_j, i), \varin^{R}(y_j, j)} {(q, a, \mathcal{R})}
                       \mid \move {s} {a} {q} \in \Delta_{y_j} \big\} &\mbox{if } x_i \succ y_j, \\[2mm]
               \big\{ \cnmove {s} {\varin^{L}(x_i, i), \varin^{R}(x_i, j)} {(q, a, S)}
                       \mid \move {s} {a} {q} \in \Delta_{x_i} \land S \in \{\mathcal{L}, \mathcal{R}\} \big\} &\mbox{otherwise.} \\
          \end{cases}
        \end{equation*}
        \item next we have $\Delta_{2 \rightarrow 3}$ containing transitions that sample the second mismatch symbol if it is different from the one seen previously:
            \begin{equation*}
              \Delta_{2 \rightarrow 3} =
              \begin{cases}
                 \big\{ \cnmove {(s, a_1, \mathcal{L})} {\varex^{L}(y_j, i),\varin^{R}(y_j, j)} {(q, \top)}
                 \mid \move {s} {a_2} {q} \in \Delta_{y_j} \land a_1 \not = a_2 \big\} &\mbox{if } x_i \prec y_j, \\[2mm]
                 \big\{ \cnmove {(s, a_1, \mathcal{R})} {\varin^{L}(x_i, i), \varex^{R}(x_i, j)} {(q, \top)}
             \mid \move {s} {a_2} {q} \in \Delta_{x_i} \land a_1 \not = a_2 \big\} &\mbox{if } x_i \succ y_j, \\[2mm]
                 \begin{aligned}
                    & \big\{ \cnmove {(s, a_1, \mathcal{R})} {\varin^{L}(x_i, i), \varex^{R}(x_i, j)} {(q, \top)}
                        \mid \move {s} {a_2} {q} \in \Delta_{x_i} \land a_1 \not = a_2 \big\} \cup{} \\
                    & \big\{ \cnmove {(s, a_1, \mathcal{L})} {\varex^{L}(x_i, i), \varin^{R}(x_i, j)} {(q, \top)}
                        \mid \move {s} {a_2} {q} \in \Delta_{x_i} \land a_1 \not = a_2 \big\} \\
                 \end{aligned} &\mbox{otherwise}
              \end{cases}
            \end{equation*}
        \item Finally, we have $\Delta_\epsilon$ containing $\epsilon$-transitions copied from $\aut_\epsilon$:
        \begin{align*}
          \Delta_{\epsilon} =
              &\big\{ \cnmove {s} {0, 0} {q}
                        \mid \move {s} {\epsilon} {q} \in \Delta_\epsilon
               \big\} \cup{} \\
              &\big\{ \cnmove {(s, a, S)} {0, 0} {(q, a, S)}
                      \mid \move {s} {\epsilon} {q} \in \Delta_\epsilon \land
                           a \in \alphabet \land
                           S \in \{\mathcal{L}, \mathcal{R}\}
               \big\} \cup{} \\
              &\big\{ \cnmove {(s, \top)} {0, 0} {(q, \top)}
                      \mid \move {s} {\epsilon} {q} \in \Delta_\epsilon
               \big\}
        \end{align*}
    \end{enumerate}

\textbf{Size analysis.}
The total number of states is $|Q| = 2 |\alphabet|
|Q_\epsilon| + 2 |Q_\epsilon| = \bigo{|Q_\epsilon| |\alphabet| }$. As for the
transitions, we have
\begin{align*}
|\Delta| &=
   \underbrace{\bigg( \sum_{x \in \vars \setminus \{x_i, y_j\}} |\Delta_{\mathit{NoMis}}(x)| \bigg) +
   | \Delta_{\mathit{Mis}}(x_i)| +
   | \Delta_{\mathit{Mis}}(y_j)| +
   | \Delta_\epsilon| }_{{} \le 2|\Delta_\epsilon|(|\alphabet|+2)} {}+
   \underbrace{\Delta_{1 \rightarrow 2} +
   \Delta_{2 \rightarrow 3}}_{{}\le 2|\Delta_\epsilon||\alphabet|} \\
     &= \bigo{|\Delta_\epsilon||\alphabet|}.
\end{align*}
As $|\aut_\epsilon| = |\regconstr|$, we have $|\caFirst_{i, j}| = \bigo{ |\alphabet| \cdot |\regconstr|}$.
Finally, observe that, for a given choice $(i, j)$, $\caCounterL$ and $\caCounterR$ updates
are bounded by $i$ and $j$, respectively.

\textbf{Correctness.} We claim that $\caFirst_{i, j}$ has an accepting state
reachable with $\caCounterL = \caCounterR$ iff the input formula $x_1
\cdots x_n \neq y_1 \cdots y_m$ is satisfiable with a model $\sigma$ and a
position $k$ such that $\sigma(x_1 \cdots x_n)[i] \neq \sigma(y_1 \cdots
y_m)[i]$ with $\sigma(x_1 \cdots x_n)[i]$ and $\sigma(y_1 \cdots y_m)[i]$ being
a letter of $x_i$ and $y_j$, respectively. We assume that $x_i \prec y_j$ and
note that the remaining cases can be proven in the same fashion. First, we observe
that, by construction of $\caFirst_{i,j}$, any run $\rho$ of $\caFirst_{i, j}$ maps to (possibly multiple) assignments
$\sigma\colon \vars \rightarrow \alphabet^*$ such that $\sigma(x) \in \langof{\aut_x}$
for every $x$.

For the forward direction, it suffices to show that the mismatch is located in
$x_i$ and $y_j$ as the construction of $\caFirst_{i, j}$ guarantees the sampled
letters to be different. Let $\rho = q_0 \cnmove{} {l_1, r_1} {q_1} \cdots
\cnmove {q_{n-1}} {l_n, r_n} {q_{n}}$ be an accepting run, such that
$\caCounterL = \sum_{1 \le k \le n} l_k = \sum_{1 \le k \le n} l_k =
\caCounterR$. From the construction of $\caFirst_{i, j}$, we have that any transition
$\cnmove {q_{k-1}} {l_k, r_k} {q_k}$ that such that $\var(q_{k-1}) = \var(q_{k}) = x$
for some variable $x \neq x_i \neq y_j$, has $l_k$ and $r_k$ equal to the number of
variables that precedes $x_i$ and $y_j$, respectively. Let us focus on the sum $L$ of $\caCounterL$
updates along transitions $\cnmove
{q_{o-1}} {l_{o}, r_{o}} {q_{o}}$ with $\var(q_{o-1}) =
\var(q_o) = x_i$. From the construction of $\caFirst_{i, j}$, $L$
looks in the following way:
$$ L = \varin^L(x_i, i) + \cdots + \varin^L(x_i, i) + \varex^L(x_i, i) + \cdots \varex^L(x_i, i). $$
Observing that $\varin^L(x_i, i) = \varex^L(x_i, i) + 1$, we can rearrange $L$ into
\begin{equation*}
    L = P \cdot \varex^L(x_i, i) + N
\end{equation*}
where $P$ is the total number of summands, and $N < P$ is the number of $\varex^L(x_i, i)$
summands. Noting that $P$ is the same as the length of the word assigned to $x_i$ by $\rho$,
we conclude that $N$ is the mismatch position inside the variable $x_i$. The same chain
of thought can be done to conclude that the second mismatch is within $y_j$.

The opposite direction is done in similar fashion, but reversed. Given a model
$\sigma$ of $D$ with a~mismatch on position $k$ within $x_i$ and $y_j$, one
constructs a run $\rho$ of $\caFirst_{i, j}$ such that a word $w_x$ assigned to
any variable $x$ by $\rho$ satisfies $|w_x| = |\sigma(x)|$. Let $K = (\sum_{1
\le l < i} |\sigma(x_l)|)$. The constructed run $\rho$ takes the $|k - K|$-th
transition from $\Delta_\mathit{Mis}(x_i)$ along the symbol $\sigma(x_1 \cdots
x_n)[k]$. Taking a~transition from $\Delta_{\mathit{Mis}}(y_j)$ similarly and
finishing the run so that all assigned words are length-consistent, we are left
with a complete run $\rho$. It suffices to show that $\caCounterL =
\caCounterR$, which can can be done taking the mismatch position $k$, observing
how the words assigned to variables preceding $x_i$ ($y_j$) contribute to $k$
and then concluding that the run makes precisely the same contributions to
$\caCounterL$ ($\caCounterR$).
\end{proof}

Knowing how to construct $\caFirst_{i, j} = (Q_{i, j}, \Delta_{i, j}, I_{i, j},
F_{i, j})$, we can construct $\caFirst = (Q^{1}, \Delta^{1}, I^{1}, F^{1})$
where
$$
Q^{1} = \biguplus_{\substack {1 \le i \le n \\ 1 \le j \le m}} Q_{i, j},\qquad
\Delta^{1} = \biguplus_{\substack {1 \le i \le n \\ 1 \le j \le m}} \Delta_{i, j},\qquad
I^{1} = \biguplus_{\substack {1 \le i \le n \\ 1 \le j \le m}} I_{i, j}, \qquad
F^{1} = \biguplus_{\substack {1 \le i \le n \\ 1 \le j \le m}} F_{i, j}.
$$
By construction, $\caFirst$ has an accepting state reachable with $\caCounterL = \caCounterR$
iff the input disequality $D$ is satisfiable. As for the size, $|Q^{1}| =
\bigo{nm |Q_\epsilon||\alphabet|}$, $|\Delta^{q}| = \bigo{nm
|\Delta_\epsilon||\alphabet|}$, i.e., $|\caFirst| = \bigo{nm \cdot |\alphabet|
\cdot |\regconstr|}$. Every counter update $\caCounterL \gets \caCounterL + k_1$
and $\caCounterL \gets \caCounterL + k_2$ can be bounded by $n$ and $m$,
respectively.

Next, we create a single-counter automaton $\caSecond = (Q^{2}, \Delta^{2}, I^{2}, F^{2})$
with a counter $\caCounter$
where $Q^{2} = Q^{1}$, $I^{2} = I^{1}$, $F^{2} = F^{1}$ and $\Delta = \{ \cnmove q {k_1 - k_2 } r \mid \cnmove {q} {k_1, k_2 } {r} \in \Delta^{1} \}$.
Clearly, $\caSecond$ has the property that an~accepting state is reachable with $\caCounter = 0$ iff
the input disequality $D$ is satisfiable. Furthermore, $|Q^{2}| = |Q^{1}|$, $|\Delta^{2}| = |\Delta^{1}|$,
and every update $\caCounter \gets \caCounter + k$ can be bounded by $|k| \le (m + n)$.

Finally, we can create a single-counter automaton $\caThird$ based on
$\caSecond$ with counter $\caCounter'$ such that any counter updates are
limited to $\caCounter' \gets \caCounter' + k$ where $k \in \{-1, 0, 1\}$ by
replacing any transition $\cnmove {q} {l} {r} \in \Delta^2$ with $|l| > 1$ by a
series $t_1, \dots, t_{|l|}$ of intermediate transitions through fresh states
such that every $t_o$ for $1 \le o \le |l|$ updates $\caCounter'$ by $+1$ (if $l > 0$) or $-1$ (if $l
< 0$). As the absolute value of any counter update in $\caSecond$ is
bounded by $m+n$, such a transformation produces automaton with $|Q^{3}| \le
|Q^2| + (m+n)|\Delta^{2}|$ states. That is $|Q^{3}| =
\bigo{nm|\alphabet|(|Q_\epsilon| + (m+n)|\Delta_\epsilon|)} = \bigo{nm (m+n) |\alphabet| |\aut_\epsilon| }$.
In other words, $\caThird$ is of polynomial size in the size of the input. Since
0-reachability in a single-counter automaton can be decided in
$\clP$~\cite{AlurC11,melskireps1997}, satisfiablity of $D$ can be also
established in $\clP$.

\section{Remaining Definitions for Reducing Multiple Disequations to LIA}
\label{app:liareduction}

To complete the reduction of a system of disequations $\varphi \defiff
\bigwedge_{1 \le i \le n} D_{i}$ to an equisatisfiable LIA formula $\psi$ using
a tag automaton $T$ we are missing the definitions of $\fLen {k}$ and
$\fMismatch {k}$ where $1 \le k \le n$. Let the $k$-th disequality be
$D_k \;\defiff\; x_{1} \dots x_{N} \neq y_{1} \dots y_{M}$.

The formula $\fLen {k}$ expressing that $D_k$ is satisfied due to its sides
having different lengths can be expressed as follows:
\begin{equation}
	\fLen {k} \defiff
		\bigg( \sum_{0 \le i \le N} \numof {\taglenof{x_i}} \bigg) \neq
		\bigg( \sum_{0 \le j \le M} \numof {\taglenof{y_j}} \bigg).
\end{equation}

\newcommand{\fMisPos}[1]{\varphi^{#1}_{\mathit{Pos}}}
\newcommand{\fAlign}[1]{\varphi^{#1}_{\mathit{Align}}}
\newcommand{\fConflictPresent}[1]{\varphi^{#1}_{\mathit{\exists}}}

Next, we define $\fMismatch {k}$ that holds iff the run samples two distinct
mismatch symbols for $D_k$ that are located at the same (global) position.
\begin{align}
	\fMisPos {k}(s, v) \defiff
	 \bigwedge_{\substack{ 2 \le l \le 2n}}
        \Bigg(
			\Big(
				& \sum_{ \substack{ \texttt{a} \in \alphabet} }
                    \numof{\tagof{ \tagmis_l, v, D_k, s, \texttt{a}} } \Big)+
                    \numof{\tagof{ \tagcop_l, v, D_k, s }}
			 > 0
			\rightarrow
            p_{D_k, s} = \sum_{0 < k \le l} \numof{\tagof{\tagpos_k, v}}
        \Bigg) \land {} \notag \\
    & \sum_{ \substack{ \texttt{a} \in \alphabet}} \numof{\tagof{ \tagmis_1, v, D_k, s, \texttt{a}} } > 0 \rightarrow p_{D_k, s} = \numof{ \tagof{\tagpos_1, v}}
\end{align}
\begin{equation}
	\fAlign {k} (i, j) \defiff
		\bigg( \sum_{1 \le l < i} \numof{\taglenof{x_l}} + p_{D_k, \leftsymb} \bigg) =
		\bigg( \sum_{1 \le l < j} \numof{\taglenof{y_l}} + p_{D_k, \rightsymb} \bigg)
\end{equation}
\begin{equation}
	\fConflictPresent {k} (s, v) \defiff
		\bigg(
            \Big(
			\sum_{ \substack{ \texttt{a} \in \alphabet }}
                \numof{\tagof{ \tagmis_1, v, D_k, s, \texttt{a} }}
            \Big) +
            \Big(
			\sum_{ \substack{ 2 \le l \le 2n \\ \texttt{a} \in \alphabet }}
                \numof{\tagof{ \tagmis_l, v, D_k, s, \texttt{a} }} +
                \numof{\tagof{ \tagcop_l, v, D_k, s }}
            \Big)
		\bigg) > 0
\end{equation}
\begin{equation}
	\fMismatch {k} \defiff
		\bigvee_{\substack{1 \le i \le N \\ 1 \le j \le M}}
		\bigg(
			\fMisPos {k} (\leftsymb, x_i) \land
			\fMisPos {k} (\rightsymb, y_j) \land
			\fAlign {k} \land
			\fConflictPresent {k}(\leftsymb, x_i) \land
			\fConflictPresent {k}(\rightsymb, y_j) \land
			r_{D_k, \leftsymb} \neq r_{D_k, \rightsymb}
		\bigg)
\end{equation}
Intuitively, the $\fMisPos {k}(s, v)$ formula asserts that value of $p_{D_k, s}$
correctly counts the mismatch position inside the variable $v$ wrt. the order
in which the mismatch sampling transitions were taken. The $\fAlign {k}$ formula
ensures that the global mismatch position is the same for both sides of $D_k$.
Next, $\fConflictPresent {k} (s, v)$ holds if there was a mismatch sampled
for the side $s$ of $D_k$ in variable $v$. Finally, we have all necessary
formulae to define the resulting equisatisfiable formula $\varphi$ for $\psi$:
\begin{equation}
    \varphi \defiff \fFair \wedge \fConsist \wedge \fCopies \land \bigwedge_{1 \le i \le n} \big( \fLen {i} \lor \fMismatch {i} \big)
\end{equation}

\label{appendix:ca_sizes}

\vspace{-3mm}
\section{Proof of \clNP-hardness of the $\notcontains$ predicate} \label{app:npHardess}

\newcommand{\clause}[0]{\mathcal{C}}
\newcommand{\sep}[0]{\mathbf{\triangle}}
\newcommand{\enc}[0]{\mathrm{Enc}}
\newcommand{\placeHolder}[0]{\textcolor{black}{w}}
\newcommand{\placeHolderC}[0]{\textcolor{blue}{w}}

We show that the problem of deciding the satisiability of a single $\notcontains$ predicate is $\clNPhard$
by reduction from the $3$-SAT problem.
Let $X = \{x_1, \dots, x_N\}$ be a set of propositional variables, and let $\varphi \defequiv
\bigwedge_{1 \le i \le n} \clause_i$ be a finite conjunction of clauses where each
clause $\clause_i$ is a disjunction of three literals, i.e., $\clause_i \defequiv
l_{i, 1} \lor l_{i, 2} \lor l_{i, 3}$ with $l_{i, j} = x$ or $l_{i, j} = \neg x$ for some $x \in X$ and $1 \le j \le 3$.
For every propositional variable $x \in X$ we introduce two string variables $s_x$ and $\overline{s_x}$ with $\langof {s_x} =
\langof {\overline{s_x}} = \{0, 1\}$,
corresponding to $x$ and its negation $\neg x$. Let $\vars$ be the set of all such variables.

We encode clauses in a straightforward manner as a concatenation of
corresponding string variables. More formally, we encode any clause $\clause_i
= l_{1} \lor l_2 \lor l_3$ as a word $\enc(\clause_i) = s_1 s_2 s_3$ where $s_j
\in \vars$ is a string variable corresponding to the literal $l_{j}$ for any $1
\le j \le 3$. Let $\sep$ be a fresh alphabet symbol serving the role of a separator.
Given a 3-SAT formula $\varphi$, we construct a formula $\psi$ as in
\cref{eq:3sat}, using different colors to distinguish between different roles
parts of the created strings play in the reduction
\vspace{-2mm}
\begin{gather}
    w_{x, \overline{x}} \defequal
        \textcolor{blue}{000} \textcolor{red}{00s_x\overline{s_x}} \sep
        \textcolor{blue}{000} \textcolor{red}{s_x\overline{s_x}11}
    \\
    \psi \defequiv
        \notcontains(
            \underbrace{\textcolor{blue}{000} \textcolor{red}{0011}}_{U},
            \underbrace{
            \textcolor{blue}{\enc(\clause_1)} \textcolor{red}{0011} \sep \cdots
            \textcolor{blue}{\enc(\clause_n)} \textcolor{red}{0011}}_{W_{\mathrm{SAT}}} \sep
            \underbrace{
            w_{x_1, \overline{x_1}} \sep
            \cdots
            w_{x_N, \overline{x_N}}
            }_{W_{\mathrm{x \neq \overline{x}}}}
            )  \label{eq:3sat}
\end{gather}

To see that $\varphi$ and $\psi$ are equisatisfiable, we divine the right-hand side of
$\notcontains$ in $\psi$ into two words, $W_{\mathrm{SAT}}$ and $W_{\mathrm{x \neq \overline{x}}}$.
Aligning $U$ between two occurrences of the separator $\sep$ in $W_{\mathrm{SAT}}$, we see
that the \textcolor{red}{red} parts agree, and, therefore, in order for $U$ to not be contained
withing $W_{\mathrm{SAT}}$ we must find an assignment $\sigma$ such that none of the \textcolor{blue}{blue}
subwords in $W_{\mathrm{SAT}}$ are \textcolor{blue}{000}, i.e., we enforce that every clause
has at least one literal satisfied.

Second, we have the word $W_{\mathrm{x \neq \overline{x}}}$ consisting of a concatenation of
$w_{x, \overline{x}}$ for every propositional variable $x \in X$ separated by $\sep$.
The \textcolor{blue}{blue} subwords of any $w_{x, \overline{x}}$ agree with the
\textcolor{blue}{blue} subword of $U$, and, therefore, any model $\sigma$ must assign values
to $s_x$ and $\overline{s_x}$ such that $\textcolor{red}{0011} \neq \textcolor{red}{00s_x \overline{s_x}}$
and $\textcolor{red}{0011} \neq \textcolor{red}{s_x \overline{s_x} 11}$. Therefore, for $\sigma$ to be
a model of $\psi$, it must hold that $\sigma(s_x) \neq \sigma(\overline{s_x})$ for any pair of string variables
$s_x$ and $\overline{s_x}$ corresponding to the propositional variable $x$. Therefore, the word
$W_{x \neq \overline{x}}$ enforces that exactly one of $x$ and $\neg x$ holds. We obtain
the following theorem.

\begin{theorem}
    The satisfiability of a single $\notcontains$ predicate is \clNPhard.
\end{theorem}